\documentclass[aps,pra,reprint,10pt,superscriptaddress,showpacs,amssymb,amsfonts]{revtex4-1}
\usepackage{amsmath,amssymb,amsthm,amscd,latexsym,epsfig,bm,subfigure}

\usepackage[usenames,dvipsnames]{color}
\usepackage[colorlinks,bookmarks=false,citecolor=blue,linkcolor=red,urlcolor=blue]{hyperref}
\usepackage{verbatim}

\def\beq{\begin{equation}}
\def\eeq{\end{equation}}
\def\be{\begin{equation}}
\def\ee{\end{equation}}

\newcommand{\zz}{\mathbb{Z}_2}
\newcommand{\z}{\mathbb{Z}}

\theoremstyle{plain}
\theoremstyle{plain}
\newtheorem{thmtext}{\protect\theoremtextname}
\providecommand{\theoremname}{Theorem}
\providecommand{\theoremtextname}{Theorem}

\newtheorem{prop}{\protect\propositionname}
\theoremstyle{plain}
\providecommand{\propositionname}{Proposition}

\begin{document}

\title{Flux-fusion anomaly test and bosonic topological crystalline insulators}

\author{Michael Hermele}
\affiliation{Department of Physics, University of Colorado, Boulder, Colorado 80309, USA}
\affiliation{Center for Theory of Quantum Matter, University of Colorado, Boulder, Colorado 80309, USA}
\author{Xie Chen}
\affiliation{Department of Physics and Institute for Quantum Information and Matter, California Institute of Technology, Pasadena, CA 91125, USA}
\date{\today}

\begin{abstract}
We introduce a method, dubbed the flux-fusion anomaly test, to detect certain anomalous symmetry fractionalization patterns in two-dimensional symmetry enriched topological (SET) phases.  We focus on bosonic systems with $\zz$ topological order, and symmetry group of the form $G = {\rm U}(1) \rtimes G'$, where $G'$ is an arbitrary group that may include spatial symmetries and/or time reversal.  The anomalous fractionalization patterns we identify cannot occur in strictly $d=2$ systems, but can occur at surfaces of $d=3$ symmetry protected topological (SPT) phases.  This observation leads to examples of $d=3$ bosonic topological crystalline insulators (TCIs) that, to our knowledge, have not previously been identified.  In some cases, these $d=3$ bosonic TCIs can have an anomalous superfluid at the surface, which is characterized by non-trivial projective transformations of the superfluid vortices under symmetry.  The basic idea of our anomaly test is to introduce fluxes of the ${\rm U}(1)$ symmetry, and to show that some fractionalization patterns cannot be extended to a consistent action of $G'$ symmetry on the fluxes.  For some anomalies, this can be described in terms of dimensional reduction to $d=1$ SPT phases.  We apply our method to several different symmetry groups with non-trivial anomalies, including $G = {\rm U}(1) \times \zz^T$ and $G = {\rm U}(1) \times \zz^P$, where $\zz^T$ and $\zz^P$ are time-reversal and $d=2$ reflection symmetry, respectively.
\end{abstract}


\maketitle

\section{Introduction}
\label{sec:intro}

Following the theoretical prediction \cite{kane05a,kane05b,bernevig06,moore07,fu07,roy09} and experimental discovery \cite{konig07,hsieh08} of time-reversal invariant topological band insulators, 
it has become clear that symmetry plays a rich and varied role in topological phases of matter.
New families of symmetric topological phases have been identified theoretically, and significant strides have been made in the classification and characterization of such phases.  Much of the recent progress, with some important exceptions, has focused on systems with internal (or, on-site) symmetry, such as time reversal, ${\rm U}(1)$ charge symmetry and ${\rm SO}(3)$ spin symmetry.  For example, free-fermion topological insulators and superconductors with internal symmetry have been fully classified \cite{kitaev09,ryu10}.  Subsequent work identified the symmetry protected topological (SPT) phases, some of which are strongly interacting generalizations of topological insulators that do not admit a free-electron description \cite{pollmann10,fidkowski11,turner11,chen11a,schuch11,chen13}.

Less attention has been paid to the role of crystalline space group symmetry in topological phases, especially in the setting of strongly interacting systems.  Of course, such symmetry is common and varied in real solids, in contrast to a relatively small number of realistic internal symmetries.  Therefore, with an eye toward eventual experimental realizations of new topological phases, it is important to develop theories of such phases with crystalline symmetry \cite{wen02,lfu11,essin13,slager13,ymlu14,chsieh14,you14,yqi15a,zaletel15,cho15,isobe15,yqi15b,ando15topological}.   
To accomplish this task, new theoretical approaches are needed, as some of the existing tools to classify and characterize topological phases are limited to internal symmetry.

In this paper, we consider two-dimensional ($d=2$) topologically ordered systems, where crystalline and other symmetries play a non-trivial role via their action on anyon quasiparticle excitations \cite{wen02,kitaev06,essin13,mesaros13,barkeshli14,lukaszunpub}.  Such systems are said to be in symmetry-enriched topological (SET) phases. We introduce a method, the \emph{flux-fusion anomaly test}, which allows us to show that some putative SET phases cannot exist in strictly two dimensions.  However, such states can exist as surfaces of $d=3$ SPT phases.  Our method allows us to identify new examples of $d=3$ SPT phases dubbed bosonic topological crystalline insulators (TCIs), which are outside the scope of existing theoretical approaches, via their surface SET phases.  Bosonic TCIs in $d=3$, named after electronic TCIs \cite{ando15topological}, are SPT phases where the protecting symmetry includes both ${\rm U}(1)$ and the space group symmetry of a clean $d=2$ surface.  These states are interesting not only in the context of spin or boson systems, but as a possible stepping stone toward understanding electronic TCIs with strong interactions, and we hope our results can spur more progress in this direction.

The jumping off point for our approach is a consideration of \emph{symmetry fractionalization patterns} in $d=2$.  Provided we assume symmetry does not permute anyon species, the action of symmetry fractionalizes into an action on individual anyons, hence the term symmetry fractionalization.  The classic example is the fractional charge of Laughlin quasiparticles in fractional quantum Hall liquids \cite{laughlin83}.  We refer to a complete description of symmetry fractionalization in a topologically ordered system as a symmetry fractionalization pattern.  Distinct patterns of symmetry fractionalization -- including for crystalline symmetry -- have been classified \cite{wen02,essin13,mesaros13}, and the symmetry fractionalization pattern is a universal property of a SET phase \cite{essin13}.

\begin{table*}
\label{tab:summary}
\begin{tabular}{c|c|c|c|c}
\parbox[t]{3.2cm}{Symmetry $G = {\rm U}(1) \rtimes G'$} & \parbox[t]{3.5cm}{Vison fractionalization classes [$H^2(G', \zz)$]} &
\parbox[t]{4.0cm}{Anomaly-negative vison fractionalization classes (${\cal N}$)}  &
 \parbox[t]{4.5cm}{$d=3$ SPT phases distinguished by anomaly test (${\cal S}$)} & \parbox[t]{1.5cm}{Anomaly \\ type}
\\
\hline
${\rm U}(1) \times \zz^T$ & $\zz$ & $\z_1$ & $\zz$ & 1 \\
\hline
\parbox[t]{3.2cm}{ ${\rm U}(1) \times \zz^P$ \\ (reflection) } & $\zz$ & $\z_1$ & $\zz$ & 2 \\
\hline
\parbox[t]{2.8cm}{${\rm U}(1) \times pm$ \\ (translation \& parallel reflection)} & $(\zz)^4$ & $(\zz)^2$ & $(\zz)^2$  & 2 \\
\hline
\parbox[t]{2.8cm}{$({\rm U}(1) \rtimes \zz^T) \times p1$ \\ (translation only)} & $(\zz)^4$ & $(\zz)^2$ & $(\zz)^2$ & 3 \\
\hline
\parbox[t]{2.8cm}{$({\rm U}(1) \rtimes \zz^T) \times pm$} & $(\zz)^8$ & $(\zz)^3$ & $(\zz)^5$ & 2,3 \\
\hline
\parbox[t]{2.8cm}{${\rm U}(1) \times p4mm$ \\ (square lattice)} & $(\zz)^6$ & $(\zz)^3$ & $(\zz)^3$ & 2 \\
\hline
\parbox[t]{2.8cm}{$({\rm U}(1) \rtimes \zz^T) \times p4mm$} & $(\zz)^{10}$ & $(\zz)^4$ & $(\zz)^6$ & 2,3
\end{tabular}
\caption{Summary of results. Each row is a distinct symmetry group, given in the first column.  The last column indicates the type or types of anomalies that appear, as described in the text.  The meaning of the other columns is discussed in the text.  $\z_1$ denotes the trivial group. In all these cases we consider $Z_2$ gauge theory whose gauge charge $e$ carries half ${\rm U}(1)$ charge while the gauge flux $m$ carries zero charge.}
\end{table*}

A symmetry fractionalization pattern may be anomalous, which means that it cannot occur in a strictly $d=2$ system, but is instead realized at the surface of a $d=3$ symmetry-protected topological (SPT) phase \cite{vishwanath13,metlitski13,cwang13,chen14}.  In this case, we say we have a surface SET phase.  SPT phases \cite{pollmann10,fidkowski11,turner11,chen11a,schuch11,chen13} have an energy gap, lack spontaneous symmetry breaking, and, upon weakly breaking whatever symmetries are present, are in the trivial phase; that is, the ground state wave function can be adiabatically continued to a product state when symmetry is explicitly broken. It follows that SPT phases lack bulk excitations with non-trivial braiding statistics.  Instead, edge or surface properties are generally non-trivial; for $d=3$ SPT phases, one possibility is to have a surface SET phase with anomalous symmetry fractionalization.

While a number of results have been obtained on anomalous symmetry fractionalization of internal symmetry \cite{vishwanath13,metlitski13,cwang13,chen14}, generalization to incorporate crystalline symmetry is not straightforward.  Our approach, the flux-fusion anomaly test, is a method to test for anomalous symmetry fractionalization for symmetries of the form $G = {\rm U}(1) \rtimes G'$, where $G'$ is an arbitrary group that may include crystalline symmetry.    We focus on bosonic systems, such as spin models or systems of bosons.  We note that some results on anomalous reflection symmetry fractionalization have recently appeared in Ref.~\onlinecite{yqi15b}.  We also note that the ``monopole tunneling'' approach developed in~\cite{metlitski13} and used in~\cite{cwang13} is closely related but not equivalent to the flux-fusion anomaly test, as discussed further in Sec.~\ref{sec:discussion}.

The basic idea of the flux-fusion anomaly test is to start with a symmetry fractionalization pattern for a $d=2$ SET phase, to introduce fluxes of the ${\rm U}(1)$ symmetry, and then to determine whether the fractionalization pattern can be extended to an action of $G'$ symmetry on the ${\rm U}(1)$ fluxes.  Sometimes this is impossible, signaling anomalous symmetry fractionalization.  These considerations only depend on the fusion rules of fluxes and anyon excitations, hence the name for the anomaly test.  We emphasize that we do not need to consider flux threading or flux insertion as a dynamical process.

We implement this idea by gauging a subgroup $\z_n \subset {\rm U}(1)$, and studying the resulting theory.  Gauging symmetry has been employed to study SPT phases, where different phases can be distinguished using the statistics of excitations in the gauged theory \cite{gu12}.  Here, the gauged theory is itself a SET phase with $G'$ symmetry.  We are able to show that some symmetry fractionalization patterns are anomalous by studying the action of $G'$ symmetry on the anyons of the gauged SET phase.

We primarily consider symmetries of the form $G = {\rm U}(1) \times G_{{\rm space}}$ and $G = ({\rm U}(1) \rtimes \zz^T) \times G_{{\rm space}}$, where $\zz^T$ is time reversal and $G_{{\rm space}}$ is a $d=2$ space group. These symmetries arise in a variety of physical settings.  For example, both symmetries are natural in systems of bosons, including situations where electrons form sufficiently tightly bound Cooper pairs.  The former symmetry can arise in a Heisenberg or XY spin system if one ignores time reversal symmetry.  The latter symmetry occurs in a Heisenberg model in a Zeeman field; the field naively breaks time reversal, but preserves a combination of time reversal and a $\pi$ spin rotation perpendicular to the field axis.  We focus on situations where $G$ constrains the symmetry fluxes to be bosons, which simplifies the analysis; we show this occurs whenever time reversal or reflection symmetry is present.

We do not discuss symmetries of the form $G = {\rm U}(1) \times \zz^T \times G_{{\rm space}}$.  This important class of symmetries occurs in time-reversal symmetric XY or Heisenberg spin models.  Application of our anomaly test for these symmetries is subtle (see Sec.~\ref{sec:discussion}), and requires a more intricate analysis that will be presented in a separate paper \cite{tcharge}.

Partially for simplicity, and partially for its physical relevance, we concentrate on two-dimensional $\zz$ topological order, which means that the fusion and braiding of the anyon quasiparticles is the same as the deconfined phase of $\zz$ gauge theory with gapped matter, or, equivalently, Kitaev's toric code model \cite{kitaev03}.  SET phases with $\zz$ topological order are synonymous with gapped $\zz$ quantum spin liquids (QSLs) \cite{chakraborty89,read91,wen91,balents99,senthil00,moessner01a, moessner01b,balents02,kitaev03}, which are of current interest in part due to evidence that such a phase occurs in the $S =1/2$ Heisenberg antiferromagnet on the kagome lattice \cite{yan11,depenbrock12,jiang12b}.  While the symmetries we consider here are more relevant for other systems, $\zz$ QSLs can also occur in those systems. Showing that a given symmetry fractionalization pattern is anomalous constrains the possibilities for $d=2$ $\zz$ QSLs.

As mentioned above, each anomalous symmetry fractionalization pattern we find provides a surface theory for a $d=3$ bosonic TCI. Unlike the case of SPT phases protected by internal symmetry, there is not an existing theory of $d=3$ bosonic TCIs, so it is particularly useful to obtain examples of such phases.  We are able to obtain many such examples, and to discuss some of their physical properties, via their anomalous surface theories.  It is not our goal to provide complete classifications of bosonic TCIs.

For some bosonic TCIs, we can go beyond surface SET phases, and construct a dual vortex field theory for an anomalous surface superfluid.  These superfluids, like some of the surface theories for bosonic topological insulators studied in Ref.~\onlinecite{vishwanath13}, are distinguished by non-trivial symmetry fractionalization of their vortex excitations \cite{senthil04a,senthil04b,balents05}.  The dual vortex field theories thus obtained are convenient to work with, and can be used to explore surface phase diagrams and phase transitions, which may be an interesting direction for future work.

While it is not the focus of this paper, our approach can be used to study internal symmetries when $G = {\rm U}(1) \rtimes G'$, and is complementary to existing approaches in that case.  In particular, for $G = {\rm U}(1) \times \zz^T$, where $\zz^T$ is time reversal, our approach shows that certain fractionalization patterns are anomalous, a result also obtained in previous works \cite{vishwanath13,cwang13}.  The flux-fusion approach confirms that result, without making assumptions about the form of the edge theory of $d=2$ SET phases \cite{cwang13}, or relying on a complete analysis of all possible phases of a surface field theory \cite{vishwanath13}.

Table~\ref{tab:summary} summarizes the main results.  Underlying the detailed results of the table are three distinct types of anomalies:
\begin{enumerate}
\item  For $G = {\rm U}(1) \times \zz^T$, $({\cal T}^m)^2 = -1$ is anomalous, where ${\cal T}^m$ gives the action of time reversal on visons.

\item Whenever $G$ contains a ${\rm U}(1) \times \zz^P$ subgroup, where $\zz^P$ is reflection symmetry, $(P^m)^2 = -1$ is anomalous, where $P^m$ gives the action of the reflection on visons.

\item  Whenever $G$ contains a ${\rm U}(1) \rtimes \zz^T$ subgroup, and also contains some  discrete unitary operation $g$ that commutes with the ${\rm U}(1) \rtimes \zz^T$ subgroup, then
\begin{equation}
{\cal T}^m g^m = - g^m {\cal T}^m 
\end{equation}
is anomalous, where ${\cal T}^m$ and $g^m$ give the action of ${\cal T}$ and $g$, respectively, on visons.  For example, $g$ can be a lattice translation or  reflection.
\end{enumerate}
The first two types of anomalies can be understood in terms of dimensional reduction to $d=1$ SPT phases, but it appears the third type of anomaly cannot be understood in this manner (Sec.~\ref{sec:1dspt}).

We now provide some additional details in order to present Table~\ref{tab:summary}, followed by an outline of the remainder of the paper.  As noted, we focus on $\zz$ topological order, which supports four types of quasiparticle excitations, labeled by $1, e, m, \epsilon$.  Of these, $1$ particles are topologically trivial and can be created by local operators, while the remaining particle types are anyons that cannot be locally created.  We describe the fusion and braiding properties in Sec.~\ref{subsec:gauging}.  Here, we simply note that $\zz$ topological order is realized in the deconfined phase of $\zz$ gauge theory with gapped, bosonic matter, in which case $e$ is the bosonic $\zz$ gauge charge, $m$ is the bosonic $\zz$ gauge flux, and $\epsilon$ is the fermionic charge-flux bound state.  We will also refer to $m$ particles as visons.

We assume through the paper that symmetry does not permute the anyon species.  In this case, the action of symmetry on the anyons is determined by giving the  \emph{fractionalization class} of $e$ and $m$ \footnote{Because $\epsilon$ is a $e$-$m$ bound state, its  fractionalization class is determined by that of $e$ and $m$.}.   
For each of $e, m$, the fractionalization class is an element of $H^2(G, \zz)$.  Here, this is specified uniquely by two pieces of information: 1) whether the particle carries integer or half-odd integer ${\rm U}(1)$ charge, and 2) an element $[\omega_e], [\omega_m] \in H^2(G', \zz)$ that describes the action of $G'$.
Each of $e, m$ transforms as a projective representation of $G'$, and $[\omega_e], [\omega_m]$ encode information about these projective representations that is a universal property of a SET phase (or surface SET phase).  We always choose $e$ to carry half-odd-integer charge, and $m$ to carry integer charge.  One motivation for this choice is that it describes most $\zz$ QSLs that have been proposed to occur in fairly realistic models of spins or bosons.  It can be shown, via a coupled layer construction \cite{cwang13}, that all such symmetry fractionalization patterns (in fact, \emph{any} consistent symmetry fractionalization pattern) can occur as a surface of some $d=3$ SPT phase, which may be the trivial SPT phase (see Appendix~\ref{app:coupled-layer}).  Deciding whether the bulk SPT phase is non-trivial is equivalent to determining whether the corresponding symmetry fractionalization pattern is anomalous.

Under these assumptions, the flux-fusion anomaly test shows that some choices of $[\omega_m]$ imply the symmetry fractionalization pattern is anomalous.
 This result is independent of $[\omega_e]$, which does not play a role in the anomaly test.  Column 2 of Table~\ref{tab:summary} is simply $H^2(G', \zz)$, the set of all possible vison fractionalization classes for $G'$ symmetry.  The anomaly test gives a subset of vison fractionalization classes that ``test negative'' for an anomaly and thus \emph{may} occur strictly in $d=2$.  We refer to such classes as \emph{anomaly-negative}; they form a subgroup ${\cal N}$ of $H^2(G', \zz)$ given in column 3 of Table~\ref{tab:summary}.  It is important to note that anomaly-negative fractionalization classes may still be anomalous; the flux-fusion anomaly test cannot establish that a symmetry fractionalization pattern is non-anomalous.
 
 Finally, for a fixed $[\omega_e]$, the anomaly test gives a set of distinct $d=3$ SPT phases (one of which is always the trivial SPT phase), which are labeled by elements of the quotient ${\cal S} = H^2(G', \zz) / {\cal N}$, given in column 4 of the table. It is important to note that the anomaly test does not distinguish all SPT phases with a given symmetry, so column 4 does not give the full classification of such phases.

We now give an outline of the remainder of the paper.  Section~\ref{sec:simple} gives a simple, somewhat heuristic illustration of the anomaly test in the case of $G = {\rm U}(1) \times  \zz^T$ (time reversal) symmetry. The anomaly test is then described in more detail and greater generality in Sec.~\ref{sec:ft-general}.  First, Sec.~\ref{subsec:gauging} describes the fusion and braiding properties both before and after gauging $\z_n \subset {\rm U}(1)$.  In Sec.~\ref{subsec:anomaly-test}, we describe the action of $G'$ symmetry on the $\z_n$ flux $\Omega$, and use this to present the anomaly test.  Especially for spatial symmetry, it is important for our analysis that $\Omega$ is a boson, which is shown to be the case in Appendix~\ref{app:Omega-is-boson} whenever time reversal or reflection symmetry is present.

In Sec.~\ref{sec:examples}, we apply the anomaly test to the examples of $G = {\rm U}(1) \times \zz^T$, $G = {\rm U}(1) \times \zz^P$,  $G = {\rm U}(1) \times pm$, and $G = ({\rm U}(1) \rtimes \zz^T) \times p1$, where $pm$ is a $d = 2$ space group containing translation and reflection operations, and $p1$ is the $d = 2$ space group consisting only of translations.  We find anomalous symmetry fractionalization patterns in each case.  The first three of these symmetries have anomalies of type 1 and 2 as described above, which can be understood from the viewpoint of dimensional reduction to $d=1$ SPT phases, which is a different way to apply the anomaly test (Sec.~\ref{sec:1dspt}).  In contrast, the last symmetry has type 3 anomalies that apparently cannot be understood in terms of dimensional reduction, as discussed in Sec.~\ref{sec:1dspt}.

Section~\ref{sec:btci} describes how the results from the flux-fusion anomaly test can be used to identify and distinguish some non-trivial $d=3$ SPT phases, including $d=3$ bosonic TCIs.  As discussed in Sec.~\ref{sec:anomalous-superfluids}, some of the bosonic TCIs that we find can have an anomalous surface superfluid that preserves the $G'$ symmetry.  These anomalous superfluids are characterized by vortex excitations that transform projectively under the $G'$ symmetry in a way that is not allowed strictly in $d=2$.  We describe how to construct dual vortex field theories that provide a convenient means to study the physical properties of these surface superfluids and neighboring surface phases.


The paper concludes in Sec.~\ref{sec:discussion} with a discussion of open issues raised by the present results.  Some of the more technical aspects of our results are presented in several appendices, and, in Appendix~\ref{sec:more-examples}, the anomaly test is applied to a few more examples of symmetry groups.

\section{Simple illustration of the anomaly test}
\label{sec:simple}

We begin by giving a somewhat heuristic illustration of the flux-fusion anomaly test, for the case of $G = {\rm U}(1) \times \zz^T$ symmetry.  This symmetry is chosen for simplicity, and for the fact that it has been previously studied using a different approach \cite{cwang13}.  Here, we focus on conveying the intuition and some of the key ideas of our approach.  A more rigorous and more general discussion follows in Sec.~\ref{sec:ft-general}.

Here and throughout the paper, we assume $d=2$ $\zz$ topological order, and that symmetry does not permute the anyon species.  To specify the symmetry fractionalization pattern, we need to give the fractionalization class for both $e$ and $m$ particles.  For the present symmetry, we need to specify whether each particle carries integer  or half-odd-integer ${\rm U}(1)$ charge, and whether it transforms as a Kramers singlet [$({\cal T}^{a})^2 = 1$] or a Kramers doublet [$({\cal T}^{a})^2 = -1$], where $a = e,m$, and ${\cal T}^a$ gives the action of time reversal on anyon $a$.  We denote particles with half-odd integer charge by $C$, and Kramers doublets by $T$, while $0$ is used to indicate particles carrying trivial quantum numbers (integer charge and Kramers singlet).  A fractionalization pattern is thus specified, for example, by the notation $eCmT$ \cite{cwang13}; in this case, $e$ particles carry half-odd-integer charge and are Kramers singlets, while $m$ particles carry integer charge and are Kramers doublets.

We restrict our attention to the case where $e$ carries half-odd-integer charge and $m$ carries integer charge, which includes four fractionalization patterns: $eCm0$, $eCTm0$, $eCmT$, and $eCTmT$.  It is known that the former two patterns are non-anomalous (can be realized in $d=2$); this can be established, for example, via explicit construction of parton gauge theories.  The latter two patterns were argued in Ref.~\onlinecite{cwang13} to be anomalous, via an approach that we contrast with ours at the end of this section.

Our anomaly test is based on introducing fluxes $\Omega_{\phi}$ of the ${\rm U}(1)$ symmetry, where $\phi \in \left[0, 2\pi\right)$.  For the purposes of the present discussion, these fluxes are static point defects in space, obtained by modifying the Hamiltonian.  The symmetry flux $\Omega_{\phi}$ is defined by the following property:  if $Q$ is a local (\emph{i.e.}, non-anyon) excitation carrying unit ${\rm U}(1)$ charge, bringing $Q$ counterclockwise around $\Omega_{\phi}$ results in the statistical phase $\phi$.  We make the restriction $0 \leq \phi < 2\pi$ because $\Omega_{\phi}$ and $\Omega_{\phi + 2\pi}$ have the same mutual statistics with $Q$ and thus carry the same symmetry flux.

Given a fractionalization pattern, the flux-fusion anomaly test proceeds via two steps, which we summarize before proceeding.  First, we study the fusion of symmetry fluxes, and show that, roughly speaking, $\phi = 2\pi$ flux is not trivial, but instead is a $m$ particle excitation.  Second, we consider the action of $\zz^T$ symmetry on  symmetry fluxes $\phi$, and ask whether it is possible to choose this symmetry action to be consistent with the assumed symmetry fractionalization of $m$, given the fusion properties of the fluxes.  We will see there is an inconsistency if $m$ is a Kramers doublet, so that $eCmT$ and $eCTmT$ are anomalous fractionalization patterns.

First, to study the fusion properties of symmetry fluxes, we consider the mutual statistics of a flux $\Omega_{\phi}$ with anyons $e, \epsilon, m$.  We choose particular anyons $e$ and $\epsilon$ carrying ${\rm U}(1)$ charge $1/2$, and $m$ which is neutral under ${\rm U}(1)$.  We could consider anyons with other allowed values of the charge (for example, there will also be $e$ particles with charge $-1/2$), but this does not affect the results.  Let $\Theta_{a, \Omega_{\phi}}$ be the statistical phase angle when anyon $a$ is brought counterclockwise around the flux $\phi$.  Then, given the assumed charge values for the anyons, we have
\begin{eqnarray}
\Theta_{e,\Omega_{\phi}} &=& \Theta_{\epsilon, \Omega_{\phi}} =\frac{\phi}{2} \label{eqn:flux-mutual-1} \\
\Theta_{m, \Omega_{\phi}} &=& 0 \text{.} \label{eqn:flux-mutual-2}
\end{eqnarray}

To obtain some intuition for the fusion properties of the symmetry fluxes, suppose for the moment that we relax the restriction $\phi < 2\pi$.  Then, if $\phi = 2\pi$, we have formally $\Theta_{e, \Omega_{2\pi}} = \Theta_{\epsilon,\Omega_{2\pi}} = \pi$ and $\Theta_{m, \Omega_{2\pi}} = 0$.  Since $\Omega_{2\pi}$ carries trivial symmetry flux (it has trivial mutual statistics with $Q$), it must be identified with one of the anyon quasiparticles.  Putting $\phi \to 2\pi$ in Eqs.~(\ref{eqn:flux-mutual-1}) and~(\ref{eqn:flux-mutual-2}), we have the identification $\Omega_{2\pi}  = m$.  Along the same lines, we can identify $\Omega_{4\pi} = 1$.

We prefer to keep the restriction $0 \leq \phi < 2\pi$, in which case essentially the same result can be obtained as follows:  Suppose that we have two $\pi$ fluxes $\Omega_{\pi}$.  The total flux is $2\pi$, which is equivalent to no symmetry flux at all.  Therefore, we have the fusion rule
\begin{equation}
\Omega_{\pi} \Omega_{\pi} = a \text{,}
\end{equation}
where $a$  is a quasiparticle excitation that carries no symmetry flux, but may be a non-trivial anyon.
The particle $a$ can be identified by its mutual statistics with $e$, $m$ and $\epsilon$, which follows from the additivity properties of statistics.  For example,
\begin{equation}
\Theta_{e,a} = \Theta_{e, \Omega_{\pi} \Omega_{\pi} } = 2 \Theta_{e, \Omega_{\pi}} = \pi \text{.}
\end{equation}
Similarly, $\Theta_{\epsilon, a} = \pi$ and $\Theta_{m, a} = 0$, which implies $a = m$ and
\begin{equation}
\Omega_{\pi} \Omega_{\pi} = m \text{.} \label{eqn:pi-plus-pi-is-m}
\end{equation}
It should be noted that this result has a discrete character and does not make use of the fact that ${\rm U}(1)$ is a continuous group.  Indeed, the same result holds if we replace ${\rm U}(1)$ by the discrete group $\zz$.

Next, we consider the action of time reversal symmetry ${\cal T}$ on the symmetry fluxes $\Omega_{\pi}$.  First, we observe that ${\cal T}$ does not change the value of the flux $\phi$, because ${\cal T}$ commutes with ${\rm U}(1)$ rotations. Therefore, $\Omega_{\pi}$ transforms either as a Kramers singlet or a Kramers doublet under time reversal.  If we assume that $m$ is a Kramers doublet, we now have a contradiction with Eq.~(\ref{eqn:pi-plus-pi-is-m}):  whether $\Omega_{\pi}$ is a Kramers singlet or doublet, the composite $\Omega_{\pi} \Omega_{\pi}$ must be a Kramers singlet.

We have thus found that that 
 $eCmT$ and $eCTmT$ are anomalous fractionalization patterns.  This is true because, in strict $d=2$, it should always be possible to introduce ${\rm U}(1)$ symmetry fluxes and to view these as point objects, so the contradiction we obtained means that a fractionalization pattern cannot occur strictly in $d=2$.  On the other hand, on the surface of a $d=3$ SPT phase, symmetry fluxes are line objects that penetrate into the bulk, and it may not be sensible to view them as point objects where they pierce the surface.  Therefore, $eCmT$ and $eCTmT$ may occur on the surface of a $d=3$ SPT phase.  Indeed, this is the case, and was  demonstrated in Ref.~\onlinecite{cwang13} via an elegant coupled layer construction.

The above analysis is complementary to the approach of Ref.~\onlinecite{cwang13}.  There, among other results, Chern-Simons theory was used to construct chiral boson edge theories for SET phases with $\zz$ topological order and $G = {\rm U}(1) \times \zz$ symmetry.  For some symmetry fractionalization patterns, including $eCmT$ and $eCTmT$, it was shown that no corresponding edge theory can be constructed, and it was concluded that these symmetry fractionalization patterns are anomalous.  Strictly speaking, to draw this conclusion, one has to assume that the class of edge theories considered is in some sense sufficiently general, and, while this assumption seems reasonable, we do not know of an argument that this is the case.  The flux-fusion approach requires no such assumption, and in the present case, its results agree with those of Ref.~\onlinecite{cwang13}, for those fractionalization patterns where both approaches can be applied.

\section{Flux-fusion anomaly test: general discussion}
\label{sec:ft-general}

\subsection{Gauging $\z_n \subset {\rm U}(1)$ symmetry}
\label{subsec:gauging}

The simple discussion of the anomaly test in Sec.~\ref{sec:simple} is based on inserting ${\rm U}(1)$ symmetry fluxes, which are static point defects in space.  Because our objective is to consider crystalline symmetry, this approach is not ideal, because inserting a non-dynamical flux at some point in space will usually partially or fully break the crystalline symmetry.  In addition, there is not an existing theory describing the action of $G'$ symmetry on fluxes of the continuous ${\rm U}(1)$ symmetry.

Therefore, we prefer to proceed by gauging a $\z_n$ subgroup of the ${\rm U}(1)$ symmetry, for all integers $n \geq 2$.  That is, we imagine minimally coupling our system to a dynamical $\z_n$ gauge field, where the $\z_n$ gauge group is identified with $\z_n \subset {\rm U}(1)$ global symmetry.  In Appendix~\ref{app:gauging-procedure}, we give an explicit procedure showing that, for the symmetry groups considered in this paper, it is possible to gauge this $\z_n$ subgroup while preserving $G' \subset G$ symmetry.  The resulting theory is a \emph{gauged} SET phase, where the symmetry flux behaves as a gapped, dynamical quasiparticle excitation.  This allows us to study symmetry fluxes without breaking crystalline symmetry.  In addition, we can build on existing results to describe the action of $G'$  on the excitations of the gauged SET phase.

We consider a $d=2$ SET phase with $\zz$ topological order and $G = {\rm U}(1) \rtimes G'$ symmetry.  We now describe the fusion and braiding properties of the anyons of the SET phase.  Fusion of anyons is described by the Abelian group ${\cal A} = \zz \times \zz$, generated by $e$ and $m$, which obey the relations
\begin{eqnarray}
e^2 &=& m^2 = 1 \\
\epsilon &\equiv& e m = m e \text{.}
\end{eqnarray}
We assume that $e$ carries half-odd-integer charge under ${\rm U}(1)$.  Under $\z_n \subset {\rm U}(1)$ symmetry, this means that
\begin{equation}
(U^e_n)^n = -1 \label{eqn:half-zn} \text{,}
\end{equation}
where $U^e_n$ is a unitary operator representing the action of a generator of $\z_n$ on a single $e$ particle.  Half-odd integer charge is only non-trivial for even $n$; if $n$ is odd, then Eq.~(\ref{eqn:half-zn}) can be trivialized by the allowed redefinition $U^e_n \to - U^e_n$.  Therefore, we restrict attention to even values of $n$.  We also assume that $m$ carries integer ${\rm U}(1)$ charge, so that under $\z_n$ we have $(U^m_n)^n = 1$.  The action of $G'$ on $e$ and $m$ is characterized below in Sec.~\ref{subsec:anomaly-test}.

To specify the statistics, we introduce some notation that will be particularly helpful in describing the gauged SET phase.
For anyons $a,b \in {\cal A}$, let $\theta_a$ give the self-statistics angle of $a$, and let $\Theta_{a,b}$ be the mutual statistics angle, where $a$ is taken counterclockwise around $b$.  These quantities satisfy the following general properties for any $a,b,c \in {\cal A}$:
\begin{eqnarray}
\theta_{1} &=& \Theta_{1,a} = 0 \\
\Theta_{a,a} &=& 2 \theta_a \\
\Theta_{a, b} &=& \Theta_{b, a} \\
\Theta_{ab,c} &=& \Theta_{a,c} + \Theta_{b,c} \\
\theta_{ab} &=& \theta_a + \theta_b + \Theta_{a,b} \text{.}
\end{eqnarray}
These and other equations for $\theta_a$ and $\Theta_{a,b}$ are always understood to be true modulo $2\pi$.
The statistics of $\zz$ topological order is then fully specified by
\begin{eqnarray}
\theta_e &=& \theta_m = 0 \\
\Theta_{e,m} &=& \pi \text{.}
\end{eqnarray}
These equations say that $e$ and $m$ are bosons with $\Theta_{e,m} = \pi$ mutual statistics.

We now consider the gauged SET phase, obtained by gauging $\z_n \subset {\rm U}(1)$.  The anyons of the gauged SET phase are Abelian; this follows from Eq.~399 and the surrounding discussion of Ref.~\onlinecite{barkeshli14}. The fusion rules are described by the Abelian group ${\cal A}_G$, which is generated by $e, m, Q$, and $\Omega$.  Here, $Q$ is the unit $\z_n$ symmetry charge, which is a local excitation of the un-gauged theory, but is now an anyon in the gauged SET phase.  $\Omega$ is the unit $\z_n$ symmetry flux.  Upon gauging $\z_n$, the $e$ and $m$ sectors in the un-gauged theory each break into $n$ different sectors with distinct $\z_n$ symmetry charge.  In the gauged SET phase, $e$ and $m$ each correspond to a particular choice among such subsectors.  The choice of subsector is arbitrary, and can be changed by redefining $e$ or $m$ by binding symmetry charges; for example $e \to Q e$ is an allowed redefinition.  There is also arbitrariness in the choice of symmetry flux, which can be redefined by  $\Omega \to Q \Omega$, or by $\Omega \to a \Omega$, where 
 $a$ is an anyon of the un-gauged theory.
 
The fusion rules are
\begin{eqnarray}
Q^n &=&  1 \label{eqn:qn-fusion} \\
e^2 &=& Q \label{eqn:e2-fusion} \\
m^2 &=& 1 \label{eqn:m2-fusion} \\
\Omega^n &=& a Q^k  \label{eqn:Omegan-fusion} \text{.}
\end{eqnarray}
Equation~(\ref{eqn:qn-fusion}) is obvious.  Equations~(\ref{eqn:e2-fusion}) and~(\ref{eqn:m2-fusion}) correspond to making a particular choice of $e$ and $m$ among the possible subsectors.  The most important fusion rule in our analysis is Eq.~(\ref{eqn:Omegan-fusion}).  There, $a$ is an anyon of the un-gauged theory to be determined, and $k$ is some as yet unknown integer satisfying $0 \leq k < n$.  This equation expresses the fact that $\Omega^n$ carries no $\z_n$ symmetry flux, but otherwise, at this stage in the analysis, could be an arbitrary particle in the gauged SET phase.

In order to fix the fusion rule Eq.~(\ref{eqn:Omegan-fusion}), we consider the statistics of the gauged SET phase.  We have
\begin{eqnarray}
\theta_e &=& \theta_m = 0 \label{eqn:em-bosons} \\
\Theta_{e,m} &=& \pi \label{eqn:em-mutual} \\
\theta_Q &=& \Theta_{e,Q} = \Theta_{m,Q} = 0 \label{eqn:Q-trivial-braiding} \\
\Theta_{Q,\Omega} &=& \frac{2\pi}{n} \label{eqn:QOmega-mutual} \\
\Theta_{e,\Omega} &=& \frac{\pi}{n} + p_e \pi \label{eqn:eOmega-mutual} \\
\Theta_{m, \Omega} &=& p_m \pi  \label{eqn:mOmega-mutual} \text{.}
\end{eqnarray}
Here, Eqs.~(\ref{eqn:em-bosons}) and~(\ref{eqn:em-mutual}) are the braiding statistics for the un-gauged SET phase.  Equation~(\ref{eqn:Q-trivial-braiding}) holds because the symmetry charge $Q$ must have trivial braiding with itself and with anyons of the un-gauged theory.  Equation~(\ref{eqn:QOmega-mutual}) is the defining property of the symmetry flux $\Omega$.  Finally, Eqs.~(\ref{eqn:eOmega-mutual}) and~(\ref{eqn:mOmega-mutual}) follow from Eqs.~(\ref{eqn:e2-fusion}) and~(\ref{eqn:m2-fusion}), respectively, with unknown parameters $p_e, p_m = 0,1$.

We redefine $e$ and $m$ to set $p_e = p_m = 0$.  For example, if $p_e = 1$, we redefine $e \to Q^{n/2} e$.  This leaves the fusion rules unchanged, and results in $\Theta_{e, \Omega} = \pi / n$, without modifying the other statistics angles.

Now, we use the statistics to constrain the flux fusion rule, Eq.~(\ref{eqn:Omegan-fusion}).  Using $\Theta_{m, \Omega} = 0$, we have $\Theta_{m, \Omega^n} = 0$.  Consistency with Eq.~(\ref{eqn:Omegan-fusion}) then requires either $a=1$ or $a=m$.  Similarly, $\Theta_{e,\Omega} = \pi / n$ implies $\Theta_{e, \Omega^n} = \pi$, which requires either $a = m$ or $a = \epsilon$.  Therefore, $a = m$, and
\begin{equation}
\Omega^n = m Q^k \text{.} \label{eqn:Omegan-fusion-2}
\end{equation}

So far, we have not mentioned $\theta_{\Omega}$, the self-statistics of the symmetry flux.  Unlike the other statistics angles, this parameter does not follow immediately from our assumptions, but it can be related to the integer $k$ appearing in Eq.~(\ref{eqn:Omegan-fusion-2}).  First, Eq.~(\ref{eqn:Omegan-fusion-2}) implies that $\Omega^n$ is a boson, so $\theta_{\Omega^n} = n^2 \theta_{\Omega} = 0$, and therefore
\begin{equation}
\theta_{\Omega} = \frac{2\pi q}{n^2}
\end{equation}
for some integer $q$ satisfying $0 \leq q < n^2$.  In fact, we can further restrict the range of $q$.  To see this, we make the redefinition $\Omega \to e \Omega$ and $m \to Q^{n/2} m$, which preserves the fusion rules, and leaves all the statistics angles unchanged except $\theta_{\Omega}$.  The effect of this redefinition is to shift $q \to q + n/2$, which allows us to restrict $0 \leq q < n/2$.

We can now relate $q$ and $k$ by noting that $\Theta_{\Omega, \Omega^n} = 2 n \theta_{\Omega} = 4 \pi q / n$, and also $\Theta_{\Omega, \Omega^n} = \Theta_{\Omega, m Q^k} = 2\pi k / n$, so that $4 \pi q / n = 2 \pi k / n$.  This has no solution if $k$ is odd, so $k$ must be even.  Given the restrictions on the range of $k$ and $q$, the unique solution for $q$ is then $q = k/2$, and we have shown
\begin{equation}
\theta_{\Omega} = \frac{\pi k}{n^2} \text{,} \label{eqn:Omega-statistics}
\end{equation}
where $k$ is even and satisfies $0 \leq k < n$.  In particular, for $n = 2$ we have $\Omega^2 = m$, as stated in Sec.~\ref{sec:simple}.

Physically, we expect $k$ to parametrize the quantized Hall response.  Inserting $2\pi$ flux at some point in space produces a local charge accumulation of $\sigma_{xy}$, in appropriate units.  If we view fusion of $n$ fluxes $\Omega$ as equivalent to a dynamical process where $n$ fluxes are inserted, then, because $m$ is neutral under $\z_n$, Eq.~(\ref{eqn:Omegan-fusion-2}) implies
\begin{equation}
k = \sigma_{xy} \operatorname{mod} n \text{.}
\end{equation}
This physical interpretation of $k$ leads us to expect $k = 0$ whenever $G'$ symmetry forbids a quantized Hall response.  Indeed, in Appendix~\ref{app:Omega-is-boson}, we show that $k = 0$  whenever $G'$ contains time reversal or spatial reflection symmetry.

Whenever $k=0$, by Eq.~(\ref{eqn:Omega-statistics}), $\Omega$ is a boson.  This will enable a simple description of the action of $G'$ symmetry on $\Omega$ and $m$, so from now on we will always assume conditions are such that we can take $\Omega$ to be a boson.  Under this assumption, we collect here the properties of the gauged SET phase obtained from the discussion above.  The fusion rules are
\begin{eqnarray}
Q^n &=&  1 \label{eqn:qn-fusion-summary} \\
e^2 &=& Q \label{eqn:e2-fusion-summary} \\
m^2 &=& 1 \label{eqn:m2-fusion-summary} \\
\Omega^n &=& m  \label{eqn:Omegan-fusion-summary} \text{,}
\end{eqnarray}
and the statistics are specified by
\begin{eqnarray}
\theta_e &=& \theta_m = 0 \label{eqn:em-bosons-summary} \\
\Theta_{e,m} &=& \pi \label{eqn:em-mutual-summary} \\
\theta_Q &=& \Theta_{e,Q} = \Theta_{m,Q} = 0 \label{eqn:Q-trivial-braiding-summary} \\
\Theta_{Q,\Omega} &=& \frac{2\pi}{n} \label{eqn:QOmega-mutual-summary} \\
\Theta_{e,\Omega} &=& \frac{\pi}{n}  \label{eqn:eOmega-mutual-summary} \\
\Theta_{m, \Omega} &=& 0  \label{eqn:mOmega-mutual-summary} \\
\theta_{\Omega} &=& 0 \label{eqn:Omega-statistics-summary} \text{.}
\end{eqnarray}
These are precisely the fusion rules and statistics of $\z_{2n}$ gauge theory, or, equivalently, the $\z_{2n}$ version of the toric code model.  For Abelian anyons, fusion rules and statistics are enough to uniquely specify the unitary modular tensor category that describes a theory of anyons \cite{zwangpc,drinfeld10}.  Therefore, the theory of anyons in the gauged SET phase is identical to that in the $\z_{2n}$ toric code.

\subsection{Symmetry action on $m$, $\Omega$ and the anomaly test}
\label{subsec:anomaly-test}

In order to apply the anomaly test, we first have to characterize the action of $G'$ symmetry on the anyons of the un-gauged SET phase \cite{essin13}.  In general, the fractionalization class of $e$ or $m$ is an element of the group $H^2(G, \zz)$.  In the present case, as is shown in Appendix~\ref{app:fracinfo}, it is enough to specify separately the action of ${\rm U}(1)$ and $G'$ on each of $e$ and $m$.  That is, there is no additional information associated with interplay between ${\rm U}(1)$ and $G'$.

Each of $e$, $m$ transforms under a $\zz$ projective representation of $G'$ denoted $\Gamma^e$, $\Gamma^m$, respectively.  We focus on $m$ particles; the corresponding equations hold for $e$ particles with trivial modifications.  For $g_1, g_2 \in G'$, we have
\begin{equation}
\Gamma^m(g_1) \Gamma^m(g_2) = \omega_m(g_1, g_2) \Gamma^m (g_1 g_2) \text{,} \label{eqn:projective-rep}
\end{equation}
where $\omega_m(g_1, g_2) \in \zz$ is called a $\zz$ factor set.  The corresponding object for $e$ particles is denoted $\omega_e$.  Associative multiplication of the $\Gamma^m$'s implies
\begin{equation}
\omega_m(g_1, g_2) \omega_m(g_1 g_2, g_3) = \omega_m( g_1, g_2 g_3) \omega_m(g_2, g_3) \text{.} \label{eqn:z2-associativity}
\end{equation}
In general, any $\zz$-valued function $\omega_m(g_1, g_2)$ satisfying Eq.~(\ref{eqn:z2-associativity}) is called a $\zz$ factor set.

Physical properties are unchanged under a redefinition $\Gamma^m(g) \to \lambda^{-1}(g) \Gamma^m(g)$ for $\lambda(g) \in \zz$, which induces a \emph{projective transformation} on the factor set,
\begin{equation}
\omega_m(g_1, g_2) \to \lambda^{-1}(g_1) \lambda^{-1}(g_2) \lambda(g_1 g_2) \omega_m(g_1, g_2) \text{.}  \label{eqn:z2-projective-transformation}
\end{equation}
Here, $\lambda^{-1}(g) = \lambda(g)$, but the inverse signs are kept to expose the formal similarities with the discussion of symmetry action on $\Omega$, below.
Equivalence classes of factor sets under such projective transformations are denoted $[\omega_m]_{\zz}$, and are the distinct fractionalization classes of $m$.  The $\zz$ subscript reminds us that both $\omega_m$ and the projective transformations $\lambda$ take values in $\zz$.  In the language of group cohomology theory, fractionalization classes  $[\omega_m]_{\zz}$ are elements of the Abelian group $H^2(G', \zz)$, the second group cohomology of $G'$ with $\zz$ coefficients. The group multiplication in $H^2(G', \zz)$ is obtained from multiplication of functions;  that is, if $\omega_{ab}(g_1,g_2) = \omega_a(g_1,g_2) \omega_b(g_1,g_2)$, then $[\omega_a]_{\zz} [\omega_b]_{\zz} = [\omega_{ab}]_{\zz}$.

Considering all symmetries together, the symmetry fractionalization pattern of the SET phase can be denoted $eC[\omega_e]m0[\omega_m]$, where $C$ ($0$) indicates that $e$ ($m$) carries half-odd-integer (integer) ${\rm U}(1)$ charge.  When using this notation, to avoid cumbersome expressions, we drop the $\zz$ subscript for the fractionalization classes.

The flux-fusion anomaly test will be able to determine that $eC[\omega_e]m0[\omega_m]$ is anomalous for certain choices of $[\omega_m]_{\zz}$, independent of $[\omega_e]_{\zz}$.  When the anomaly test does not find an anomaly, we say that a symmetry fractionalization pattern is \emph{anomaly-negative}.  This terminology recognizes that the flux-fusion anomaly test is not expected to detect all possible anomalies, and some anomaly-negative fractionalization patterns can still be anomalous.

To proceed, we now consider the gauged SET phase, and characterize the action of $G'$ symmetry on $\Omega$.  First, we need to consider the possibility that some operations may permute the anyons of the gauged SET phase, and, in particular, may map $\Omega$ to some other anyon.  For some operation $g \in G'$, let $g \star \Omega$ denote the anyon resulting from applying $g$ to $\Omega$.  If $g$ commutes with ${\rm U}(1)$, is unitary, and is either an internal symmetry or a proper space group operation, then $g \star \Omega = \Omega$.  This follows from the fact that such an operation leaves $Q$, $e$ and $m$ invariant, and also leaves the statistics invariant; that is, $\Theta_{g \star a, g \star b} = \Theta_{a,b}$.  However, it is not the case that all $g \in G'$ leave $\Omega$ invariant; in particular, we will be interested in time reversal and reflection symmetry.  These operations may send $\Omega \mapsto \Omega$ or $\Omega \mapsto \Omega^{2n-1}$, depending on whether the operation in question commutes with ${\rm U}(1)$, as is discussed in detail in Appendix~\ref{app:Omega-is-boson}.

Because some operations in $G'$ may not preserve the anyon type of $\Omega$, in describing the action of symmetry, we have to go somewhat beyond the framework developed in Ref.~\onlinecite{essin13}.  We introduce field operators $\psi_k$ ($k = 1,\dots,2n-1$).  Each $\psi_k$ is a many-component object, with components not explicitly written, where each component creates a $\Omega^k$ particle in some state.  In particular, $\psi_n$ creates a $m$ particle.  These field operators are non-local objects.  However, because all the $\Omega^k$ particles are bosons and have bosonic mutual statistics, the non-local character of $\psi_k$ is not expected to play a role in the following discussion.  It is also convenient to collect all the field operators into the object $\Psi = (\psi_1 \cdots \psi_{2n-1})$.

All physical states and local operators are invariant under $\z_{2n}$ gauge transformations implemented by the unitary operator ${\cal G}[\lambda]$, for $\lambda \in \z_{2n}$, which acts on the field operators by
\begin{equation}
{\cal G}[\lambda] \psi_k {\cal G}[\lambda]^{-1} = \lambda^k \psi_k \text{.}
\end{equation}
For a symmetry operation $g \in G'$, we denote the corresponding unitary or anti-unitary operator by $S(g)$, which acts on field operators by 
\begin{equation}
g: \Psi \mapsto S(g) \Psi S(g)^{-1} \text{.}
\end{equation}
The operators $S(g)$ form a representation of $G'$ up to $\z_{2n}$ gauge transformations, that is
\begin{equation}
S(g_1) S(g_2) = {\cal G}[ \phi_n(g_1, g_2) ] S(g_1 g_2) \text{,}  \label{eqn:twisted-projective-rep}
\end{equation}
for $\phi_n(g_1, g_2) \in \z_{2n}$.  This is the most general multiplication law consistent with the requirement that $S(g)$ act linearly on local operators, for example (schematically), $(\psi_1)^{2n}$.  Mathematically, we have defined a kind of generalized projective representation, which is similar to but not identical to the projective representation describing the action of symmetry on $m$ [Eq.~(\ref{eqn:projective-rep})].  

The crucial difference between $S(g)$ and more familiar projective representations is that, in general, $S(g)$ does not commute with the gauge transformation ${\cal G}[\lambda]$.  We note that some symmetries $g \in G'$ map $g : \Omega \mapsto \Omega^{2n-1}$.  We keep track of this information by defining
\begin{equation}
s(g) = \left\{ \begin{array}{ll} +1 , & g : \Omega \mapsto \Omega \\
-1 , & g : \Omega \mapsto \Omega^{2n-1} 
\end{array}\right. \text{.}
\end{equation}
In addition, some operations in $G'$ may be anti-unitary, so we define
\begin{equation}
u(g) = \left\{ \begin{array}{ll}
1 \text{,} & g \text{ unitary} \\
-1 \text{,} & g \text{ anti-unitary} 
\end{array}\right. \text{.}
\end{equation}
We note that both $s$ and $u$ are group homomorphisms mapping $G' \to \zz$.  We then introduce the function
\begin{equation}
t(g) = s(g) u(g) \text{.}
\end{equation}
By considering the action of $S(g)$ and gauge transformations on field operators, it is straightforward to show
\begin{equation}
S(g) {\cal G}[\lambda] = {\cal G}[\lambda^{t(g)} ] S(g) \text{,}  \label{eqn:SGmult}
\end{equation}
which shows that $t(g)$ characterizes the non-commutativity of $S(g)$ and gauge transformations.  We thus refer to $S(g)$ as a $t$-twisted $\z_{2n}$ projective representation of $G'$.

Equation~(\ref{eqn:SGmult}) allows us to use associativity of the product $S(g_1) S(g_2) S(g_3)$ to derive the associativity condition on $\phi_n$,
\begin{equation}
\phi_n(g_1, g_2) \phi_n(g_1 g_2, g_3) = \phi_n(g_1, g_2 g_3)  [\phi_n(g_2, g_3)]^{t(g_1)} \text{.}  \label{eqn:twisted-associativity}
\end{equation}
We refer to $\phi_n$, and, indeed, any $\z_{2n}$-valued function satisfying Eq.~(\ref{eqn:twisted-associativity}), as a $t$-twisted $\z_{2n}$ factor set.  Paralleling the discussion of ordinary projective representations above, we are free to redefine $S(g)$ by a gauge transformation,
\begin{equation}
S(g) \to {\cal G}[\lambda^{-1}(g)] S(g) \text{.}
\end{equation}
This induces a projective transformation on the factor set,
\begin{equation}
\phi_n(g_1, g_2) \to \lambda^{-1}(g_1) [\lambda(g_2)]^{-t(g_1)} [\lambda(g_1 g_2)] \phi_n(g_1, g_2) \text{.}  \label{eqn:twisted-proj-trans}
\end{equation}
Equivalence classes $[\phi_n]_{\z_{2n}}$ of factor sets under such transformations characterize the action of $G'$ symmetry on $\Omega$.  These equivalence classes are elements of the cohomology group $H_t^2(G', \z_{2n})$, where the $t$ subscript denotes the non-trivial action of $G'$ on the $\z_{2n}$ coefficients, encoded in the function $t(g)$.  We refer to this as $t$-twisted cohomology.  We note that, for $G'$ finite, on-site and unitary, we have recovered a special case of the twisted cohomology theory used to describe the action of symmetry on anyons in the category-theoretic description of SET phases \cite{barkeshli14,tarantino15}.

 In fact, $[\phi_n]_{\z_{2n}}$ simultaneously characterizes the action of $G'$ on all particles $\Omega^k$ obtained by fusing $\Omega$'s together.  This includes $m = \Omega^n$.  The action of symmetry on $m$ is given by considering the action of $S(g)$ on $\psi_n$; in particular,
 \begin{eqnarray}
  && S(g_1)  S(g_2) \psi_n S(g_2)^{-1} S(g_1)^{-1} \nonumber \\ &=& [\phi_n(g_1, g_2)]^n S(g_1 g_2) \psi_n S(g_1 g_2)^{-1} \nonumber \\
  &\equiv& \omega_m(g_1, g_2) S(g_1 g_2) \psi_n S(g_1 g_2)^{-1} \text{.}
 \end{eqnarray}
 Therefore, we have shown
 \begin{equation}
 \omega_m(g_1, g_2) = [\phi_n(g_1,g_2)]^n \text{.}  \label{eqn:main-equation}
 \end{equation}
 
 Equation~(\ref{eqn:main-equation}), which holds for all even $n \geq 2$, is the crucial equation underlying the anomaly test.  The essential idea is to take advantage of the fact that $\Omega$ is a ``$n$th root'' of $m$ in the gauged SET phase, and to ask whether a given symmetry action on $m$ can be consistently extended to a symmetry action on its $n$th root $\Omega$.  If not, then an anomaly has been detected. 
 
  In more detail, the logic is as follows: Given $[\omega_m]_{\zz}$, we choose some particular factor set $\omega_m(g_1, g_2)$ in the desired equivalence class (the particular choice within the class does not matter).  Then, for each even $n \geq 2$, we ask whether it is possible to solve Eq.~(\ref{eqn:main-equation}) for $\phi_n(g_1, g_2)$, where $\phi_n$ is required to  satisfy Eq.~(\ref{eqn:twisted-associativity}). If for any even $n \geq 2$, a solution fails to exist, the symmetry fractionalization pattern is anomalous.  If a solution exists for all even $n \geq 2$, the symmetry fractionalization pattern is anomaly-negative.
 
Equation~(\ref{eqn:main-equation}) immediately implies that anomaly-negative $m$ particle fractionalization classes form a subgroup that we denote ${\cal N} \subset H^2(G', \zz)$.

At first glance, it might appear cumbersome to apply the flux-fusion anomaly test.  Fortunately, it is not necessary to consider Eq.~(\ref{eqn:main-equation}) directly for every even $n \geq 2$.  Instead, there is a simple and easily computable characterization of which $[\omega_m]_{\zz}$ are anomaly-negative.  To describe this characterization, we first note that $\omega_m$ can be viewed as a $t$-twisted ${\rm U}(1)$ factor set.  This means that, given $\omega_m(g_1,g_2)$, we allow for projective transformations
\begin{equation}
\omega_m(g_1, g_2) \to \lambda^{-1}(g_1) [\lambda(g_2)]^{-t(g_1)} [\lambda(g_1 g_2)] \omega_m(g_1,g_2) \text{,}
\end{equation}
where $\lambda(g) \in {\rm U}(1)$.  The corresponding equivalence class under these transformations is denoted $[\omega_m]_{{\rm U}(1)}$, and is an element of the cohomology group $H^2_t(G', {\rm U}(1))$.
Formally, there is a map $\rho_2 : H^2(G', \zz) \to H^2_t(G', {\rm U}(1))$ defined by $\rho_2( [\omega_m]_{\zz} ) = [\omega_m]_{{\rm U}(1)}$.  (In Appendix~\ref{app:characterization-theorem}, it is shown that $\rho_2$ is well-defined, is a group homomorphism, and is unique in a certain natural sense.)

Intuitively, it seems natural for cohomology with ${\rm U}(1)$ coefficients to arise out of the flux-fusion anomaly test.  Ultimately, it ought to be possible to dispense with gauging $\z_n \subset {\rm U}(1)$ for all even $n$, in favor of working directly with continuous ${\rm U}(1)$ fluxes.  Either approach should give the same results, so we speculate that the $H^2_t(G', {\rm U}(1))$ cohomology may describe the action of $G'$ symmetry on ${\rm U}(1)$ fluxes.  Moreover, as discussed in more detail in Sec.~\ref{sec:anomalous-superfluids}, $[\omega_m]_{{\rm U}(1)}$ does have a nice 
physical interpretation: it characterizes the symmetry fractionalization of vortex excitations in a superfluid.  This  allows us to obtain results on anomalous $d=2$ superfluids.  We note that $t$-twisted ${\rm U}(1)$ cohomology also appears in the cohomology approach to SPT phases with time reversal symmetry, where anti-unitary operations act non-trivially on the ${\rm U}(1)$ coefficients, and the cohomology groups are denoted by $H^n(G, {\rm U}_T(1))$ \cite{chen13}.

Anomaly-negative $m$ particle fractionalization classes $[\omega_m]_{\zz}$ are fully characterized by the following theorem, which is proved in Appendix~\ref{app:characterization-theorem}.
\begin{thmtext} \label{thm:characterization-maintext}
If $H^2_t(G', {\rm U}(1)) = {\rm U}(1)^k \times A$, where $A$ is a finite product of finite cyclic factors, then the symmetry fractionalization pattern $eC[\omega_e]m0[\omega_m]$ is anomaly-negative if and only if $[\omega_m]_{{\rm U}(1)} = \rho_2( [\omega_m]_{\zz} )$ lies in the connected component of $H^2_t(G', {\rm U}(1) )$ that contains the identity element.
\end{thmtext} 
The assumption on the form of $H^2_t(G', {\rm U}(1) )$ is true for all the examples we have considered, and we believe it is likely to be true in general.

This theorem allows us to apply the flux-fusion anomaly test via the following procedure:
\begin{enumerate}
\item Compute the group $H^2(G', \zz)$ of $m$ particle fractionalization classes under $G'$ symmetry.  Find a convenient explicit parametrization of distinct fractionalization classes $[\omega_m]_{\zz}$.

\item Compute the $t$-twisted cohomology group $H^2_t(G', {\rm U}(1) )$, and find an explicit parametrization.

\item Find the map $\rho_2$ discussed above, for which $[\omega_m]_{{\rm U}(1)} = \rho_2( [\omega_m]_{\zz})$.

\item For each $m$ particle fractionalization class $[\omega_m]_{\zz}$, determine whether $[\omega_m]_{{\rm U}(1)}$ can be continuously deformed to the identity element of $H^2_t(G', {\rm U}(1))$.  If this is impossible, the fractionalization pattern $eC[\omega_e]m0[\omega_m]$ is anomalous, for any $[\omega_e]_{\zz}$.

\item The results for a given symmetry $G = {\rm U}(1) \rtimes G'$ can be summarized by describing the $m$ particle fractionalization classes for which $eC[\omega_e]m0[\omega_m]$ is anomalous.
\end{enumerate}
This procedure is illustrated in detail, and made more concrete, in the examples presented in Section~\ref{sec:examples} and Appendix~\ref{sec:more-examples}.

\section{Examples}
\label{sec:examples}

In this Section, we apply the flux-fusion anomaly test in a few cases, in order of increasing complexity.  In each case, we fix a symmetry $G$, and follow the procedure outlined in Sec.~\ref{subsec:anomaly-test}.  A crucial aspect is the calculation of second cohomology groups, for $G'$ presented in terms of generators and relations.  We illustrate our approach to these calculations in each example, leaving a more careful mathematical justification to Appendix~\ref{app:cohomology}.

These examples enable a more concrete discussion of bosonic TCIs in Sec.~\ref{sec:btci}, and anomalous $d=2$ superfluids in Sec.~\ref{sec:anomalous-superfluids}.  A number of other examples are considered in Appendix~\ref{sec:more-examples}.

\subsection{$G = {\rm U}(1) \times \zz^T$}
\label{sec:ex1}

We begin with the case of $G = {\rm U}(1) \times \zz^T$ symmetry, that was already discussed in Sec.~\ref{sec:simple} and in previous works \cite{vishwanath13,cwang13}.  This symmetry is simple enough to analyze using Eq.~(\ref{eqn:main-equation}) directly; that approach, in fact, is essentially identical to the treatment in Sec.~\ref{sec:simple}.  However, to pave the way for more complex examples, we follow the procedure outlined in Sec.~\ref{subsec:anomaly-test}.

It is convenient to present the group $G' = \zz^T$ in terms of generators and relations.  Here, this is trivial; the single generator ${\cal T}$ obeys the relation ${\cal T}^2 = 1$.  Next, we consider a general $\zz$ projective representation giving the action of $G'$ on a $m$ particle.  The generator is now written ${\cal T}^m$, and the relation becomes
\begin{equation}
({\cal T}^m)^2 = \sigma^m_T \text{,}  \label{eqn:ex1-z2-proj-rep}
\end{equation}
for $\sigma^m_T \in \zz$.  
We are allowed to redefine ${\cal T}^m \to - {\cal T}^m$, but this does not affect $\sigma^m_T$.  Therefore, because $\sigma^m_T = \pm 1$ is invariant under projective transformations, we can tentatively conclude that it labels two distinct fractionalization classes $[\omega_m]_{\zz}$.  To be sure this conclusion is correct, we need to check that each choice of $\sigma^m_T$ in fact corresponds to a factor set $\omega_m(g_1,g_2)$, for $g_1, g_2 \in \zz^T$. This can be accomplished by exhibiting a projective representation for each choice of $\sigma^m_T$.  In the present case, these representations are just familiar Kramers singlets ($\sigma^m_T = 1$) and doublets ($\sigma^m_T = -1$).  Therefore, $H^2(G', \zz) = \zz$, with $\sigma^m_T = \pm 1$ explicitly parametrizing the cohomology group, and labeling the distinct fractionalization classes $[\omega_m]_{\zz}$.

Next, we have to compute $H^2_t(G', {\rm U}(1))$.  To do this, we consider a general $t$-twisted ${\rm U}(1)$ projective representation of $G' = \zz^T$, again in terms of generators and relations.  We denote the generator by ${\cal T}^t$.  We also have to specify the function $t(g)$; it is sufficient to give the values of $t$ for the generators, and in this case, $t({\cal T}) = -1$.  The relation becomes
\begin{equation}
({\cal T}^t)^2 = \alpha_T \text{.}  \label{eqn:ex1-T-relation}
\end{equation}
Here, $\alpha_T$ is short-hand for the gauge transformation ${\cal G}[\alpha_T]$, with $\alpha_T \in {\rm U}(1)$.  So, for example, we can write  ${\cal T}^t \alpha_T = \alpha^{-1}_T {\cal T}^t$.  It is important to note that we can adjust the phase of the generator, by redefining ${\cal T}^t \to \lambda {\cal T}^t$, but this leaves $\alpha_T$ unchanged.

Because $\alpha_T$ is invariant under projective transformations, it is tempting to conclude that $\alpha_T \in {\rm U}(1)$ labels distinct equivalence classes $[\omega]_{{\rm U}(1)} \in H^2_t(G', {\rm U}(1))$.  However, this conclusion is not correct, because the possible values of $\alpha_T$ are constrained.  That is, there does not exist a $t$-twisted ${\rm U}(1)$ factor set for arbitrary $\alpha_T \in {\rm U}(1)$.  To see this, we conjugate both sides of Eq.~(\ref{eqn:ex1-T-relation}) by ${\cal T}^t$, and readily obtain $({\cal T}^t)^2 = \alpha^{-1}_T$.  This is consistent only if $\alpha_T \in \zz$.

As before, we need to verify that both choices $\alpha_T = \pm 1$ actually give rise to $t$-twisted ${\rm U}(1)$ factor sets.  The same Kramers singlet and doublet representations can be viewed as $t$-twisted ${\rm U}(1)$ projective representations, so, once again, we can exhibit a representation realizing each choice of $\alpha_T$.  Therefore, $H^2_t(G', {\rm U}(1)) = \zz$, which is explicitly parametrized by $\alpha_T \in \zz$.

To find the map $\rho_2$ giving $[\omega_m]_{{\rm U}(1)}$ in terms of $[\omega_m]_{\zz}$, suppose we have a $\zz$ projective representation as described in Eq.~(\ref{eqn:ex1-z2-proj-rep}) with some value of $\sigma^m_T$.  This $\zz$ projective representation can immediately be viewed as a $t$-twisted ${\rm U}(1)$ projective representation, with $\alpha_T = \rho_2(\sigma^m_T) = \sigma^m_T$.  In this case, then, $\rho_2 : \zz \to \zz$ is the identity map; more non-trivial examples will arise for other symmetries.

To conclude, we see that $\sigma^m_T = 1$ is anomaly-negative, because $\rho_2(1) = 1$.  On the other hand, $\sigma^m_T = -1$ is anomalous, because $\rho_2(-1) = -1$, which is not continuously connected to the identity element in $H^2_t(G', {\rm U}(1)) = \zz$.  The group ${\cal N}$ of anomaly-negative vison fractionalization classes is thus trivial, ${\cal N} = \z_1$.

\subsection{$G = {\rm U}(1) \times \zz^P$}
\label{sec:ff-z2p}

Next, we consider the case of a single lattice reflection symmetry ($\zz^P$), which commutes with the ${\rm U}(1)$.  Of course, any realistic system with reflection symmetry will also have a larger space group, including translation symmetry.  A physically reasonable viewpoint  is to imagine that we are interested in a system that has additional space group symmetry beyond $\zz^P$, but we are ``forgetting'' about the rest of the symmetry, and only making use of a ${\rm U}(1) \times \zz^P$ subgroup in our analysis.

Our discussion parallels the treatment given above for time reversal symmetry.  The group $G' = \zz^P$ is generated by $P$, which obeys the relation $P^2 = 1$.  Acting on $m$ particles, the generator is written $P^m$ and obeys
\begin{equation}
(P^m)^2 = \sigma^m_P \text{,}
\end{equation}
with $\sigma^m_P \in \zz$.  As before, $\sigma^m_P$ is invariant under $P^m \to - P^m$.   Both choices of $\sigma^m_P$ can be realized; for example, we can choose one-dimensional representations $P^m = 1$ (for $\sigma^m_P = 1$) and $P^m = i$ (for $\sigma^m_P = -1$).  Therefore, $H^2(G', \zz) = \zz$, parametrized by $\sigma^m_P$. A physical consequence of $\sigma^m_P$ is that, if a pair of $m$ particles are created and moved to reflection symmetric points, the resulting wave function has a reflection eigenvalue of $\sigma^m_P$, relative to the reflection eigenvalue of the ground state.

Next, we consider a general $t$-twisted ${\rm U}(1)$ projective representation generated by $P^t$, obeying the relation
\begin{equation}
(P^t)^2 = \alpha_P \text{,}
\end{equation}
with $\alpha_P \in {\rm U}(1)$.  In Appendix~\ref{app:Omega-is-boson}, it is shown that $P$ maps the symmetry flux $\Omega$ to $\Omega^{2n-1}$; since $P$ is unitary, this implies $t(P) = -1$.

At this point, the analysis is mathematically identical to that given for $G = {\rm U}(1) \times \zz^T$ symmetry in Sec.~\ref{sec:ex1}.  That is, we have $\alpha_P \in \zz$, and $H^2_t(G', {\rm U}(1)) = \zz$.  In addition, $\alpha_P = \rho_2(\sigma^m_P) = \sigma^m_P$.  Therefore, $\sigma^m_P = 1$ is anomaly-negative, while $\sigma^m_P = -1$ is anomalous, and the group of anomaly-negative vison fractionalization classes is ${\cal N} = \z_1$.  Introducing notation similar to that used in Sec.~\ref{sec:simple} for time reversal symmetry, we have found that the symmetry fractionalization patterns $eCmP$ and $eCPmP$ are anomalous, where $P$ denotes an anyon for which $P^2 = -1$.

\subsection{$G = {\rm U}(1) \times pm$}
\label{subsec:pm}

\begin{figure}
\includegraphics[width=0.7\columnwidth]{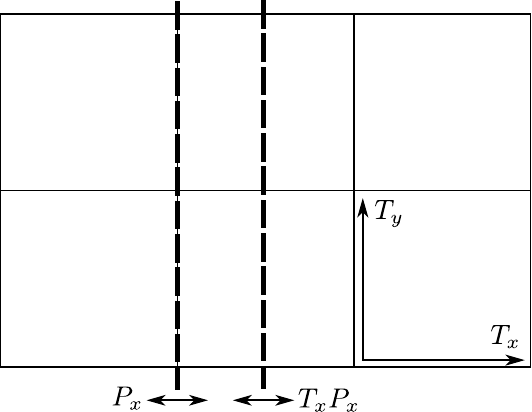}
\caption{Illustration of the operations generating the $d=2$ space group $pm$.  A square lattice, which is invariant under $pm$ symmetry, is shown to aid visualization.  $T_x$ and $T_y$ are translations by one lattice constant along the $x$- and $y$-axes, respectively. $P_x$ is a reflection, with axis indicated by the left-hand vertical dashed line.  The group $pm$ has two types of reflection axes, with $P_x$ being of one type, and $T_x P_x$ being of the other type.  The reflection axis for $T_x P_x$ is shown as the right-hand vertical dashed line.} 
\label{fig:pm-operations}
\end{figure}

We now move on to an example in which $G'$ is a $d=2$ space group.  We choose $G' = pm$, which is generated by translations $T_x, T_y, T_x^{-1}, T_y^{-1}$, and a reflection $P_x$; these operations are illustrated in Fig.~\ref{fig:pm-operations}, and obey the relations 
\begin{eqnarray}
T_x T_y T^{-1}_x T^{-1}_y &=& 1 \\
T_y P_x T^{-1}_y P_x &=& 1 \\
P_x^2 &=& 1 \\
T_x P_x T_x P_x &=& 1 \text{,}
\end{eqnarray}
which define the group $pm$.

Acting on $m$ particles, the general form of the relations is
\begin{eqnarray}
T^m_x T^m_y T^{m-1}_x T^{m-1}_y &=& \sigma^m_{txty} \\
T^m_y P^m_x T^{m-1}_y P^{m}_x &=& \sigma^m_{typx} \\
(P^m_x)^2 &=& \sigma^m_{px} \\
T^m_x P^m_x T^m_x P^m_x &=& \sigma^m_{txpx} \text{,}
\end{eqnarray}
with the $\sigma^m$'s taking values in $\zz$.  All the $\sigma^m$'s are invariant under projective transformations of the generators (\emph{e.g.} $T_x \to - T_x$), which suggests $H^2(G', \zz) = \zz^4$.  Because these relations are a subset of those used to present the square lattice space group in Ref.~\onlinecite{essin13}, it follows from Appendix A of that work that all $2^4$ possible choices of the $\sigma^m$'s indeed correspond to a factor set $\omega_m$, and indeed $H^2(G', \zz) = \zz^4$ (see also Appendix~\ref{subsec:pm-z2T} of this paper).  The fractionalization classes are thus parametrized by $[\omega_m]_{\zz} = (\sigma^m_{txty}, \sigma^m_{typx}, \sigma^m_{px}, \sigma^m_{txpx})$.

Now we need to compute $H^2_t(G', {\rm U}(1))$, noting that $t(P_x) = -1$, while $t(T_x) = t(T_y) = 1$.  The general form of the relations in a $t$-twisted projective representation is
\begin{eqnarray}
T^t_x T^t_y T^{t-1}_x T^{t-1}_y &=& \alpha_{txty} \label{eqn:example1-txtyeqn} \\
T^t_y P^t_x T^{t-1}_y P^{t}_x &=& \alpha_{typx} \\
(P^t_x)^2 &=& \alpha_{px} \label{eqn:example1-pxeqn} \\
T^t_x P^t_x T^t_x P^t_x &=& \alpha_{txpx} \label{eqn:example1-txpxeqn} \text{,}
\end{eqnarray}
where the $\alpha$'s take values in ${\rm U}(1)$.
If we redefine $T^t_y \to (\alpha_{typx})^{1/2} T^t_y$, this sends $\alpha_{typx} \to 1$, leaving the other $\alpha$'s unchanged.  In addition, the other $\alpha$'s are unchanged by redefinition of the other generators, so we have arrived at a convenient canonical gauge choice to describe a general $t$-twisted projective representation.

Next, conjugating Eq.~(\ref{eqn:example1-pxeqn}) by $P^t_x$, we find $\alpha_{px} \in \{ \pm 1 \}$, and similarly find $\alpha_{txpx} \in \{ \pm 1 \}$ by conjugating Eq.~(\ref{eqn:example1-txpxeqn}) by $T^t_x P^t_x$.  This suggests that $H^2_t(G', {\rm U}(1)) = {\rm U}(1) \times \zz \times \zz$, with the elements of the cohomology group parametrized by $[\omega]_{{\rm U}(1)} = (\alpha_{txty}, \alpha_{px}, \alpha_{txpx})$.

To verify this, we need to exhibit $t$-twisted projective representations that correspond to a generating set of $H^2_t(G', {\rm U}(1))$.  (It is enough to exhibit a generating set, because the corresponding factor sets can then be multiplied to obtain a factor set with arbitrary cohomology class.)   We introduce field annihilation (creation) operators $v_r$ ($v^\dagger_r$) for some fictitious particles residing on the sites $r = (x,y)$ of the square lattice.  Each $v_r$ is a two-component vector.  The generators are chosen to act on the field operators by
\begin{eqnarray}
T_x v_r T^{-1}_x &=& (\alpha_{txty})^{r_y} g_{tx} v_{r + \hat{x}}  \label{eqn:pm-txaction} \\
T_y v_r T^{-1}_y &=& g_{ty} v_{r + \hat{y}} \label{eqn:pm-tyaction} \\
P_x v_r P^{-1}_x &=& g_{px} v^\dagger_{P_x r}  \text{,} \label{eqn:pm-pxaction}
\end{eqnarray}
where $\alpha_{txty} \in {\rm U}(1)$, $P_x r = (-x,y)$, and $g_{tx}, g_{ty}, g_{px}$ are $2 \times 2$ unitary matrices.  Gauge transformations act on the field operators by
\begin{eqnarray}
{\cal G}[\lambda] v_r {\cal G}^{-1}[\lambda] &=& e^{i \lambda \sigma^z} v_r \\
{\cal G}[\lambda] v^\dagger_r {\cal G}^{-1}[\lambda] &=& e^{-i \lambda \sigma^z} v^\dagger_r \text{,}
\end{eqnarray}
for $\lambda \in {\rm U}(1)$, with $\sigma^z$ one of the $2 \times 2$ Pauli matrices, and where $v^\dagger_r$ denotes Hermitian conjugation of the components of $v_r$, but does not include transposition in the two-component space.  Choosing $g_{tx} = g_{ty} = g_{px} =1$ gives a continuous family of representations with $[\omega]_{{\rm U}(1)} = (\alpha_{txty},1,1)$.  Next, $\alpha_{txty} = 1$, $g_{ty} = i$, $g_{px} = i \sigma^y$, $g_{tx} = \sigma^z$ is a representation with $[\omega]_{{\rm U}(1)} = (1,-1,1)$.  (Here, again, $\sigma^x, \sigma^y, \sigma^z$ are the usual $2 \times 2$ Pauli matrices.)  Finally, $\alpha_{txty} = 1$, $g_{ty} = 1$, $g_{px} = \sigma^x$, $g_{tx} = \sigma^z$ has $[\omega]_{{\rm U}(1)} = (1,1,-1)$.  The factor sets of these three families of representations are a generating set for $H^2_t(G', {\rm U}(1)) = {\rm U}(1) \times \zz \times \zz$.

To find the map $\rho_2$, we begin with $[\omega_m]_{\zz} = (\sigma^m_{txty}, \sigma^m_{typx}, \sigma^m_{px}, \sigma^m_{txpx})$.  Viewing a corresponding projective representation as a   $t$-twisted ${\rm U}(1)$ projective representation, we can redefine $T^m_y \to (\sigma^m_{typx})^{1/2} T^m_y$, which sets $\sigma^m_{typx} \to 1$, and thus puts this projective representation in the canonical gauge described above.  Therefore, we have found
\begin{equation}
 (\alpha_{txty},\alpha_{px},\alpha_{txpx})  = \rho_2([\omega_m]_{\zz}) = (\sigma^m_{txty}, \sigma^m_{px}, \sigma^m_{txpx}) \text{.}
\end{equation}
This can be continuously deformed to the identity in $H^2_t(G', {\rm U}(1))$ if and only if $\sigma^m_{px} = \sigma^m_{txpx} = 1$.

Therefore, we have found that the fractionalization pattern $eC[\omega_e]m0[\omega_m]$, with $[\omega_m]_{{\zz}} = (\sigma^m_{txty}, \sigma^m_{typx}, \sigma^m_{px}, \sigma^m_{txpx})$, is anomalous if $\sigma^m_{px} = -1$ or $\sigma^m_{txpx} = -1$ (or both).  Equivalently, we can observe that anomaly-negative $[\omega_m]_{\zz}$ are given by $[\omega_m]_{{\zz}} = (\sigma^m_{txty}, \sigma^m_{typx},1,1)$, which form the subgroup ${\cal N} \subset H^2(G', \zz)$, with ${\cal N} = (\zz)^2$.

\subsection{$G = ({\rm U}(1) \rtimes \zz^T) \times p1$}
\label{subsec:p1}

The group $p1$ is the $d=2$ space group consisting only of translation symmetry.  Here, we consider this symmetry combined with time-reversal, which enters via the semi-direct product ${\rm U}(1) \rtimes \zz^T$.  This example is straightforward to analyze by following the steps in Sec.~\ref{subsec:pm} and Appendix~\ref{subsec:pm-z2T}, so we only quote the results.

We have $G' = p1 \times \zz^T$, which is generated by translations $T_x$, $T_y$, their inverses, and time reversal ${\cal T}$.  The relations are
\begin{eqnarray}
T_x T_y T^{-1}_x T^{-1}_y &=& 1 \\
{\cal T}^2 &=& 1 \\
{\cal T} T_x &=& T_x {\cal T} \\
{\cal T} T_y &=& T_y {\cal T} \text{.}
\end{eqnarray}
The $m$ particle symmetry fractionalization is specified by
\begin{eqnarray}
T_x^m T_y^m T^{m-1}_x T^{m-1}_y &=& \sigma^m_{txty} \\
({\cal T}^m)^2 &=& \sigma^m_T \\
{\cal T}^m T_x^m &=& \sigma^m_{Ttx} T_x^m {\cal T}^m \\
{\cal T}^m T_y^m &=& \sigma^m_{Tty} T_y^m {\cal T}^m  \text{,}
\end{eqnarray}
where the $\sigma^m$'s are $\zz$-valued phase factors.  The group of vison fractionalization classes is $H^2(G', \zz) = (\zz)^4$. The anomaly-negative fractionalization classes are those with $\sigma^m_{Ttx} = \sigma^m_{Tty} = 1$, and form the group ${\cal N} = (\zz)^2$.  We thus have an anomalous fractionalization pattern when $\sigma^m_{Ttx} = -1$, $\sigma^m_{Tty} = -1$, or both.  The disjoint sets of SPT phases distinguished by the anomaly test are labeled by elements of ${\cal S} = H^2(G', \zz) / {\cal N} = (\zz)^2$.

It is interesting to note that, in this case, we find anomalies involving time reversal that \emph{cannot} be understood in terms of ${\rm U}(1) \times \zz^T$ subgroups of $G$; these are type 3 anomalies as discussed in Sec.~\ref{sec:intro}.  These anomalies occur when one or more of $\sigma^m_{Ttx}$ or $\sigma^m_{Tty}$ are equal to $-1$.  It appears these anomalies cannot be understood in terms of dimensional reduction to $d=1$ SPT phases, as explained in Sec.~\ref{sec:1dspt}.

\section{Dimensional reduction viewpoint}
\label{sec:1dspt}

For type 1 and 2 anomalies (see Sec.~\ref{sec:intro}), which includes the examples described in Sections~\ref{sec:ex1} through~\ref{subsec:pm}, the flux-fusion anomaly test can be understood in terms of dimensional reduction to $d=1$ SPT phases. This viewpoint provides a different way of applying the anomaly test, which does not depend on some of the formalism introduced in Sec.~\ref{sec:ft-general}.  In particular, for the discussion below, we do not need the description of symmetry action on $\z_n$ fluxes $\Omega$ presented in Sec.~\ref{subsec:anomaly-test}.  Our discussion makes significant use of the results of Ref.~\onlinecite{zaletel15}, especially for the case of reflection symmetry.

We imagine putting the un-gauged $d=2$ SET phase with $\zz$ topological order on a cylinder with large but finite circumference, in such a way that the symmetry is preserved. The longitudinal dimension of the cylinder remains infinite.  We can approach the limit of an infinitely long cylinder either
via finite-length cylinders with open boundary conditions, or those with their ends identified periodically.
The minimally entangled states (MES) of the SET phase are those with definite anyon flux threaded along the cylinder (see Ref.~\onlinecite{zaletel15} for a more complete discussion). The cylinder is a $d=1$ system, gapped and symmetric. Therefore each of the MES is in a $d=1$ SPT phase.

\begin{figure}
\includegraphics[width=0.65\columnwidth]{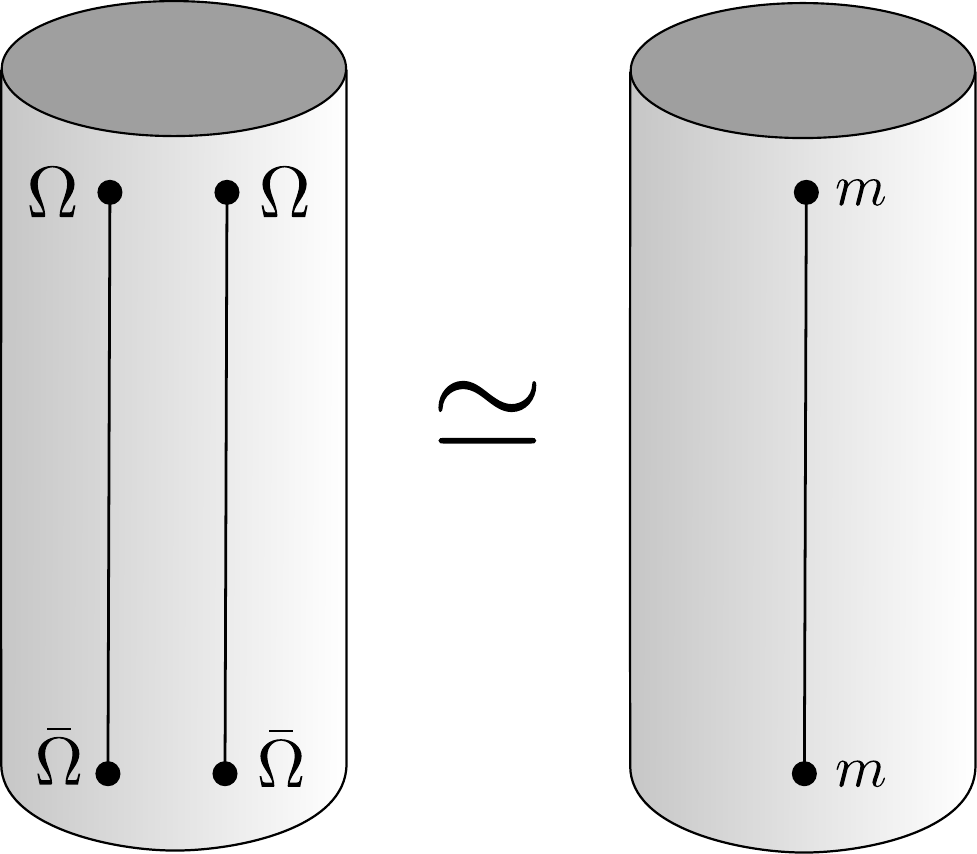}
\caption{Illustration of the flux-fusion anomaly test using dimensional reduction to a $d=1$ cylinder.  The $\zz \subset {\rm U}(1)$ flux $\Omega$ is threaded twice along the left-hand cylinder, while $m$ is threaded along the right-hand cylinder.  Due to the fusion rule $\Omega^2 = m$, these two systems are in the same $d=1$ SPT phase.}
\label{fig:cylinders}
\end{figure}

Starting from any of the MES, we imagine threading an anyon $a$ along the cylinder (see Fig.~\ref{fig:cylinders}).  This can be accomplished by a process creating an $a$-$\bar{a}$ pair in the bulk ($\bar{a}$ is the anti-particle of $a$), and then separating $a$ and $\bar{a}$ to infinity.  This maps the initial MES to a different MES, and so can be thought of as a mapping of $d=1$ SPT phases.    In Ref.~\onlinecite{zaletel15}, it was argued that this mapping between SPT phases only depends on the nature of symmetry action on $a$ in the original $d=2$ SET phase.

Now we consider the case of $G = {\rm U}(1) \times \zz^T$ symmetry, making our usual assumptions that the $e$-particle ($m$-particle) carries half-odd-integer (integer) ${\rm U}(1)$ charge.  Starting with some MES, we thread $m$ along the cylinder.  We use only the $\zz^T$ symmetry to analyze the dimensionally reduced $d=1$ SPT phases, so that there are two phases distinguished by a $\zz$ invariant, that corresponds to the presence or absence of Kramers doublet end states at the open boundaries \cite{gu09,pollmann12}.  If the $m$ is a Kramers singlet [$({\cal T}^m)^2 = 1$], then threading it along the cylinder leaves this $\zz$ invariant unchanged.  On the other hand, if $m$ is a Kramers doublet [$({\cal T}^m)^2 = -1$], threading it along the cylinder flips the $\zz$ invariant.  

To apply the anomaly test via dimensional reduction, we next introduce fluxes of the $\zz \subset {\rm U}(1)$ symmetry.  Note that we do not fully gauge the $\zz$ symmetry; it is enough to consider static flux defects of this symmetry, without introducing a dynamical gauge field.  We can thread the flux $\Omega$ along the cylinder, which amounts to introducing a flux defect $\Omega$ near one end, and the corresponding anti-defect $\bar{\Omega}$ near the other end.  Because the time reversal symmetry maps $\Omega \mapsto \Omega$ (as shown in Appendix~\ref{app:Omega-is-boson}), this can be done while preserving $\zz^T$, and threading $\Omega$ gives another map between $d=1$ SPT phases.  This map either flips the $\zz$ invariant or leaves it the same.

Finally, because $\Omega^2 = m$, threading $\Omega$ twice along the cylinder is the same as threading $m$, as illustrated in Fig.~\ref{fig:cylinders}.  But threading $\Omega$ twice must leave the $\zz$ SPT invariant unchanged, which means that threading $m$ must also leave this invariant unchanged.  This recovers the result that $m$ being a Kramers doublet is anomalous \footnote{This argument, and the corresponding argument for reflection symmetry, relies on the expectation that the effect of threading a flux or anyon on the $d=1$ SPT invariant does not depend on the MES one starts from before threading.  It can be argued this expectation holds by first noting the circumference of the cylinder can be taken as large as desired, so that we can think of the state as two-dimensional, even though ultimately we view it as a $d=1$ system to compute the SPT invariant.  Then two fluxes or anyons can be threaded along the cylinder in regions that are separated by a parametrically large distance, which corresponds to acting in these regions with string operators transporting the fluxes or anyons.  For $d=1$ SPT invariants of internal symmetry (such as time reversal), making an entanglement cut across a string exposes an anyon (or flux) that transforms projectively.  If we act with two well-separated strings and make an entanglement cut, the resulting states will involve a tensor product of two projective representations, which corresponds to adding SPT invariants in the usual way.  For reflection symmetry, the $d=1$ SPT invariant is simply the reflection eigenvalue of the string operator \cite{zaletel15}, and these eigenvalues add when well-separated strings act on the same MES.}.

In this argument, we used the fact that time reversal maps $\Omega \mapsto \Omega$.  It is important to note that this is not \emph{a priori} obvious simply because $\Omega$ is a $\pi$ flux of the ${\rm U}(1)$ symmetry, because $\Omega$ may be attached to an anyon under the action of symmetry,   as indeed occurs in the following example.  Instead, we need the considerations of Appendix~\ref{app:Omega-is-boson} to conclude $\Omega \mapsto \Omega$ under time reversal.

For $G = {\rm U}(1) \times \zz^P$ symmetry, a very similar discussion applies.  We choose the $\zz^P$ symmetry to exchange the two ends of the cylinder; that is, the $d=2$ reflection symmetry becomes $d=1$ reflection symmetry upon dimensional reduction.  SPT phases in $d=1$ protected by such $\zz^P$ symmetry are also characterized by a $\zz$ invariant \cite{gu09,pollmann12}.  Ref.~\onlinecite{zaletel15} argued that threading $m$ along the cylinder preserves this invariant if $(P^m)^2 = 1$, and flips the invariant if $(P^m)^2 = -1$.  Again, we consider the effect of threading a flux $\Omega$, of the $\zz \subset {\rm U}(1)$ symmetry, along the cylinder.  Because $\zz^P$ maps $\Omega$ to the anti-flux $\bar{\Omega} = \Omega^{3} = m \Omega$ (Appendix~\ref{app:Omega-is-boson}), this can be done while preserving the $\zz^P$ symmetry.  Therefore threading $\Omega$ either flips or preserves the $\zz$ SPT invariant.  At this point the discussion proceeds identically to the case of time reversal above, and we find that the symmetry fractionalization pattern with $(P^m)^2 = -1$ is anomalous.

This discussion also applies directly to the case of $G = {\rm U}(1) \times pm$ symmetry, because the anomalous symmetry fractionalization patterns found in Sec.~\ref{subsec:pm} are associated with two different $\zz^P$ subgroups of $pm$.  One of these is generated by $P_x$, and the other is generated by $T_x P_x$.  

Using the approach of Sec.~\ref{sec:ft-general}, we also find anomalies associated with the interplay between time reversal and other symmetries, that apparently cannot be understood from the dimensional reduction point of view.  These were designated type 3 anomalies in Sec.~\ref{sec:intro}, and arise when time reversal forms a semi-direct product with ${\rm U}(1)$ [\emph{i.e.} ${\rm U}(1) \rtimes \zz^T \subset G$].  They occur in the examples $G = ({\rm U}(1) \rtimes \zz^T) \times p1$ (Sec.~\ref{subsec:p1}) and $G = ({\rm U}(1) \rtimes \zz^T) \times pm$, $G = ({\rm U}(1) \rtimes \zz^T) \times p4mm$ (Appendix~\ref{sec:more-examples}).
For example, in each of these cases, $T_x^m$  (${\cal T}^m$) gives the action of translation in the $x$-direction (time reversal) on $m$-particles, and these generators obey the relation
\begin{equation}
{\cal T}^m T^m_x = \sigma^m_{Ttx} T^m_x {\cal T}^m \text{,}
\end{equation}
where we find that $\sigma^m_{Ttx} = -1$ is anomalous.  If we try to apply dimensional reduction here, we observe that time reversal maps the flux $\Omega$ to the anti-flux $\bar{\Omega}$, without exchanging the two ends of the cylinder.  This means that the state obtained upon threading $\Omega$ breaks time reversal, and is thus not a $d=1$ SPT phase preserving the symmetries involved in the anomaly.

\section{Bosonic topological crystalline insulators}
\label{sec:btci}

Here, we use the results from the flux-fusion anomaly test to identify and distinguish some $d=3$ bosonic TCIs.  We focus on the examples of $G = {\rm U}(1) \times \zz^P$ and $G = {\rm U}(1) \times pm$ symmetry discussed above in Sec.~\ref{sec:examples}; the latter example is sufficiently complex to illustrate the corresponding general results.
 The discussion for SPT phases with $G = {\rm U}(1) \times \zz^T$ symmetry entirely parallels that given in Sec.~\ref{sec:z2p-tci}, where we thus also comment on that case.  The focus is on understanding the extent to which information obtained from the anomaly test can distinguish the TCI phases identified, without using further information, while also illustrating what additional information can be used to make finer distinctions among phases.

\subsection{$G = {\rm U}(1) \times \zz^P$}
\label{sec:z2p-tci}

In Sec.~\ref{sec:ff-z2p}, we found that the symmetry fractionalization patterns $eCmP$ and $eCPmP$ are anomalous.  A coupled layer construction (see Ref.~\onlinecite{cwang13} and Appendix~\ref{app:coupled-layer}) shows that these fractionalization patterns -- indeed, \emph{any} fractionalization pattern -- can be realized as a surface of a $d=3$ SPT phase, or bosonic TCI.  Because each fractionalization pattern is anomalous, its corresponding SPT phase must be non-trivial.

Having established that the $eCmP$ and $eCPmP$ SPT phases are non-trivial, we would like to understand whether these phases are distinct from one another.  Moreover, as is well-known, SPT phases can be added together by combining together two decoupled systems, and observing that the combined system thus obtained is also an SPT phase.  This operation is expected to form an Abelian group.  We would also like to know how the $eCmP$ and $eCPmP$ SPT phases behave under this addition operation.

Let us consider adding together two copies of the $eCmP$ SPT phase.  This results in a surface with two decoupled surface SET phases, which we denote by $eCmP \oplus eCmP$.  Denoting the anyons in one SET phase by $e_1, m_1$, and in the other by $e_2, m_2$, we consider the result of condensing $e_1 e_2$ and $m_1 m_2$.  This condensation destroys the topological order, since all anyons of the $eCmP \oplus eCmP$ surface are either condensed, or are confined due to non-trivial mutual statistics with the condensate.  In addition, both the particles condensed have integer charge and $P^2 = 1$, so that they can be condensed without breaking any symmetries.  Therefore, we have obtained a gapped, symmetry-preserving, trivial surface, and the resulting SPT phase is the trivial phase.

The same conclusion clearly holds for the $eCPmP$ SPT phase.  Indeed, the conclusion holds for any $SPT$ phase with surface $\zz$ topological order (at least as long as the symmetry does not permute the anyon species).  To find SPT phases with order higher than two under addition, we would need to consider other types of topological order (\emph{e.g.} $\z_n$ topological order), or potentially those with symmetries permuting the anyon species.

Next, we consider adding together the $eCmP$ and $eCPmP$ SPT phases, obtaining a $eCmP \oplus eCPmP$ surface.  This surface can be simplified by condensing $m_1 m_2$, which again can be done without breaking any symmetries.  This results in a new surface SET phase with $\zz$ topological order, with anyons $e,m$, given in terms of the original anyons by $e = e_1 e_2$ and $m = m_1 \simeq m_2$.  Therefore, the fractionalization pattern after condensing $m_1 m_2$ is $ePmP$.  

The flux-fusion anomaly test provides no information about $ePmP$, so without additional information we cannot draw any further conclusions.  It has been shown via other methods that $ePmP$ is anomalous \cite{yqi15b, hsong16}, so the two SPT phases are different.

The flux-fusion anomaly test on its own allows us to distinguish a pair of SPT phases.  We can take this pair to be either the trivial phase and the $eCmP$ phase, or the trivial phase and the $eCPmP$ phase.  Both pairs form a $\zz$ subgroup group under addition of SPT phases; this is the result appearing in the fourth column of Table~\ref{tab:summary}.

The same discussion holds for $G = {\rm U}(1) \times \zz^T$, replacing ``$P$'' by ``$T$'' everywhere, so that we consider the non-trivial SPT phases with $eCmT$ and $eCTmT$ surfaces.  In this case, we note that $eTmT$ has also been argued to be anomalous \cite{vishwanath13,cwang13}.

\subsection{General results and $G = {\rm U}(1) \times pm$}
\label{sec:pm-tci}

We now consider bosonic TCIs for the general case of $G = {\rm U}(1) \rtimes G'$ symmetry,  using the example of  $G = {\rm U}(1) \times pm$ to illustrate the discussion.  First, we make some fixed but arbitrary choice for the fractionalization class of the $e$ particle, $[\omega_e]_{\zz} \in H^2(G', \zz)$, and we consider fractionalization patterns of the form
\begin{equation}
F^i = eC[\omega_e]m0[\omega^i_m] \text{,} \label{eqn:btci-fpatterns}
\end{equation}
where we have introduced an index $i$ to label the distinct vison fractionalization classes $[\omega^i_m]_{\zz} \in H^2(G', \zz)$.  Each fractionalization pattern $F^i$ corresponds to a bosonic TCI ($d=3$ SPT phase), for which it describes a surface SET phase (see Appendix~\ref{app:coupled-layer}).  In general, not all the $F^i$'s correspond to distinct or non-trivial SPT phases.  Formally, it will be convenient to refer to a map $\varphi : \{ F^i \} \to {\cal G}_{SPT}$, from the set of fractionalization patterns described in Eq.~(\ref{eqn:btci-fpatterns}), to the (Abelian) group of distinct $d=3$ SPT phases of the given symmetry, which we denote by ${\cal G}_{SPT}$.

Adding together the SPT phases corresponding to $F_i$ and $F_j$ gives the surface SET phase $F_i \oplus F_j$.  Labeling the $F_i$ anyons by $e_1, m_1$, and the $F_j$ anyons by $e_2, m_2$, this surface theory can be simplified by condensing $e_1 e_2$, which can be done without breaking symmetry, because both $e$ particles have the same fractionalization class.  The resulting surface theory has $\zz$ topological order, with anyons $e = e_1 \simeq e_2$ and $m = m_1 m_2$, so that we have
\begin{equation}
F_i \oplus F_j \simeq eC[\omega_e]m0[ \omega^i_m \omega^j_m] \text{.}
\end{equation}
We thus see that the fractionalization patterns add according to the same operation governing multiplication of $[\omega^i_m]_{\zz}$ in $H^2(G', \zz)$, and therefore the  set $\{ F^i \}$ can be viewed as a group isomorphic to $H^2(G', \zz)$.  In addition, this shows that the map $\varphi$ can be viewed as a group homomorphism $\varphi : H^2(G', \zz) \to {\cal G}_{SPT}$.

We know the group $H^2(G', \zz)$, and the flux-fusion anomaly test gives us some knowledge about the map $\varphi$.  The goal is to use this information to learn as much as possible about ${\cal G}_{SPT}$.   If $F^i$ is a fractionalization pattern known to be anomalous, then $\varphi(F^i) \neq 1$; that is, the corresponding SPT phase is non-trivial.  In the example $G = {\rm U}(1) \times pm$, recall that $[\omega_m]_{\zz} = (\sigma^m_{txty}, \sigma^m_{typx}, \sigma^m_{px}, \sigma^m_{txpx})$, and the corresponding SPT phase is non-trivial whenever $\sigma^m_{px} = -1$, or $\sigma^m_{txpx} = -1$, or both.

In addition, we would ideally like to know which fractionalization patterns are non-anomalous, and thus map to the trivial SPT phase.  Formally, the set of such fractionalization patterns is the kernel of $\varphi$, and is a subgroup of $H^2(G', \zz)$ denoted by $\operatorname{Ker} \varphi$.  In general, we do not know $\operatorname{Ker} \varphi$.  However, we do know which fractionalization patterns are anomaly-negative, these also form a subgroup denoted ${\cal N} \subset H^2(G', \zz)$.  In the present example, ${\cal N} = \zz \times \zz$ consists of those fractionalization classes of the form $[\omega_m]_{\zz} = (\sigma^m_{txty}, \sigma^m_{typx}, 1, 1)$.  In general, we have $\operatorname{Ker} \varphi \subset {\cal N} \subset H^2(G', \zz)$; that is, non-anomalous fractionalization patterns are a subgroup of anomaly-negative ones.

Now, we consider the group ${\cal S} = H^2(G', \zz) / {\cal N}$, elements of which are cosets of ${\cal N}$.  We will see that there are at least as many distinct SPT phases as there are elements of ${\cal S}$.  From the fact $\operatorname{Ker} \varphi \subset {\cal N} \subset H^2(G', \zz)$, it follows immediately that distinct elements of ${\cal S}$ map to disjoint sets of SPT phases in ${\cal G}_{SPT}$.\footnote{In detail, suppose that $F_i$ and $F_j$ are two fractionalization patterns corresponding to the same SPT phase, so that $\varphi(F_i) = \varphi(F_j)$.  Therefore, $F_i = F_j \oplus F_k$, where $F_k \in \operatorname{Ker} \varphi$.  Since $\operatorname{Ker} \varphi \subset {\cal N}$, we have $F_k \in {\cal N}$, and therefore $F_i$ and $F_j$ belong to the same coset of ${\cal N}$.  It follows that, if $F_i$ and $F_j$ belong to different cosets of ${\cal N}$, they must correspond to different SPT phases.}  The disjoint sets of SPT phases are thus labeled by elements of ${\cal S}$; this is the group that appears in the fourth column of Table~\ref{tab:summary}.

In the present example, ${\cal S} = \zz \times \zz$, and its elements $p_1, \dots, p_4$ are the four cosets
\begin{eqnarray}
p_1 &=& (1,1,1,1) \times {\cal N} \\
p_2 &=& (1,1,-1,1) \times {\cal N} \\
p_3 &=& (1,1,1,-1) \times {\cal N} \\
p_4 &=& (1,1,-1,-1) \times {\cal N} \text{.}
\end{eqnarray}
Each of these cosets corresponds to four different surface SET phases, depending on the element chosen from ${\cal N}$.  Surface SET phases belonging to the same coset may or may not correspond to the same SPT phase, but surface SET phases belonging to different cosets correspond to different SPT phases.  Therefore, in this example, there are at least four bosonic TCIs (one of which is trivial).

To obtain additional information, we need to determine $\operatorname{Ker} \varphi \subset {\cal N}$.  For example, in some cases it may be true that $\operatorname{Ker} \varphi = {\cal N}$, if all the anomaly-negative fractionalization patterns are in fact non-anomalous.  The number of distinct SPT phases obtained from each coset of ${\cal N}$ is $| {\cal N} | / | \operatorname{Ker} \varphi |$.  

Throughout this discussion, we have fixed the $e$ particle fractionalization class, but it is also natural to consider fractionalization pattens with different $e$ particle fractionalization classes, as we did for the case of $G = {\rm U}(1) \times \zz^P$ symmetry in Sec.~\ref{sec:z2p-tci}.
Suppose we add together $F^1 = eC[\omega^1_e]m0[\omega^1_m]$ and $F^2 = eC[\omega^2_e]m0[\omega^2_m]$, to obtain a $F^1 \oplus F^2$ surface, where $[\omega^1_e]_{\zz} \neq [\omega^2_e]_{\zz}$.  If it happens that $[\omega^1_m]_{\zz} = [\omega^2_m]_{\zz}$, we can condense $m_1 m_2$ to obtain a surface SET phase with fractionalization pattern $e0[\omega^1_e \omega^2_e]m0[\omega^1_m ]$.  Here, none of the anyons carry fractional ${\rm U}(1)$ charge.  There is not a general understanding of which such fractionalization patterns are anomalous.  However, many such patterns can be explicitly constructed strictly in $d=2$, using, for instance, exactly solvable models or parton gauge theory \cite{hsong15}.  This has been done for square lattice space group symmetry using exactly solvable models \cite{hsong15}, and could be done for other symmetry groups as needed. In addition, in the case of reflection symmetry, the $ePmP$ fractionalization pattern is anomalous \cite{yqi15b,hsong16}.  Results along these lines can thus be used to obtain further information on bosonic TCI phases, which we leave for future work.

\section{Anomalous Superfluids}
\label{sec:anomalous-superfluids}

Our results on anomalous symmetry fractionalization patterns can also be used to obtain anomalous surface superfluid states of $d=3$ bosonic TCIs.  The anomalous nature of these superfluids arises from the symmetry fractionalization of vortices, which transform projectively under $G'$ symmetry.  It is particularly useful to identify such anomalous superfluids, because it is straightforward to proceed from their formal description to explicit field theories, which can be used to describe not only the surface superfluid phase, but also nearby surface phases and phase transitions.

A related prior work is  Ref.~\onlinecite{vishwanath13}, which studied $d=3$ bosonic topological insulators (with ${\rm U}(1)$ and time reversal symmetry) by constructing field theories for anomalous surface superfluids.  Some of those superfluids are  characterized by non-trivial vortex symmetry fractionalization, and were argued to be anomalous based on the impossibility of realizing a trivial, gapped surface in an explicit dual vortex field theory for the surface.  Our results are complementary, allowing one to establish that some dual vortex field theories for $d=2$ superfluids are anomalous \emph{without} a detailed and potentially subtle analysis of possible phases.

As usual, we consider a surface SET phase with $\zz$ topological order and symmetry fractionalization pattern $eC[\omega_e]m0[\omega_m]$, but making the additional assumption that $[\omega_e]_{\zz} = 1$; that is, the $e$ particle transforms trivially under $G'$.  It is therefore possible to condense an $e$ particle carrying ${\rm U}(1)$ charge $1/2$ and obtain a superfluid, with spontaneously broken ${\rm U}(1)$ symmetry, where $G'$ symmetry is preserved.  Under these circumstances, the vison of the SET phase becomes the elementary $2\pi$-vortex of the superfluid \cite{senthil00}, so the vortex thus inherits the $G'$ transformation properties of the $m$ particle.  If we start with an anomalous surface SET phase, the resulting surface superfluid must also be anomalous, because both are surfaces of the same non-trivial SPT phase.  

It is well-known that vortices can transform projectively under symmetry \cite{senthil04a,senthil04b,balents05}.  This can be seen conveniently in the dual description of a superfluid, where vortices carry the charge of a non-compact ${\rm U}(1)$ gauge field, for which the photon is nothing but the superfluid sound mode.  Symmetry operations acting on vortices can thus be augmented by ${\rm U}(1)$ gauge transformations, and the symmetry acts projectively.  In fact, in more detail, vortices transform as a $t$-twisted ${\rm U}(1)$ projective representation of $G'$, and vortex fractionalization classes are thus elements of $H^2_t(G', {\rm U}(1))$.  This can be seen by introducing field operators for the vortices, as was done in Sec.~\ref{subsec:anomaly-test} for the $\z_n$ flux $\Omega$.  Here, the field operators are labeled by an integer, which is simply the vorticity, and we have ${\rm U}(1)$ gauge transformations rather than $\z_{2n}$ gauge transformations.  Otherwise, the discussion entirely parallels that given in Sec.~\ref{subsec:anomaly-test}.

Because a $m$ particle becomes a vortex upon entering the superfluid phase, the vortex fractionalization class $[\omega_v] \in H^2_t(G', {\rm U}(1) )$ is given by
\begin{equation}
[\omega_v] = [\omega_m]_{{\rm U}(1)} = \rho_2( [\omega_m]_{\zz} ) \text{.}  \label{eqn:vf-mf}
\end{equation}
Remarkably, $[\omega_m]_{{\rm U}(1)}$, which provides a simple mathematical characterization of which fractionalization patterns are anomaly-negative, also has direct physical meaning as the fractionalization class of vortices in the superfluid phase.  This allows us to establish that superfluids with certain vortex fractionalization classes are anomalous.

This conclusion is bolstered by proceeding in the reverse direction; that is, we can start with a superfluid, and condense pairs of vortices to obtain a SET phase with $\zz$ topological order.  This can be done without breaking symmetry as long as the vortex fractionalization class satisfies $[\omega_v]^2 = 1$, so that vortex pairs transform trivially.  The fractionalization class of the $m$ particle must satisfy Eq.~(\ref{eqn:vf-mf}), but we note this does not completely determine $[\omega_m]_{\zz}$ given $[\omega_v]$.  We expect that, given $[\omega_v]$, the different possible choices of $[\omega_m]_{\zz}$ correspond to inequivalent condensates of paired vortices; detailed study of this point is left for future work.

Which vortex fractionalization classes are anomalous?  We answer this question for the example of $G = {\rm U}(1) \times pm$ symmetry, and then make some comments on the answer more generally.  First, Eq.~(\ref{eqn:vf-mf}) implies that $[\omega_v] = (\alpha_{txty} = \pm 1, \alpha_{px}, \alpha_{txpx})$ is anomalous whenever $\alpha_{px} = -1$, $\alpha_{txpx} = -1$, or both, because these $[\omega_v]$ are obtained from anomalous $[\omega_m]_{\zz}$.

We can also find more anomalous vortex fractionalization classes, by starting with an anomalous superfluid, and adding a layer of $d=2$ superfluid.  This can be done by first assuming that each layer has an independent ${\rm U}(1)$ symmetry, and then breaking the resulting ${\rm U}(1) \times {\rm U}(1)$ down to ${\rm U}(1)$; that is, we allow unit charge excitations to tunnel between the two layers.  Before breaking the ${\rm U}(1) \times {\rm U}(1)$ symmetry, each layer has independent vortices, schematically labeled by $v_1$ and $v_2$.  After breaking the symmetry, $v_1$ and $v_2$ vortices are confined together, so that the new superfluid state has vortices $v = v_1 v_2$.  This results is a modified vortex fractionalization class $[\omega_v] = [\omega_{v_1}] [\omega_{v_2}]$.

In the present example, vortex fractionalization classes $[\omega_v] = (\alpha_{txty},1,1)$, with $\alpha_{txty}$ an arbitrary ${\rm U}(1)$ phase, can occur in $d=2$.  Writing $\alpha_{txty}  = e^{2\pi i \bar{n}}$, such superfluids occur for bosons on the square lattice at filling $\bar{n}$.  This is easily seen via straightforward application of boson-vortex duality to such a model; briefly, the vortices feel the background boson density as a magnetic flux of $2\pi \bar{n}$ per plaquette, and thus transform projectively under translation symmetry.  By adding layers of such non-anomalous superfluids, we can see that the only non-anomalous vortex fractionalization classes are $[\omega_v] = (\alpha_{txty},1,1)$, and all others are anomalous.  

This result can be stated in a more general fashion, namely, $[\omega_v]$ is non-anomalous if and only if it can be continuously deformed to the identity element of $H^2_t(G', {\rm U}(1))$.  We conjecture that this result holds for all symmetries $G = {\rm U}(1) \rtimes G'$, but we do not have a general argument, for two reasons.  First, we note that the discussion above relied on being able to find all non-anomalous vortex fractionalization classes for $G = {\rm U}(1) \times pm$ symmetry. Second, in this case, $H^2_t(G', {\rm U}(1))$ was a product of ${\rm U}(1)$ and $\zz$ factors, containing no $\z_n$ factors for $n > 2$.  If there were such $\z_n$ factors, some vortex fractionalization classes could not be obtained by condensing the $e$ particle of a SET phase with $\zz$ topological order.  It is likely that this could be handled by generalizing the flux-fusion approach to SET phases with $\z_n$ topological order, a problem that is left for future work.

We conclude this section, and illustrate the utility of the present results, by describing the construction of a field theory for the surface of a bosonic TCI with symmetry $G = {\rm U}(1) \times pm$. We work in a dual description of the surface superfluid, introducing a two-component complex field $\Phi_v$ for the superfluid vortices.  $\Phi_v$ carries unit charge of the ${\rm U}(1)$ gauge field $a_{\mu}$, in terms of which the global ${\rm U}(1)$ current is $j_{\mu} = \epsilon_{\mu \nu \lambda} \partial_{\nu} a_{\lambda} / 2 \pi$.  We work in Euclidean space time.  We choose the $pm$ symmetry generators to act on the vortices as follows:
\begin{eqnarray}
T_x, T_y : \Phi_v(x,y, \tau) &\to& \Phi_v(x,y,\tau) \\
P_x : \Phi_v(x,y,\tau) &\to& (i \sigma^y) \Phi_v(-x, y, \tau) \text{.}
\end{eqnarray}
As usual, we neglect the action of lattice translations $T_x$ and $T_y$ on the spatial position of the continuum field $\Phi_v$, as this only leads to subleading gradient terms.  The presence of the Pauli matrix $i \sigma^y$ in the action of $P_x$ implies that $P_x^2 = -1$ acting on $\Phi_v$, and we have the vortex symmetry fractionalization $[\omega_v] = (1,-1,-1)$.  This fractionalization class is anomalous, so we are not describing a $d=2$ superfluid, but rather the surface of a bosonic TCI.  We note that we chose $\Phi_v$ as a two-component field in order to realize this non-trivial fractionalization class.

The Lagrangian is
\begin{equation}
{\cal L} = | (\partial_{\mu} + i a_{\mu} ) \Phi_v |^2 + \frac{K_s}{2} \sum_{\mu} j_{\mu}^2 + V( \Phi_v ) + \cdots \text{.}
\end{equation}
Here, the first term is the vortex kinetic energy,  the second term controls the superfluid stiffness, and $V( \Phi_v)$ is a potential for the vortex field, whose form is dictated by gauge invariance and the action of the microscopic symmetries.  The ellipsis includes various other perturbations allowed by symmetry.  This field theory can be used to study the superfluid phase itself (where $\Phi_v$ is massive), neighboring phases described as condensates of $\Phi_v$ (which break lattice symmetries), surface SET phases where $\Phi_v^2$ is condensed, and transitions among these phases.

\section{Discussion and outlook}
\label{sec:discussion}

We introduced the flux-fusion anomaly test, a method to detect some anomalous symmetry fractionalization patterns in $d=2$ SET phases with $\zz$ topological order.  This constrains the possible physical properties of strictly $d=2$ SET phases, and is a step toward the full classification of such phases in the presence of crystalline symmetry.  In addition, the same results allow us to identify and distinguish some $d=3$ SPT phases via their surface theories, including some bosonic TCIs that have not previously been identified to our knowledge.  For some of the bosonic TCIs, we identified not only surface SET phases with anomalous symmetry fractionalization, but also anomalous surface superfluids distinguished by the projective symmetry transformations of vortex excitations.

We note that the flux-fusion anomaly test is closely related but not equivalent to a ``monopole tunneling'' approach developed in~\cite{metlitski13} and used in~\cite{cwang13}.  In this approach, one considers a $d=3$ SPT phase with ${\rm U}(1)$ symmetry, gauges the ${\rm U}(1)$ symmetry, and studies magnetic monopole excitations in the bulk.  If one has ${\rm U}(1) \times \zz^T$ symmetry, the monopole can be a Kramers doublet, indicating the bulk SPT phase is non-trivial. Tunneling a monopole into the bulk through a superfluid surface leaves  a vortex behind on the surface, which must also be a Kramers doublet.  Condensing double vortices on the surface leads to the $eCmT$ state, which is thus a surface SET of a non-trivial SPT phase.

Using the description of symmetry action on vortices given in Sec.~\ref{sec:anomalous-superfluids}, very similar reasoning can be applied to the symmetries considered in this paper, and the same anomalies we find can presumably be diagnosed.  However, the two approaches are not equivalent.  In particular, the flux-fusion method is more general as it can be applied for discrete symmetries, \emph{e.g.} $G = \z_n \rtimes G'$, as mentioned below.

In this paper, we focused primarily on symmetries of the form $G = {\rm U}(1) \times G_{{\rm space}}$ and $G = [ {\rm U}(1) \rtimes \zz^T ] \times G_{{\rm space}}$, where $G_{{\rm space}}$ is a $d=2$ space group.  The latter symmetry is particularly relevant for systems of bosons.  We did not consider the very important class of symmetries $G = {\rm U}(1) \times \zz^T \times G_{{\rm space}}$, which are relevant for spin systems with continuous spin rotation symmetry.  For example, when ${\rm U}(1) \subset {\rm SO}(3)$, these are the symmetry groups of Heisenberg spin models.  The reason for this omission is a surprising finding that complicates application of our anomaly test:  for these symmetries, it is sometimes impossible to gauge $\z_n \subset {\rm U}(1)$ without breaking $\zz^T \times G_{{\rm space}}$ symmetry, even in strictly two dimensions \cite{tcharge}.  This can occur when some lattice sites transform in a projective representation of the on-site ${\rm U}(1) \times \zz^T \subset G$ symmetry; for example, when ${\rm U}(1) \subset {\rm SO}(3)$, this means that there are $S = 1/2$ or other half-odd-integer spins in the system.  Naively, it would appear the anomaly test is less useful for these symmetries, but, remarkably, it turns out that this obstruction to gauging the symmetry makes the anomaly test significantly more powerful.  These results will be presented in a forthcoming work \cite{tcharge}.

More generally, to which symmetry groups does the flux-fusion anomaly test apply?  One point is that ${\rm U}(1)$ symmetry is not required, and it is simple to generalize the results of this paper to symmetries $G = \z_n \rtimes G'$ in a straightforward manner. This works as long as $n$ is even, so that we can sensibly choose $e$ to carry half-charge of $\z_n$, and as long as the $G'$ symmetry constrains the $\z_n$ flux to be a boson.  We note that, if $G = \zz \times G'$, the $\zz$ flux is always a boson, independent of $G'$ (see Sec.~\ref{subsec:gauging}).  In principle, we can also consider $G = G_o \rtimes G'$, where $G_o$ is some finite, on-site, unitary symmetry.  In practice, in the latter case, one generally obtains a non-Abelian gauged SET phase, which can be expected to increase the complexity of analysis required.

A related point is that the requirement that symmetry fluxes are bosons is not fundamental, but is rather imposed because it simplifies the analysis.  For on-site, unitary symmetries that do not permute the anyons of the gauged SET phase, we believe it likely this requirement plays no role and can simply be ignored.  More generally, one needs a description of the action of symmetry on non-bosonic anyons, which is subtle and not yet fully understood for crystalline symmetry \cite{essin13,ymlu14,yqi15b,zaletel15}.  However, we expect that the necessary theory will become available with further progress, in which case it can be applied to broaden the applicability of the flux-fusion anomaly test.

It is also interesting to consider generalizing the flux-fusion anomaly test to other topological orders.  The basic idea of the anomaly test is, given an action of symmetry on the anyons of an un-gauged SET phase, to determine whether this action can be extended consistently to symmetry fluxes.  This idea applies more generally to SET phases with topological orders and symmetries beyond those considered here, although we do not expect our detailed analysis to apply in general.  For on-site, discrete, unitary symmetries, the framework of $G$-crossed tensor category provides a comprehensive description of SET phases \cite{barkeshli14} and a systematic means of detecting anomalies \cite{barkeshli14,chen14}.  For  symmetries where both approaches apply, the flux-fusion anomaly test as developed in this paper is certainly less general than the $G$-crossed tensor category approach, but it has the advantage of identifying some anomalies in a physically intuitive way. Moreover, without the need to introduce fluxes for all symmetries, the flux-fusion approach can be easily applied to continuous, anti-unitary and spatial symmetry, as illustrated by the examples discussed in this paper.

The examples studied in this paper can be analyzed without resort to $G$-crossed tensor category theory due to the simplicity of the topological orders involved. The magnetic sectors of the $\mathbb{Z}_2$ topological order and the gauged $\mathbb{Z}_{2n}$ theory have trivial $F$ and $R$ matrices. Therefore, when analyzing their transformation under symmetry, we do not need to worry about the ``gauge transformation'' on fusion spaces, as discussed in Eq.~58 of Ref.~\onlinecite{barkeshli14}. This greatly simplifies the mathematical structure involved, and the flux fusion procedure as discussed in this paper can be implemented.

General SET phases can have nontrivial $F$ and $R$ matrices, and it is important to take  ``gauge transformations'' into account when analyzing symmetry action. To avoid this complexity, we can restrict to the case where the symmetry flux $\Omega_h$ (for $h \in G$) remains invariant under the symmetry action of $g \in G$. That is, we require (1) $h$ commutes with $g$, so that $\Omega_h$ remains the same symmetry flux, and also (2) $\Omega_h$ remains in the same topological sector and is not attached to an anyon under the action of $g$. The second condition can be violated when $g$ and $h$ act non-commutatively on the anyons. For example, in the projective semion example of Ref.~\onlinecite{chen14}, with $G = \zz \times \zz$ symmetry, the two $\zz$ symmetries anti-commute with each other on the semion, and the flux of one $\zz$ is glued to a semion under the action of the other $\zz$. The $F$ and $R$ matrices are non-trivial in this example, so we do not expect a straightforward generalization of the flux fusion method described in this paper to apply.

When the above two conditions are satisfied, symmetry $g$ has a local action on $\Omega_h$ and we can talk about the symmetry fractionalization of $g$ on $\Omega_h$ without worrying about  ``gauge transformations.''  Here we remark that $g$ and $h$ can be the same type of symmetry operation.  More precisely, in the main text we only discussed cases where $g$ and $h$ lie in two different factors of a semidirect product.  However,  this is not necessary, and the flux fusion idea can apply even when $g = h$.  We discuss such an example in a  study of $d=3$ SET phases \cite{3DSET}.

Beyond the anomaly test itself, one natural direction for further studies is to develop an understanding of the physical properties of the bosonic TCIs we have identified.  In light of prior work on bosonic topological insulators with ${\rm U}(1)$ and time reversal symmetry \cite{vishwanath13}, we expect that the surface dual vortex field theories discussed in Sec.~\ref{sec:anomalous-superfluids} will be particularly useful in this regard.  Along the same lines, it will be interesting to look for simple, physically reasonable models realizing bosonic TCIs.

We also hope that our results on bosonic TCIs will be useful as a stepping stone to identify and perhaps classify TCIs of interacting electrons.  As has been established for electronic topological insulators [with ${\rm U}(1)$ and time reversal symmetry], there are non-trivial electronic topological phases that can be understood by forming composite bosonic particles out of electrons (Cooper pairs, or spins), and putting these objects into a bosonic SPT phase \cite{cwang14}.  This is an important part of the classification of interacting electronic topological insulators given in Ref.~\onlinecite{cwang14}.

{\emph{Note added.}  During the review process of this paper, some closely related work has appeared.  In particular, Ref.~\onlinecite{yqi15c} extended the flux-fusion anomaly test to $\zz$ spin liquids with ${\rm SO}(3)$ spin rotation symmetry, and showed that the vison symmetry fractionalization in $S = 1/2$ Heisenberg models on square and kagome lattices is completely fixed.  References~\onlinecite{zaletel16} and~\onlinecite{cincio15} adapted and used flux-fusion to constrain the symmetry fractionalization of the chiral spin liquid phase of the kagome Heisenberg model.  Another related development is the work of Ref.~\onlinecite{hsong16}, which presented an approach to classify SPT phases protected by point group symmetry based on a kind of dimensional reduction; we anticipate this approach can be generalized to provide an alternate characterization and more complete classification of the bosonic TCIs identified here.

\acknowledgments{M.H. is grateful to T. Senthil for several conversations that helped inspire some of these ideas, and to Olexei Motrunich for a useful discussion.  X.C. would like to thank Meng Cheng for helpful discussions. We thank Zhenghan Wang for useful correspondence.   M.H. was supported by the U.S. Department of Energy, Office of Science, Basic Energy Sciences, under Award \# DE-FG02-10ER46686 and subsequently under Award \# DE-SC0014415,
and by Simons Foundation grant No.  305008 (sabbatical support). X.C. is supported by the Caltech Institute for Quantum Information and Matter and the Walter Burke Institute for Theoretical Physics. Publication of this article was funded in part by the University of Colorado Boulder Libraries Open Access Fund.}

\appendix

\section{Procedure to gauge $\z_n \subset {\rm U}(1)$ symmetry}
\label{app:gauging-procedure}

Here, we consider symmetry groups $G = ({\rm U}(1) \rtimes \zz^T) \times G_s$, where $G_s$ is a $d=2$ space group, and describe an explicit procedure to gauge the $\z_n \subset {\rm U}(1)$ symmetry.  In particular, we verify that this can be done while preserving $G' = \zz^T \times G_s$ symmetry.  We also discuss the case of $G = {\rm U}(1) \times \zz^T$, giving a procedure to gauge $\z_n \subset {\rm U}(1)$ while preserving $\zz^T$.  While these conclusions may appear obvious, they do not hold in general for other symmetry groups (in particular, for $G = {\rm U}(1) \times \zz^T \times G_s$).  This has interesting consequences that will be explored in a future publication \cite{tcharge}.

First we discuss the case $G = ({\rm U}(1) \rtimes \zz^T) \times G_s$.  We consider a bosonic model defined on a lattice with sites $r$, which is invariant under the space group symmetry $G_s$.  Each $g \in G_s$ acts on lattice sites, which we write formally as $r \mapsto g r$.  There is a Hilbert space ${\cal H}_r$ associated with each lattice site, and the full Hilbert space is the tensor product ${\cal H} = \otimes_r {\cal H}_r$.  Because the ${\rm U}(1)$ symmetry is on-site, for each lattice site $r$ there is a charge density operator $N_r$ with integer eigenvalues.  Because time reversal forms a semidirect product with the ${\rm U}(1)$, we have
\begin{equation}
{\cal T} N_r {\cal T}^{-1} = N_r \text{.}
\end{equation}
In general, we might wish to allow for a shift $N_r \to N_r + \delta N_r$ under time reversal.  But, since we assume the ground state is invariant under ${\cal T}$, we must have $\langle N_r \rangle = \langle N_r \rangle + \delta N_r$, and $\delta N_r = 0$.
Moreover, the space group operation $g \in G_s$ is represented by $U_g$, and acts on the charge density by
\begin{equation}
U_g N_r U^{-1}_g = N_{g r} \text{.}
\end{equation}

To gauge the $\z_{n} \subset {\rm U}(1)$ symmetry, we introduce $\z_n$ electric field and vector potential operators, that reside on oriented links $\ell = (r,r')$, where each link joins a pair of lattice sites $r$ and $r'$.  The set of links is chosen to make the lattice into a connected graph that is invariant under space group symmetry; for example, choosing links to join nearest-neighbor sites is sufficient in many cases.  

The electric field $e_\ell$ and vector potential $a_\ell$ act on the Hilbert space of link $\ell = (r,r')$, which we choose to be $n$-dimensional with basis $\{ |0 \rangle, |1 \rangle, \dots, | n-1 \rangle \}$.  The link operators are defined by
\begin{eqnarray}
a_{\ell} | k \rangle &=& \exp \Big( \frac{2 \pi i k}{n} \Big) | k \rangle \\
e_{\ell} | k \rangle &=& | (k + 1) \!\!\!\! \mod n \rangle \text{.}
\end{eqnarray}
These lattice vector fields are oriented, so that if $\bar{\ell} = (r', r)$ is $\ell$ with reversed orientation, then $e_{\bar{\ell}} = e^\dagger_{\ell}$ and $a_{\bar{\ell}} = a^\dagger_{\ell}$.

We impose the Gauss' law constraint
\begin{equation}
\prod_{\ell \sim r} e_{\ell} = \exp\Big[ \frac{2 \pi i}{n} N_r  \Big] \text{,}  \label{eqn:zn-gauss-law}
\end{equation}
where the product is over those links $\ell$ that join $r$ to other sites, with orientation pointing away from $r$.  Choosing time reversal and space group operations $g \in G_s$ to act on $e_{r r'}$ by
\begin{eqnarray}
{\cal T} e_{r r'} {\cal T}^{-1} &=& e^{\dagger}_{r r'} \\
U_g e_{r r'} U^{-1}_g &=& e_{gr, gr'} \text{,}
\end{eqnarray}
we see that the Gauss' law constraint respects the $G'$ symmetry.  In addition, the Hamiltonian has to be made gauge-invariant via the minimal coupling prescription, which can be done while respecting $G'$.

Now we consider the case $G = {\rm U}(1) \times \zz^T$.  Again we have a lattice with sites $r$; because there is no space group symmetry, the lattice does not have to satisfy any symmetry conditions.  Each site again has a charge density operator $N_r$ with integer eigenvalues.  Time reversal now acts by
\begin{equation}
{\cal T} N_r {\cal T}^{-1} = \delta N_r - N_r \text{,}
\end{equation}
where $\delta N_r$ must be an integer. By making constant integer shifts of $N_r$, we may choose $\delta N_r = 0,1$.  Next, by combining pairs of sites together as needed, and making further integer shifts of $N_r$, we can set $\delta N_r = 0$.

Introducing $\z_n$ electric fields and vector potentials as above, and imposing the Gauss' law Eq.~(\ref{eqn:zn-gauss-law}), we choose time reversal to act on the electric field by
\begin{equation}
{\cal T} e_{r r'} {\cal T}^{-1} = e_{r r'} \text{.}
\end{equation}
With this choice, the Gauss' law constraint respects the $\zz^T$ symmetry, as desired.

\section{Conditions under which $\Omega$ is a boson, and permutation of anyons in the gauged SET phase}
\label{app:Omega-is-boson}

Here, we show that the $\z_n$ symmetry flux $\Omega$ is a boson in the gauged SET phase whenever time reversal or reflection symmetry is present.  We also show that these operations either map $\Omega \mapsto \Omega$, or $\Omega \mapsto \Omega^{2n-1}$, depending on whether they commute with the ${\rm U}(1)$ symmetry.

The starting point for the analyses below are the fusion rules and statistics of the gauged SET phase.  According to the discussion of Sec.~\ref{subsec:gauging}, the fusion rules are
\begin{eqnarray}
Q^n &=&  1 \\
e^2 &=& Q \\
m^2 &=& 1 \\
\Omega^n &=& m Q^k  \text{,}
\end{eqnarray}
and the statistics are specified by
\begin{eqnarray}
\theta_e &=& \theta_m = 0  \\
\Theta_{e,m} &=& \pi  \\
\theta_Q &=& \Theta_{e,Q} = \Theta_{m,Q} = 0  \\
\Theta_{Q,\Omega} &=& \frac{2\pi}{n}  \\
\Theta_{e,\Omega} &=& \frac{\pi}{n}   \\
\Theta_{m, \Omega} &=& 0   \\
\theta_{\Omega} &=& \frac{\pi k}{n^2}  \text{.}
\end{eqnarray}
Here, $0 \leq k < n$ is an even integer.

The statistics of the gauged SET phase must obey certain conditions in the presence of time reversal or reflection symmetry.  For $\zz^T$ time reversal symmetry generated by ${\cal T}$, we write the action of ${\cal T}$ on some anyon $a$ in the gauged SET phase as ${\cal T} \star a$.  The statistics must satisfy
\begin{eqnarray}
\theta_{{\cal T} \star a} &=& - \theta_a \label{eqn:Ta} \\
\Theta_{ {\cal T} \star a, {\cal T} \star b } &=& - \Theta_{a, b} \text{.} \label{eqn:Tab}
\end{eqnarray}
These relations hold because the time reversed (clockwise) exchange process with time reversed anyons must give the same result as the ordinary exchange process before time reversal.

Next, under $\zz^P$ reflection symmetry generated by $P$, we denote the action of $P$ on $a$ by $P \star a$.  Equations~(\ref{eqn:Ta}) and~(\ref{eqn:Tab}) again hold, simply replacing $T$ by $P$.  This is the case because a counterclockwise exchange process is mapped to a clockwise one under $P$.

We will use these relations to show that $\Omega$ is a boson whenever $\zz^T$ or $\zz^P$ symmetry is present.  There are four cases to consider, where $G$ contains a subgroup ${\rm U}(1) \times \zz^T , {\rm U}(1) \times \zz^P, {\rm U}(1) \rtimes \zz^T$ or ${\rm U}(1) \rtimes \zz^P$.  We handle these cases one by one:

\emph{Case 1}: $G$ contains a subgroup ${\rm U}(1) \times \zz^T$.

Because ${\cal T}$ reverses the sign of ${\rm U}(1)$ charge, in the gauged SET phase we have
\begin{equation}
{\cal T} \star Q = \bar{Q} = Q^{n-1} \text{.}
\end{equation}
To determine the action of ${\cal T}$ on $e$ and $m$, note that ${\cal T}$ leaves these anyons invariant in the un-gauged SET phase, but it reverses their ${\rm U}(1)$ symmetry charges.  The $e$ sector of the gauged SET phase consists of those $e$ particles of the un-gauged SET phase whose ${\rm U}(1)$ charge modulo $n$ is $1/2$.  Similarly, the $m$ sector in the gauged SET phase consists of those $m$ particles in the un-gauged SET phase with the ${\rm U}(1)$ charge $0 \operatorname{mod} n$.  Therefore, we have
\begin{eqnarray}
{\cal T} \star e &=& \bar{Q} e \\
{\cal T} \star m &=& m \text{.}
\end{eqnarray} 

Now, let $\Omega' \equiv {\cal T} \star \Omega$.  In general, we can write
\begin{equation}
\Omega' = e^p \Omega^q \text{,}
\end{equation}
for integers $0 \leq p, q < 2n - 1$ that we will determine.  This is a unique parametrization of all $(2n)^2$ anyons in the gauged SET phase.  

Using Eq.~(\ref{eqn:Tab}),
\begin{equation}
0 = \Theta_{m, \Omega} = - \Theta_{m, \Omega'} = - p \Theta_{m,e} = - p \pi \text{.}
\end{equation}
This implies $p$ is even, and letting $\tilde{p} = p/2$, we have $\Omega' = Q^{\tilde{p}} \Omega^q$.  Next, we apply Eq.~(\ref{eqn:Tab}) again, this time to the mutual statistics of $e$ and $\Omega$, to obtain
\begin{eqnarray}
\frac{\pi}{n} &=& \Theta_{e, \Omega} = - \Theta_{e \bar{Q}, \Omega'} = - \Theta_{e Q^{n-1}, Q^{\tilde{p}} \Omega^q } \\
&=& - \Big[ q \frac{\pi}{n} + (n-1) q \frac{2 \pi}{n} \Big] \\
&=& q \frac{\pi}{n} \text{.}
\end{eqnarray}
This implies $q=1$, and so far we have shown $\Omega' = Q^{\tilde{p}} \Omega$.

Finally, we consider the self-statistics of $\Omega$, and apply Eq.~(\ref{eqn:Ta}), finding
\begin{eqnarray}
\frac{\pi k}{n^2} &=& \theta_{\Omega} = -\theta_{\Omega'} = - \theta_{Q^{\tilde{p}} \Omega} \\
&=& - \theta_{\Omega} - \tilde{p} \Theta_{Q, \Omega} \\
&=& - \frac{\pi k}{n^2} - \frac{2 \pi \tilde{p}}{n} \text{.}
\end{eqnarray}
This implies
\begin{equation}
\frac{2 \pi k}{n^2} = - \frac{2 \pi \tilde{p}}{n} \text{,}
\end{equation}
which has the unique solution $k = \tilde{p} = 0$.

Therefore we have shown that 
\begin{equation}
{\cal T} \star \Omega = \Omega \text{.}
\end{equation}
We also showed that $k = 0$, so that $\Omega$ is a boson ($\theta_{\Omega} = 0$).

\emph{Case 2}: $G$ contains a subgroup ${\rm U}(1) \times \zz^P$.

In this case, reflection does not act on ${\rm U}(1)$ charge, and we have
\begin{eqnarray}
P \star Q &=& Q \\
P \star e &=& e \\
P \star m &=& m \text{.}
\end{eqnarray}
As above, we let $\Omega' \equiv {\cal T} \star \Omega$ and write
\begin{equation}
\Omega' = e^p \Omega^q \text{,}
\end{equation}
for integers $0 \leq p, q < 2n - 1$ to be determined.

We follow the same strategy as in Case 1, repeatedly applying Eqs.~(\ref{eqn:Tab}) and~(\ref{eqn:Ta}) to determine $\Omega'$.  First we have
\begin{equation}
0 = \Theta_{m,\Omega} = - \Theta_{m, \Omega'} = - p \pi \text{,}
\end{equation}
which implies $p$ is even.  We define $\tilde{p} = p/2$, so that $\Omega' = Q^{\tilde{p}} \Omega^q$.  Then
\begin{eqnarray}
\frac{\pi}{n} &=& \Theta_{e,\Omega} = - \Theta_{e, \Omega'} \\
&=& - q \Theta_{e,\Omega} = - q \frac{\pi}{n} \text{.}
\end{eqnarray}
This implies that $q = 2n-1$, and so far we have shown $\Omega' = Q^{\tilde{p}} \Omega^{2n-1}$.
Finally,
\begin{eqnarray}
\frac{\pi k}{n^2} &=& \theta_{\Omega} = -\theta_{\Omega'} \\
&=& - \big[ \theta_{\Omega^{2n-1}} + \Theta_{Q^{\tilde{p}}, \Omega^{2n-1}} \big] \\
&=& - \big[ (2n-1)^2 \frac{\pi k}{n^2} + (2n-1) \tilde{p} \frac{2\pi}{n} \big]  \text{.}
\end{eqnarray}
Rearranging terms, and dropping those that vanish modulo $2\pi$, this is equivalent to
\begin{equation}
\frac{2\pi k}{n^2} = \frac{2\pi}{n} ( \tilde{p} + 2k) \text{,}
\end{equation}
which has the unique solution $k = \tilde{p} = 0$.

Therefore we have shown that 
\begin{equation}
P \star \Omega = \Omega^{2n-1} \text{.}
\end{equation}
We also showed that $k = 0$, so that $\Omega$ is a boson ($\theta_{\Omega} = 0$).

\emph{Case 3}: $G$ contains a subgroup ${\rm U}(1) \rtimes \zz^T$.

Here, time reversal does not change the ${\rm U}(1)$ charge, so we have
\begin{eqnarray}
{\cal T} \star Q &=& Q \\
{\cal T} \star e &=& e \\
{\cal T} \star m &=& m \text{.}
\end{eqnarray}
These equations are identical to those for $P$ in case 2.  Because the symmetry conditions on statistics are the same for ${\cal T}$ and $P$ symmetry, the analysis proceeds exactly as in case 2, and we have ${\cal T} \star \Omega = \Omega^{2n-1}$ and $\theta_{\Omega} = 0$.

\emph{Case 4}: $G$ contains a subgroup ${\rm U}(1) \rtimes \zz^P$.

In this case, $P$ reverses ${\rm U}(1)$ charge, so as in case 1 we have
\begin{eqnarray}
{\cal P} \star Q &=& \bar{Q} \\
{\cal P} \star e &=& \bar{Q} e \\
{\cal P} \star m &=& m \text{.}
\end{eqnarray}
Because these equations are identical to those in case 1, the analysis proceeds identically, so that $P \star \Omega = \Omega$ and $\theta_{\Omega} = 0$.

\section{Specifying fractionalization classes in terms of ${\rm U}(1)$ and $G'$ fractionalization classes}
\label{app:fracinfo}

By definition, the fractionalization class of $e$ or $m$ is an element $[\omega] \in H^2(G, \zz)$.  In this paper, we consider $G = {\rm U}(1) \rtimes G'$, and we specify the fractionalization class by two pieces of information:  1) whether the particle carries integer or half-odd integer ${\rm U}(1)$ charge, and 2) an element $[\omega'] \in H^2(G', \zz)$.  Here, we show that all fractionalization classes can be uniquely specified in this manner.

We observe that $[\omega] \in H^2(G, \zz)$ uniquely determines elements of $H^2({\rm U}(1), \zz)$ [corresponding to the ${\rm U}(1)$ charge modulo 1] and $H^2(G', \zz)$.  These elements are obtained by restricting the arguments of the factor set $\omega(g_1, g_2)$ to the ${\rm U}(1)$ and $G'$ subgroups, respectively.  Therefore, we need only show that no additional information is needed to uniquely specify $[\omega]$.  

We consider a projective representation of $G$, where $\phi \in {\rm U}(1)$ is represented by $e^{i \phi Q}$, and $g \in G'$ is represented by $\Gamma(g)$.  Letting $\sigma_q =1$ ($\sigma_q = -1$) correspond to integer (half-odd-integer) ${\rm U}(1)$ charge, we have
\begin{eqnarray}
e^{2 \pi i Q} &=& \sigma_q \\
\Gamma(g_1) \Gamma(g_2) &=& \omega'(g_1, g_2) \Gamma(g_1 g_2) \text{.}
\end{eqnarray}
So far, we have only specified $\sigma_q$ and $[\omega'] \in H^2(G', \zz)$.

To complete the description of the projective representation, we need to describe the multiplication of an element of ${\rm U}(1)$ with an element of $G'$.  First, fix $g \in G'$, and suppose that $\phi g = g \phi$ for all $\phi \in {\rm U}(1)$.  Then, in the projective representation
\begin{equation}
e^{i \phi Q} \Gamma(g) e^{-i \phi Q} = f_g(\phi) \Gamma(g) \text{,}
\end{equation}
where $f_g(\phi) \in \{ \pm 1\}$.  Setting $\phi = 0$, clearly $f_g(0) = 1$.  Moreover, the left-hand side is a continuous function of $\phi$, so $f_g(\phi)$ must also be continuous, and $f_g(\phi) = 1$ for all $\phi$.

The other possibility we need to consider is a fixed $g \in G'$ where $\phi g = g (- \phi)$ for all $\phi \in {\rm U}(1)$.  In the projective representation,
\begin{equation}
e^{i \phi Q} \Gamma(g) e^{i \phi Q} = f_g(\phi) \Gamma(g) \text{.}
\end{equation}
Here, the same arguments show that $f_g(\phi) = 1$.

We have thus shown that the fractionalization class $[\omega] \in H^2(G, \zz)$ is completely specified by $\sigma_q$ and $[\omega'] \in H^2(G', \zz)$.

\section{Characterization of anomaly-negative fractionalization patterns}
\label{app:characterization-theorem}

We recall that, by definition, the symmetry fractionalization pattern $eC[\omega_e]m0[\omega_m]$ is anomaly-negative if and only if, for each even $n \geq 2$, there exists a $t$-twisted $\z_{n}$ factor set $\phi_n(g_1, g_2)$ solving the equation
\begin{equation}
\omega_m(g_1, g_2) = [\phi_n(g_1, g_2)]^n \text{,} \label{eqn:main-equation-appendix}
\end{equation}
where $g_1, g_2 \in G'$.  In this section, we prove a simple characterization, stated as Theorem~\ref{thm:characterization-maintext} in Sec.~\ref{subsec:anomaly-test},  of which $m$ particle fractionalization classes $[\omega_m]_{\zz}$ give rise to anomaly-negative fractionalization patterns.

It will be convenient to relate the $\zz$ and $\z_{2n}$ factor sets $\omega_m$ and $\phi_n$ to $t$-twisted ${\rm U}(1)$ factor sets.  To proceed, let $k$ be a positive integer. $Z^2_t(G', \z_{2k})$ is the Abelian group of $t$-twisted $\z_{2k}$ 2-cocycles (factor sets) for the group $G'$.
$B^2_t(G', \z_{2k})$ is the corresponding Abelian group of 2-coboundaries, which are factor sets of the form $\omega(g_1, g_2) = \lambda(g_1) [\lambda(g_2)]^{t(g_1)} [\lambda(g_1 g_2)]^{-1}$, for $\lambda(g) \in \z_{2k}$.  The second cohomology group is defined by $H^2_t(G', \z_{2k}) = Z^2_t(G', \z_{2k}) / B^2_t(G', \z_{2k} )$.  There is a projection homomorphism $\pi_{2k} : Z^2_t(G', \z_{2k} ) \to H^2(G', \z_{2k} )$.  Note that if $k=1$, we can drop the $t$ subscripts everywhere, since in that case the twisting by $t(g)$ is trivial.  The same definitions hold for ${\rm U}(1)$ coefficients, in which case we call the projection homomorphism $\pi_{{\rm U}(1)} : Z^2_t(G', {\rm U}(1) ) \to H^2_t(G', {\rm U}(1) )$.  There is an obvious inclusion map $i_{2k} : Z^2_t(G', \z_{2k} ) \to Z^2_t(G', {\rm U}(1) )$, which just expresses the fact that $\omega \in Z^2_t(G', \z_{2k})$ can also be viewed as a $t$-twisted ${\rm U}(1)$ factor set.

For each $k$, we would like to define a homomorphism $\rho_{2k} : H^2_t(G', \z_{2k} ) \to H^2_t(G', {\rm U}(1) )$, so that the following diagram is commutative:
\begin{equation}
\begin{CD}
Z^2_t(G', \z_{2k}) @>i_{2k}>> Z^2_t(G', {\rm U}(1) ) \\
@VV\pi_{2k} V @VV\pi_{{\rm U}(1)} V \\
H^2_t(G', \z_{2k} ) @>\rho_{2k}>> H^2_t(G', {\rm U}(1) ) 
\end{CD} \text{.}
\end{equation}
In fact, we will see that $\rho_{2k}$ is the unique homomorphism making this diagram commutative.

Why do we want to define $\rho_{2k}$?  Given $\omega \in Z^2_t(G', \z_{2k} )$, we can define a ${\rm U}(1)$ fractionalization class by $[\omega]_{{\rm U}(1)} = \pi_{{\rm U}(1)} ( i_{2k} ( \omega )) \in H^2_t(G', {\rm U}(1))$.  If we can find a unique $\rho_{2k}$, commutativity of the diagram tells us that $[\omega]_{{\rm U}(1)}$ depends, in a unique way, only on the $\z_{2k}$ fractionalization class $[\omega]_{\z_{2k}} = \pi_{2k}(\omega) \in H^2_t(G', \z_{2k})$, by $[\omega]_{{\rm U}(1)} = \rho_{2k} ( [\omega]_{\z_{2k}} )$.  Therefore, it is meaningful to talk about $[\omega]_{{\rm U}(1)}$ as a function of $[\omega]_{\z_{2k}}$.

We define $\rho_{2k}$ as follows.  Pick some element $c \in H^2_t(G', \z_{2k})$.  Choose a representative $\omega \in Z^2_t(G', \z_{2k} )$ so that $c = \pi_{2k}(\omega)$.  Then define $\rho_{2k}(c) = \pi_{{\rm U}(1)} ( i_{2k} ( \omega ))$.

First, we have to check $\rho_{2k}$ is well-defined, which means it must be independent of the particular representative $\omega$.  It is easy to see that $\omega, \omega' \in Z^2_t(G', \z_{2k})$ belonging to the same cohomology class, also belong to the same cohomology class after mapping under $i_{2k}$ to $Z^2_t(G', {\rm U}(1))$, so $\rho_{2k}$ is well-defined.  Next, we check $\rho_{2k}$ is a homomorphism.  Let $c, c' \in H^2_t(G', \z_{2k})$, with corresponding representatives $\omega, \omega'$.  We have $c c' = \pi_{2k}(\omega) \pi_{2k}(\omega') = \pi_{2k}(\omega \omega')$, so $\omega \omega'$ is a representative for $c c'$.
Then
\begin{equation}
\rho_{2k}(c c') = \pi_{{\rm U}(1)} (i_{2k} ( \omega \omega')) = \rho_{2k}(c) \rho_{2k}(c') \text{.}
\end{equation}

Finally, we check $\rho_{2k}$ is the unique homomorphism making the diagram commutative.  Suppose $\tilde{\rho}_{2k}$ also makes the diagram commutative, but for some $c \in H^2_t(G', \z_{2k})$, we have $\rho_{2k}(c) \neq \tilde{\rho}_{2k}(c)$.  Let $\omega \in Z^2_t(G', \z_{2k})$ be a representative for $c$, then it follows that $\pi_{{\rm U}(1)}(i_{2k}(\omega)) = \rho_{2k}(\pi_{2k}(\omega)) = \tilde{\rho}_{2k}(\pi_{2k}(\omega))$, which implies $\rho_{2k}(c) = \tilde{\rho}_{2k}(c)$, a contradiction.

Now we return to Eq.~(\ref{eqn:main-equation-appendix}).  Viewing $\omega_m$ and $\phi_n$ as ${\rm U}(1)$ factor sets, and applying $\pi_{{\rm U}(1)}$ to both sides of the equation, we have
\begin{equation}
[\omega_m]_{{\rm U}(1)} = ( [ \phi_n]_{{\rm U}(1)} )^n \text{.} \label{eqn:nth-root}
\end{equation}
Therefore, another way of putting the anomaly test is that, in order for $[\omega_m]_{\zz}$ to be anomaly-negative, $[\omega_m]_{{\rm U}(1)}$ must have a $n$th root in $H^2_t(G', {\rm U}(1))$ for all even $n \geq 2$.

Now we can prove the desired characterization of anomaly-negative $[\omega_m]_{\zz}$.  We assume that $H^2_t(G', {\rm U}(1) ) = {\rm U}(1)^k \times A$, where $A$ is a finite product of finite cyclic factors.  This assumption is true in all the examples we studied, and we believe it is likely to hold in general.

\begin{prop}
Suppose that $[\omega_m]_{{\rm U}(1)}$ lies in the connected component of $H^2_t(G', {\rm U}(1))$ that contains the identity element.  Suppose also that $H^2_t(G', {\rm U}(1) ) = {\rm U}(1)^k \times A$, where $A$ is a finite product of finite cyclic factors.  Then $[\omega_m]_{\zz}$ is anomaly-negative. \label{prop:connected-new}
\end{prop}

\begin{proof}
It follows from the assumptions that $[\omega_m]_{{\rm U}(1)}$ has an $n$th root in $H_t^2(G', {\rm U}(1))$ for any $n > 0$.  This holds by the assumption that $H^2_t(G', {\rm U}(1))$ is a product of ${\rm U}(1)$ and finite cyclic factors, so that the connected component containing $1$ is just a product of ${\rm U}(1)$'s.  Then we have
\begin{equation}
\omega_m(g_1, g_2) = \lambda^{-1}(g_1) [\lambda(g_2)]^{-t(g_1)} \lambda(g_1 g_2) ( \Omega_n(g_1, g_2) )^n \text{,}
\end{equation}
where $\Omega_n \in Z^2_t(G', {\rm U}(1))$, and $\lambda(g) \in {\rm U}(1)$.

We choose $0 \leq \theta(g) < 2\pi$ so that $\lambda(g) = e^{i \theta(g)}$.  Then we define $\alpha(g) = e^{i \theta(g) / n}$,
and we choose
\begin{equation}
\phi_n(g_1, g_2) = \alpha^{-1}(g_1) [\alpha(g_2)]^{-t(g_1)} \alpha(g_1 g_2) \Omega_n(g_1, g_2) \text{.}
\end{equation}
This is by construction a $n$th root of $\omega_m(g_1, g_2)$, so $\phi_n(g_1,g_2) \in \z_{2n}$.  To show $\phi_n \in Z^2_t(G', \z_{2n})$, we note that $\phi_n(g_1, g_2)$ clearly satisfies the relevant associativity condition.  For the given $\omega_m$, we have thus constructed a solution to Eq.~(\ref{eqn:main-equation}) for each even $n > 0$.
\end{proof}

The converse of Proposition~\ref{prop:connected-new} is also true:

\begin{prop}
If $[\omega_m]_{\zz}$ is anomaly-negative, and if $H^2_t(G', {\rm U}(1) ) = {\rm U}(1)^k \times A$, where $A$ is a finite product of finite cyclic factors, then   $[\omega_m]_{{\rm U}(1)}$ lies in the same connected component of $H_t^2(G', {\rm U}(1))$ that contains the identity element.
\end{prop}

\begin{proof}
Under the assumption, Eq.~(\ref{eqn:nth-root}) holds; that is, $[\omega_m]_{{\rm U}(1)}$ has an $n$th root in $H^2_t(G', {\rm U}(1))$ for all even $n > 0$.

We write $[\omega_m]_{{\rm U}(1)} = (\alpha, \beta)$, where $\alpha \in {\rm U}(1)^k$ and $\beta \in A$.  We will show that $\beta = 1$, which implies $[\omega_m]_{{\rm U}(1)}$ lies in the connected component containing the identity.

Write $A = \z_{p_1} \times \cdots \times \z_{p_N}$, and $\beta = (b_{1}, \dots, b_{N})$.  Observe that $[\omega_m]_{{\rm U}(1)}^2 = 1$, which implies $\beta^2 = 1$.  If $p_i$ is odd, $\beta^2 = 1$ implies $b_i = 1$.  Now consider $p_i$ even.  Then, by assumption, there exists a $p_i$th root of $\beta$, $\beta = \gamma^{p_i}$ for $\gamma \in A$.  This implies $b_i = c^{p_i}$ for some $c \in \z_{p_i}$, but for any $c \in \z_{p_i}$, $b_i = c^{p_i} = 1$.  Therefore, $\beta = 1$.
\end{proof}

Taking these two propositions together, we have proved Theorem~\ref{thm:characterization-maintext}.

\section{Computing second cohomology groups using generators and relations}
\label{app:cohomology}

Here we provide some details to justify and explain the procedure used for computing second cohomology groups in Sec.~\ref{sec:examples} and Appendix~\ref{sec:more-examples}.  We focus on the $t$-twisted $\z_{2n}$ cohomology group $H^2_t(G', \z_{2n})$.  This includes $H^2(G', \zz)$ as a special case (setting $n=1$), and the treatment for $H^2_t(G', {\rm U}(1))$ proceeds identically, simply replacing $\z_{2n}$ by ${\rm U}(1)$ throughout the discussion.

The group $H^2_t(G', \z_{2n})$ can be computed by finding, and distinguishing, all possible equivalence classes of $t$-twisted $\z_{2n}$ factor sets $\omega(g_1,g_2)$, for $g_1, g_2 \in G'$.  Recall that such a factor set is any $\z_{2n}$-valued function satisfying the twisted associativity condition, Eq.~(\ref{eqn:twisted-associativity}), and that we are referring to equivalence classes under projective transformations defined in Eq.~(\ref{eqn:twisted-proj-trans}).

Rather than directly studying factor sets, we can equivalently study $t$-twisted $\z_{2n}$ group extensions of $G'$. Such a group extension is a group $E$ for which $\z_{2n} \subset E$ is a normal subgroup, satisfying $E / \z_{2n} = G'$.  An arbitrary element $e \in E$ can be written $e = a u(g)$, where $a \in \z_{2n}$, and $u(g)$ is chosen to satisfy $\pi [ u(g) ] = g$, where $\pi : E \to G'$ is the projection map associated with the quotient of $E$ by $\z_{2n}$.  We refer to $u(g)$ as a representative of $g$ in $E$.  We require the additional property
\begin{equation}
u(g) a = a^{t(g)} u(g) \text{,}  \label{eqn:extension-t-twisting}
\end{equation}
where $t : G' \to \zz$ is the twisting homomorphism discussed in Sec.~\ref{subsec:anomaly-test}.
We note that the representative $u(g)$ is arbitrary up to projective transformations
\begin{equation}
u(g) \to \lambda^{-1}(g) u(g) \text{,} \label{eqn:extension-proj-trans}
\end{equation}
where $\lambda(g) \in \z_{2n}$.

It follows from the definition that
\begin{equation}
u(g_1) u(g_2) = \omega(g_1, g_2) u(g_1 g_2) \text{,}
\end{equation}
where $\omega(g_1, g_2) \in \z_{2n}$.  Associative multiplication of the $u(g)$'s, together with Eq.~(\ref{eqn:extension-t-twisting}), implies that $\omega$ satisfies Eq.~(\ref{eqn:twisted-associativity}), and is thus a $t$-twisted $\z_{2n}$ factor set.  In addition, under projective transformations Eq.~(\ref{eqn:extension-proj-trans}), the factor set transforms as in Eq.~(\ref{eqn:twisted-proj-trans}).  So we have shown that a group extension is associated with a unique equivalence class $[\omega] \in H^2_t(G', \z_{2n})$.

Now we would also like to show that, given a factor set $\omega(g_1, g_2)$, we can construct a corresponding group extension.  We consider a set $E$ whose elements are ordered pairs $(a, g)$, where $a \in \z_{2n}$ and $g \in G'$.  We make this set into a group by defining the multiplication operation
\begin{equation}
(a_1, g_1) \times (a_2, g_2) = (a_1 a_2^{t(g_1)} \omega(g_1, g_2), g_1 g_2 ) \text{.}
\end{equation}
With this multiplication, it can be checked that $E$ is a group, and indeed a $t$-twisted $\z_{2n}$ group extension.\footnote{To give some details, note that the identity in $E$ is $1 = (\omega(1,1)^{-1},1)$.  There is an injective homomorphism $i : \z_{2n} \to E$ defined by $i(a) = (a \omega(1,1)^{-1},1)$.  To view $E$ as a $t$-twisted $\z_{2n}$ extension, we use the normal subgroup  $i(\z_{2n}) \subset E$, which is isomorphic to $\z_{2n}$.}  Choosing $u(g) = (1,g)$, we have $u(g_1) u(g_2) = \omega(g_1, g_2) u(g_1 g_2)$, as desired.

It follows from the above discussion that, if we would like to construct all possible factor sets (or equivalence classes thereof), it is enough to construct all possible group extensions.  We now describe, in general terms, how to do this for a group $G'$ presented in terms of generators and relations.  This procedure is worked out in Sec.~\ref{sec:examples} and Appendix~\ref{sec:more-examples} for specific examples.  We note that in those sections, to simplify the discussion in the main text, we slightly abuse terminology and refer to projective representations, which are group extensions with additional vector space structure.  This additional structure is not used in the cohomology group calculations, which can be viewed more simply as calculations with group extensions.

To begin, we describe the presentation of $G'$ in terms of a finite number of generators $h_i \in G'$ ($i = 1,2,\dots$).  Note that our goal here is not to define $G'$ abstractly in terms of generators and relations, but rather to give a description of $G'$ in this manner, assuming that $G'$ is already defined by some other means.  For every $g \in G'$, we choose a fixed canonical form in terms of the generators, for example $g_1 = h_1 h_3^2$.  In general, different choices of canonical form are possible for each $g$, and fixing the canonical form should be viewed as an arbitrary choice.  Fortunately, while we use the canonical form to justify our calculation procedure, it is not necessary to make a specific choice in the explicit calculations.  It is important to note that $h_i^{-1}$ is not automatically included as a generator, but sometimes it may need to be included, so that all $g \in G'$ can be written as a product of generators.

The generators obey a finite number of relations, for example,
\begin{eqnarray}
h_1^2 &=& 1 \\
(h_1 h_2)^4 &=& 1 \text{,}
\end{eqnarray}
and so on.  For the present purposes of general discussion, we work in a convention where the right-hand side of each relation is the unit element; however, this is not always convenient in practice.  The relations must be chosen so that, given any $g_1, g_2 \in G'$ expressed in canonical form, the relations alone can be used to bring the product $g_1 g_2$ to canonical form.

Now suppose $E$ is a $t$-twisted $\z_{2n}$ group extension of $G'$.  For each $g \in G'$, by making suitable projective transformations, we can choose a canonical form for $u(g)$, which is the product of $u(h_i)$ corresponding to the canonical form of $g$.  For example, if $g_1 = h_1 h_3^2$, we choose $u(g_1) = u(h_1) [u(h_3)]^2$, with trivial $\z_{2n}$ coefficient.  It is always possible to make such a choice, by making projective transformations $u(g) \to \lambda^{-1}(g) u(g)$, where $\lambda(h_i) = 1$.  We also choose $u(1) = 1$.  In addition, if $h_i$ and $h_i^{-1}$ are both generators, we choose $u(h_i^{-1}) = [u(h_i)]^{-1}$, which can be accomplished via a projective transformation $\lambda(g)$ where $\lambda(g) = 1$ if $g \neq h_i^{-1}$.

The relations now become relations for the $u(h_i)$, with the right-hand side modified to be an arbitrary element of $\z_{2n}$, for example,
\begin{eqnarray}
\left[ u(h_1) \right]^2 &=& \alpha_1 \label{eqn:pr1} \\
\left[ u(h_1) u(h_2) \right]^4 &=&  \alpha_2 \text{,} \label{eqn:pr2}
\end{eqnarray}
for $\alpha_1, \alpha_2 \in \z_{2n}$.  We note that, due to the special choice of $u(h_i^{-1})$ when both $h_i$ and $h_i^{-1}$ are generators, we automatically have $\alpha_i = 1$ for the relation $h_i \cdot h_i^{-1} = 1$.  These relations allow us to bring any product $u(g_1) u(g_2)$ into canonical form $u(g_1 g_2)$, up to a $\z_{2n}$ phase factor determined by the $\{ \alpha_i \}$.  This phase factor is nothing but $\omega(g_1, g_2) \in \z_{2n}$, and $u(g_1) u(g_2) = \omega(g_1,g_2) u(g_1 g_2)$.  The set $\{ \alpha_i \}$ thus determines $\omega(g_1, g_2)$.  We note that the sets $\{ \alpha_i \}$ can be multiplied according to
\begin{equation}
\{ \alpha_i \} \times \{ \beta_i \} = \{ \alpha_i \beta_i \} \label{eqn:alpha-mult} \text{,}
\end{equation}
which corresponds to the multiplication of factor sets. 

It is clear that any extension $E$ can be described by a corresponding set $\{ \alpha_i \}$.  (Note that the converse of this statement is not true.)  This fact allows us to find all equivalence classes $[\omega] \in H^2_t(G', \z_{2n})$ via the following procedure.  First, we consider the $\alpha_i$ to be free parameters.  We then exploit the remaining freedom to make projective transformations, where $\lambda(g) \neq 1$ only if $g$ is a generator, to ``fix a gauge'' for the $\alpha_i$.  After gauge-fixing, distinct sets $\{ \alpha_i \}$ are inequivalent under projective transformations.  Next, we need to determine which sets $\{ \alpha_i \}$ are consistent, giving rise to an extension $E$ (or, equivalently, to a factor set $\omega$).  Some sets $\{ \alpha_i \}$ can be ruled out by algebraic manipulations of the relations; for example, one can conjugate various relations by one of the generators, which often puts constraints on some of the $\alpha_i$.  After ruling out some sets $\{ \alpha_i \}$ in this manner, one can tentatively conclude that the remaining gauge-fixed sets $\{ \alpha_i \}$ correspond to elements of $H^2_t(G', \z_{2n})$.  This not only gives a computation of the group $H^2_t(G', \z_{2n})$, but also an explicit parametrization in terms of gauge-fixed sets $\{ \alpha_i \}$, with the group multiplication given by Eq.~(\ref{eqn:alpha-mult}).

To verify this tentative answer for $H^2_t(G', \z_{2n})$, one needs to show that each $\{ \alpha_i \}$ in fact gives rise to a factor set.  It is enough to do this for sets $\{ \alpha_i \}$ that generate $H^2_t(G', \z_{2n})$.  In each case, we can verify the existence of the corresponding factor set by, for example, exhibiting a projective representation for which the relations realize the set $\{ \alpha_i \}$.

\section{Coupled layer construction}
\label{app:coupled-layer}

The fact that all symmetry fractionalization patterns are possible on the surface of some $d=3$ SPT phase (which may be the trivial SPT phase) plays an important role in the discussion of this paper.  Here, we establish this fact using a simple generalization of the coupled layer construction of Ref.~\onlinecite{cwang13}; the discussion in the first part of this Appendix closely follows Section~IV of that work.

\begin{figure}
\includegraphics[width=0.5\columnwidth]{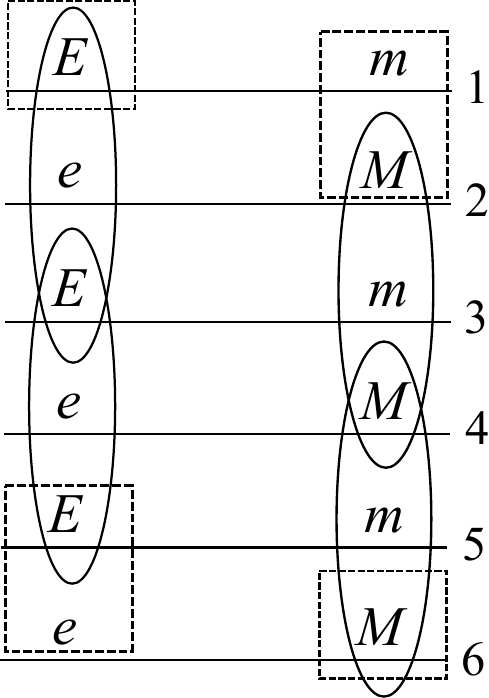}
\caption{Coupled layer construction.  Each layer is a SET phase with $\zz$ topological order.  $E$ and $M$ particles transform non-trivially under the symmetry $G$, while $e$ and $m$ transform trivially.  Composite particles indicated by ovals are condensed to obtain a $d=3$ SPT phase (which may be the trivial SPT phase).  The particles in dashed boxes remain deconfined and uncondensed, and give rise to surface SET phases at the top and bottom surfaces. By choosing the fractionalization classes of $E$ and $M$, surface SET phases with any desired symmetry fractionalization pattern can be realized by this construction.} 
\label{fig:coupled-layer}
\end{figure}

We consider a symmetry group $G$, and a fractionalization pattern $e[\omega_e]m[\omega_m]$ (perhaps anomalous) for a $d=2$ state with $\zz$ topological order.  Here, $[\omega_e]_{\zz}, [\omega_m]_{\zz} \in H^2(G, \zz)$ are the fractionalization classes of $e$ and $m$ particles, respectively.  There are no restrictions on $[\omega_e]_{\zz}$ and $[\omega_m]_{\zz}$; that is, $([\omega_e]_{\zz}, [\omega_m]_{\zz})$ is an arbitrary pair of elements of $H^2(G, \zz)$.

We build a $d=3$ system as a stack of $d=2$ layers of SET phases with $\zz$ topological order, alternating between layers where we label the two bosonic anyons as $E_i, m_i$, and layers where they are labeled $e_i, M_i$, as shown in Fig.~\ref{fig:coupled-layer}.  We choose the $E_i$ to have fractionalization class $[\omega_e]_{\zz}$, and the $M_i$ to have fractionalization class $[\omega_m]_{\zz}$.  The $e_i$ and $m_i$ have trivial fractionalization class.  We argue below that such layers can indeed be realized strictly in $d=2$.

We assume we have a total of $N$ layers with $N$ even, and condense the composite particles $E_i e_{i+1} E_{i+2}$ (for $i = 1, 3, \dots, N-3$), and $M_i m_{i+1} M_{i+2}$ (for $i = 2, 4, \dots, N - 2$).  These particles are bosons with trivial mutual statistics, so they can indeed be condensed simultaneously.  Moreover, the fractionalization classes of these particles are trivial, so they can be condensed without breaking symmetry.

In the state obtained upon condensation, all anyon excitations in the bulk are either confined or condensed.  Since the symmetry is not broken by the condensation, the resulting state is thus a $d=3$ SPT phase, which may be the trivial SPT phase.  At the $i=1$ surface, the particles $E_1$ and $m_1 M_2$ remain deconfined, and have fractionalization classes $[\omega_e]_{\zz}$ and $[\omega_m]_{\zz}$, respectively.  These are the quasiparticles of the desired surface SET phase with $\zz$ topological order and  fractionalization pattern $e[\omega_e]m[\omega_m]$.  The same holds at the $i= N$ surface, for the particles $E_{N-1} e_N$ and $M_N$.

To conclude the discussion, we need to verify that the layers in our construction are allowed strictly in $d=2$.  Equivalently, we need to argue that the fractionalization pattern $e[\omega_e]m0$ is non-anomalous for arbitrary $[\omega_e]_{\zz} \in H^2(G, \zz)$.  To do this, we construct a $\zz$ gauge theory where the matter field carrying $\zz$ gauge charge transforms with fractionalization class $[\omega_e]_{\zz}$, and show that this gauge theory can arise as a low-energy theory for a spin model.

The $\zz$ gauge charge is carried by a multi-component boson field $b^\dagger_{r \alpha}$, where $r$ labels the sites of a lattice invariant under the symmetry, and $\alpha$ labels the components.  We take the symmetry operation $g \in G$ to act on the boson field by
\begin{equation}
g : b^\dagger_{r \alpha} \mapsto \Gamma(g)^{\vphantom\dagger}_{\alpha \beta} b^\dagger_{g r, \beta} \text{.}
\end{equation}
Here, the matrices $\Gamma(g)$ are chosen to form a projective representation of $G$ whose factor set belongs to the desired fractionalization class $[\omega_e]_{\zz}$.

We choose a set $L$ of lattice links $\ell = (r, r')$ that make the lattice into a connected graph respecting the symmetry, and introduce a $\zz$ gauge field defined on links $\ell \in L$.  On each link $\ell \in L$ we introduce a two-dimensional Hilbert space, acted on by the $\zz$ vector potential $\sigma^z_{\ell}$ and the $\zz$ electric field $\sigma^x_\ell$.  These operators can be thought of as $2 \times 2$ Pauli matrices.  Apart from the action of space group operations on links, symmetry acts trivially on these fields, that is
\begin{eqnarray}
g : \sigma^{x,z}_{\ell} \mapsto \sigma^{x,z}_{g \ell} \text{.}
\end{eqnarray}

The Hamiltonian takes the form
\begin{equation}
H = - h \sum_{\ell \in L} \sigma^x_{\ell} - K \sum_{p} \prod_{\ell \in p} \sigma^z_{\ell} + u \sum_r b^\dagger_{r \alpha} b^{\vphantom\dagger}_{r \alpha}  \text{,}
\end{equation}
where $h, K, u > 0$, and the second sum is over a set of elementary plaquettes $p$ of the lattice.  We may add additional short-ranged terms consistent with symmetry, but will not need to do so for the present discussion.
We also have to specify the Gauss' law constraint, which we take to be
\begin{equation}
\prod_{r' \sim r} \sigma^x_{r r'} = (-1)^{b^\dagger_{r \alpha} b^{\vphantom\dagger}_{r \alpha} } \text{,}
\end{equation}
where the product is over those sites $r'$ joined to $r$ by some link $(r, r') \in L$.

We consider two limits of the Hamiltonian.  First, when $h = 0$, the Hamiltonian is exactly solvable, and describes a $\zz$ gauge theory in its deconfined phase, with gapped bosonic matter.  The $e$ particles, which are simply the $b^\dagger_{r \alpha}$ bosons, have fractionalization class $[\omega_e]_{\zz}$.  To see that the $m$ particles have trivial fractionalization class, we note that we can integrate out the bosonic matter in the limit where $u$ is large, to obtain a pure $\zz$ gauge theory with gauge constraint $\prod_{r' \sim r} \sigma^x_{r r'} = 1$.  Because there is no background $\zz$ gauge charge, symmetry acts trivially on the $m$ particles, and the $m$ particle fractionalization class is trivial.  Therefore, this gauge theory indeed realizes the $e[\omega_e]m0$ fractionalization pattern.

We also consider the limit $h \gg u, K$, which is a confining limit for the $\zz$ gauge field.  In this limit we may put $\sigma^x_{\ell} = 1$, and the gauge constraint becomes
\begin{equation}
(-1)^{b^\dagger_{r \alpha} b^{\vphantom\dagger}_{r \alpha} }  = 1 \text{.}
\end{equation}
This constrains the number of bosons to be even on each lattice site, and defines the Hilbert space for a bosonic model, for which the Hilbert space is a product of site Hilbert spaces.  Because all operators acting within this Hilbert space add or remove even numbers of bosons, such operators transform linearly under $G$, which is an important requirement for any physical model with $G$ symmetry.  We thus recover a sensible spin model in the confining limit of the gauge theory, and, therefore, the gauge theory can arise as a low-energy effective theory of such a spin model.  We then expect that the deconfined phase with  $e[\omega_e]m0$ fractionalization pattern can occur in this spin model, albeit for some unknown and possibly complicated Hamiltonian.

In certain special cases, it has also been shown via construction of exactly solvable spin models (\emph{i.e.}, not parton gauge theories) that all fractionalization patterns $e[\omega_e]m0$ can occur strictly in $d=2$.  This has been done for arbitrary finite, unitary, on-site symmetry \cite{hermele14}, and also for $p4mm$ square lattice space group symmetry \cite{hsong15}.

\section{More examples}
\label{sec:more-examples}

\subsection{$G = ({\rm U}(1) \rtimes \zz^T) \times pm$}
\label{subsec:pm-z2T}

This symmetry is closely related to the case $G = {\rm U}(1) \times pm$, but now with time reversal symmetry added.  The $\zz^T$ time reversal forms a semidirect product with ${\rm U}(1)$.  The generators are as in Sec.~\ref{subsec:pm}, with the addition of the time reversal operation ${\cal T}$, and we have the relations
\begin{eqnarray}
T_x T_y T^{-1}_x T^{-1}_y &=& 1 \\
T_y P_x T^{-1}_y P_x &=& 1 \\
P_x^2 &=& 1 \\
T_x P_x T_x P_x &=& 1 \\
{\cal T}^2 &=& 1 \\
{\cal T} T_x &=& T_x {\cal T}  \\
{\cal T} T_y &=& T_y {\cal T}  \\
{\cal T} P_x &=& P_x {\cal T} \text{.}
\end{eqnarray}
The $m$ symmetry fractionalization is specified by
\begin{eqnarray}
T^m_x T^m_y T^{m-1}_x T^{m-1}_y &=& \sigma^m_{txty} \\
T^m_y P^m_x T^{m-1}_y P^{m}_x &=& \sigma^m_{typx} \\
(P^m_x)^2 &=& \sigma^m_{px} \\
T^m_x P^m_x T^m_x P^m_x &=& \sigma^m_{txpx} \\
({\cal T}^m)^2 &=& \sigma^m_T \\
{\cal T}^m T_x^m &=& \sigma^m_{Ttx} T_x^m {\cal T}^m  \\
{\cal T}^m T_y^m &=&  \sigma^m_{Tty} T_y^m {\cal T}^m  \\
{\cal T}^m P_x^m &=& \sigma^m_{Tpx} P_x^m {\cal T}^m \text{,}
\end{eqnarray}
where the $\sigma^m$'s take values in $\zz$.  All the $\sigma^m$'s are invariant under projective transformations of the generators, so
we tentatively conclude that $[\omega_m]_{\zz} \in H^2(G', \zz) = (\zz)^8$.  To be sure this is correct, we need to show that each of the possible $2^8$ choices of the $\sigma^m$'s can actually be realized by a corresponding factor set.  It is enough to give a set of projective representations whose cohomology classes generate $H^2(G', \zz)$; this is done in Table~\ref{tab:pmT-genset}.

\begin{table}
\begin{tabular}{|c|c|c|c|c|c|}
\hline
Rep. number & $T_x$ & $T_y$ & $P_x$ & $U_T$ & $\sigma$'s that are $-1$ \\
\hline
1 & $\sigma^x$ & $1$ & $i \sigma^y$ & $1$  & $\sigma_{px}, \sigma_{typx}$ \\
\hline
2 & $\sigma^x$ & $1$ & $\sigma^z$ & $1$  & $\sigma_{txpx}$ \\
\hline
3 & $\sigma^x$ & $\sigma^z$ & $1$ & $1$  & $\sigma_{txty}$ \\
\hline
4 & $1$ & $\sigma^x$ & $\sigma^z$ & $1$  & $\sigma_{typx}$ \\
\hline
5 & $1$ & $1$ & $1$ & $i \sigma^y$  & $\sigma_{T}$ \\
\hline
6 & $\sigma^x$ & $1$ & $1$ & $\sigma^z$  & $\sigma_{Ttx}$ \\
\hline
7 & $1$ & $\sigma^x$ & $1$ & $\sigma^z$  & $\sigma_{Tty}$ \\
\hline
8 & $1$ & $1$ & $\sigma^x$ & $\sigma^z$  & $\sigma_{Tpx}$ \\
\hline
\end{tabular}
\caption{Set of 8 projective representations whose cohomology classes generate $H^2(G', \zz) = \zz^8$, where $G' = pm \times \zz^T$.  Note that the cohomology classes of the first four representations listed generate $H^2(pm, \zz) = \zz^4$.
The first column numbers the representations, 1 through 8.  The middle four columns specify generators of the group in the corresponding representation (time reversal is ${\cal T} = U_T K$, where $K$ is complex conjugation).  All representations in the table are two-dimensional.  Generators  are specified in terms of the Pauli matrices $\sigma^{x,y,z}$. The last column lists those $\sigma$'s that are equal to $-1$ for the corresponding representation.}
\label{tab:pmT-genset}
\end{table}

Next, we need to compute $H^2_t(G', {\rm U}(1))$, noting that $t(P_x) = -1$ and $t(T_x) = t(T_y) = t({\cal T}) = 1$. Time reversal acts trivially on the ${\rm U}(1)$ coefficients because ${\cal T}$ is anti-unitary and ${\cal T} \star \Omega = \Omega^{2n - 1}$; these two effects cancel out so that $t({\cal T}) = 1$.  We start by specifying
\begin{eqnarray}
T^t_x T^t_y T^{t-1}_x T^{t-1}_y &=& \alpha_{txty}  \\
T^t_y P^t_x T^{t-1}_y P^{t}_x &=& \alpha_{typx} \\
(P^t_x)^2 &=& \alpha_{px}  \\
T^t_x P^t_x T^t_x P^t_x &=& \alpha_{txpx} \\
({\cal T}^t)^2 &=& \alpha_T \\
{\cal T}^t T^t_x &=&  \alpha_{Ttx}  T^t_x {\cal T}^t\\
{\cal T}^t T^t_y &=&  \alpha_{Tty} T^t_y {\cal T}^t \\
{\cal T}^t P^t_x &=&  \alpha_{Tpx} P^t_x {\cal T}^t \text{,}
\end{eqnarray}
where the $\alpha$'s take values in ${\rm U}(1)$.

Following the analysis of the case of $pm$ symmetry without time reversal (Sec.~\ref{subsec:pm}), we adjust the phase of $T^t_y$ to set $\alpha_{typx} = 1$ (this does not affect $\alpha_{Tty}$), and we can restrict $\alpha_{px}, \alpha_{txpx} \in \zz$.  Next, we can set $\alpha_T = 1$ by adjusting the phase of ${\cal T}^t$.  Making this adjustment modifies $\alpha_{Tpx} \to \alpha^{-1}_T \alpha_{Tpx} \equiv \alpha'_{Tpx}$, without changing other parameters. While this can be absorbed as a redefinition of $\alpha_{Tpx}$, we will keep track of it explicitly, as this is important to work out the map $\rho_2$.  Next, we can conjugate the last three relations by ${\cal T}$, which gives $\alpha_{Ttx}, \alpha_{Tty}, \alpha'_{Tpx} \in \zz$.  Therefore
\begin{eqnarray}
T^t_x T^t_y T^{t-1}_x T^{t-1}_y &=& \alpha_{txty}  \\
T^t_y P^t_x T^{t-1}_y P^{t}_x &=& 1 \\
(P^t_x)^2 &=& \alpha_{px} \in \zz   \\
T^t_x P^t_x T^t_x P^t_x &=& \alpha_{txpx} \in \zz \\
({\cal T}^t)^2 &=& 1 \\
{\cal T}^t T^t_x  &=&  [\alpha_{Ttx} \in \zz] T^t_x {\cal T}^t \\
{\cal T}^t T^t_y &=&  [\alpha_{Tty} \in \zz] T^t_y {\cal T}^t \\
{\cal T}^t P^t_x  &=& [\alpha'_{Tpx}  \in \zz] P^t_x {\cal T}^t \text{.}
\end{eqnarray}
Note that we have not used the freedom to adjust phases of $T_x$ or $P_x$.  However, adjusting these phases has no effect on the $\alpha$'s.  This suggests the result
\begin{equation}
H^2_t(G', {\rm U}(1)) = {\rm U}(1) \times (\zz)^5 \text{,}
\end{equation}
with $[\omega]_{{\rm U}(1)} \in H^2_t(G', {\rm U}(1))$ parametrized by $[\omega]_{{\rm U}(1)} = (\alpha_{txty}, \alpha_{px}, \alpha_{txpx}, \alpha_{Ttx}, \alpha_{Tty}, \alpha'_{Tpx})$, with the first entry a ${\rm U}(1)$ phase and the last five $\zz$ phases.

To confirm this result, we need to show that each element is actually realized by some $t$-twisted ${\rm U}(1)$ factor set.  We introduce two-component field operators $v_r$ as for $G = {\rm U}(1) \times pm$ symmetry in Sec.~\ref{subsec:pm}.  We choose $T_x$, $T_y$ and $P_x$ to act on the $v_r$ as in Eqs.~(\ref{eqn:pm-txaction}-\ref{eqn:pm-pxaction}), and ${\cal T}$ acts by
\begin{equation}
{\cal T} v_r {\cal T}^{-1} = g_T v^\dagger_{r} \text{.}
\end{equation}
Here, $g_T$ is a $2 \times 2$ unitary matrix satisfying $g_T^2 = 1$, so that ${\cal T}^2 = 1$ acting on $v_r$.
We find six families of representations, whose factor sets form a generating set for $H^2_t(G', {\rm U}(1)) = {\rm U}(1) \times (\zz)^5$:
\begin{enumerate}
\item  $g_{tx} = g_{ty} = g_{px} = g_T = 1$ gives a continuous family of representations with $[\omega]_{{\rm U}(1)} = (\alpha_{txty},1,1,1,1,1)$.

\item  $\alpha_{txty} = 1$, $g_{ty} = i$, $g_{px} = i \sigma^y$, $g_{tx} = \sigma^z$, $g_T = 1$ is a representation with $[\omega]_{{\rm U}(1)} = (1,-1,1,1,1,1)$.  

\item $\alpha_{txty} = 1$, $g_{ty} = 1$, $g_{px} = \sigma^x$, $g_{tx} = \sigma^z$, $g_T = 1$ is a representation with $[\omega]_{{\rm U}(1)} = (1,1,-1,1,1,1)$.

\item $\alpha_{txty} = g_{ty} = g_{px} = 1$, $g_{tx} = \sigma^z$, $g_T = \sigma^x$ is a representation with $[\omega]_{{\rm U}(1)} = (1,1,1,-1,1,1)$.

\item $\alpha_{txty} = g_{tx} = g_{px} = 1$, $g_{ty} = \sigma^z$, $g_T = \sigma^x$ is a representation with $[\omega]_{{\rm U}(1)} = (1,1,1,1,-1,1)$.

\item $\alpha_{txty} = g_{tx} = g_{ty} = 1$, $g_{px} = \sigma^z$, $g_T = \sigma^x$ is a representation with $[\omega]_{{\rm U}(1)} = (1,1,1,1,1,-1)$.
\end{enumerate}

As in Sec.~\ref{subsec:pm}, the above analysis allows us to immediately determine the map $\rho_2 : H^2(G', \zz) \to H^2_t(G', {\rm U}(1))$, and we have
\begin{eqnarray}
&& (\alpha_{txty}, \alpha_{px}, \alpha_{txpx}, \alpha_{Ttx}, \alpha_{Tty}, \alpha'_{Tpx}) = \rho_2( [\omega_m]_{\zz} ) \nonumber \\ 
&=& ( \sigma^m_{txty}, \sigma^m_{px}, \sigma^m_{txpx}, \sigma^m_{Ttx}, \sigma^m_{Tty}, \sigma^m_{T} \sigma^m_{Tpx} ) \text{.}
\end{eqnarray}
This implies that, anomaly-negative fractionalization patterns are those with $\sigma^m_{px} =  \sigma^m_{txpx} = \sigma^m_{Ttx} =  \sigma^m_{Tty} =  \sigma^m_{T} \sigma^m_{Tpx} = 1$.  The group ${\cal N}$ of anomaly-negative vison fractionalization classes is thus ${\cal N} = (\zz)^3$.  The disjoint sets of SPT phases distinguished by the anomaly test are labeled by elements of ${\cal S} = H^2(G', \zz) / {\cal N} = (\zz)^5$.

\subsection{$G = {\rm U}(1) \times p4mm$}
\label{subsec:p4mm}

\begin{figure}
\includegraphics[width=0.7\columnwidth]{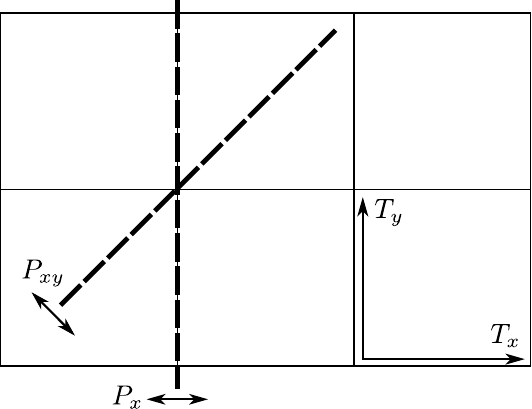}
\caption{Illustration of the operations generating the $d=2$ space group $p4mm$, the symmetry group of the square lattice.  $T_x$ and $T_y$ are translations by one lattice constant along the $x$- and $y$-axes, respectively. The vertical dashed line is the axis for the reflection $P_x$, and the diagonal dashed line is the axis for the reflection $P_{xy}$.} 
\label{fig:p4mm-operations}
\end{figure} 

The group $p4mm$ is the space group symmetry of the square lattice.  We choose generators $T_x, T_y, T_x^{-1}, T_y^{-1}, P_x$ and $P_{xy}$.  These operations are illustrated in Fig.~\ref{fig:p4mm-operations}, and obey the relations
\begin{eqnarray}
T_x T_y T^{-1}_x T^{-1}_y &=& 1 \\
T_y P_x T^{-1}_y P_x &=& 1 \\
T_y &=& P_{xy} T_x P_{xy} \\
P_x^2 &=& 1 \\
T_x P_x T_x P_x &=& 1 \\
P_{xy}^2 &=& 1 \\
(P_x P_{xy})^4 &=& 1 \text{.}
\end{eqnarray}

The $m$ particle  symmetry fractionalization is specified by
\begin{eqnarray}
T_x^m T_y^m T^{m-1}_x T^{m-1}_y &=& \sigma^m_{txty} \\
T_y^m P_x^m T^{m-1}_y P_x^m &=& \sigma^m_{typx} \\
T_y^m &=& P_{xy}^m T_x^m P_{xy}^m \label{eqn:p4mm-m-ty-relation} \\
(P_x^m)^2 &=& \sigma^m_{px} \\
T_x^m P_x^m T_x^m P_x^m &=& \sigma^m_{txpx} \\
(P_{xy}^m)^2 &=& \sigma^m_{pxy} \\
(P_x^m P_{xy}^m)^4 &=& \sigma^m_{pxpxy} \text{,}
\end{eqnarray}
where the $\sigma^m$'s take values in $\zz$.
The relation Eq.~(\ref{eqn:p4mm-m-ty-relation}) has no $\sigma^m$ parameter, 
as this can be removed by adjusting the phase $T^m_y \to - T^m_y$.  The $m$ particle fractionalization classes form the group $H^2(G', \zz) = (\zz)^6$; a generating set of projective representations verifying this result is exhibited in Appendix~A of Ref.~\onlinecite{essin13}.

To compute the $H^2_t(G', {\rm U}(1))$ cohomology, we begin by specifying the relations
\begin{eqnarray}
T_x^t T_y^t T^{t-1}_x T^{t-1}_y &=& \alpha_{txty} \\
T_y^t P_x^t T^{t-1}_y P_x^t &=& \alpha_{typx} \\
T_y^t &=& P_{xy}^t T_x^t P_{xy}^t \label{eqn:p4mm-t-ty-relation}  \\
(P_x^t)^2 &=& \alpha_{px} \label{eqn:p4mm-t-px-relation} \\
T_x^t P_x^t T_x^t P_x^t &=& \label{eqn:p4mm-t-txpx-relation} \alpha_{txpx} \\
(P_{xy}^t)^2 &=& \alpha_{pxy} \label{eqn:p4mm-t-pxy-relation} \\
(P_x^t P_{xy}^t)^4 &=& \alpha_{pxpxy} \text{,}
\end{eqnarray}
where the $\alpha$'s take values in ${\rm U}(1)$.
We note that $t(T_x) = t(T_y) = 1$, while $t(P_x) = t(P_{xy}) = -1$.

First, we adjust the phase of $T^t_y$ to set $\alpha_{typx} \to 1$.  In order to leave Eq.~(\ref{eqn:p4mm-t-ty-relation}) unchanged, we must also correspondingly adjust the phase of $T^t_x$.  Next, we adjust the phase of $P^t_x$ to set $\alpha_{pxpxy} \to 1$, which does not affect the other relations.  Finally, conjugating Eq.~(\ref{eqn:p4mm-t-px-relation}) by $P^t_x$, Eq.~(\ref{eqn:p4mm-t-txpx-relation}) by $T^t_x P^t_x$, and Eq.~(\ref{eqn:p4mm-t-pxy-relation}) by $P^t_{xy}$, we have $\alpha_{px}, \alpha_{txpx}, \alpha_{pxy} \in \zz$.  The relations thus take the form
\begin{eqnarray}
T_x^t T_y^t T^{t-1}_x T^{t-1}_y &=& \alpha_{txty} \label{eqn:p4mm-t-txty-relation2} \\
T_y^t P_x^t T^{t-1}_y P_x^t &=& 1 \\
T_y^t &=& P_{xy}^t T_x^t P_{xy}^t  \label{eqn:p4mm-t-ty-relation2}  \\
(P_x^t)^2 &=& \alpha_{px} \in \zz \\
T_x^t P_x^t T_x^t P_x^t &=& \alpha_{txpx} \in \zz \\
(P_{xy}^t)^2 &=& \alpha_{pxy} \in \zz \\
(P_x^t P_{xy}^t)^4 &=& 1  \text{.}
\end{eqnarray}
This suggests that $H^2_t(G', {\rm U}(1) ) = {\rm U}(1) \times (\zz)^3$, with $[\omega]_{{\rm U}(1)} \in H^2_t(G', {\rm U}(1) )$ parametrized by $[\omega]_{{\rm U}(1)} = (\alpha_{txty}, \alpha_{px}, \alpha_{txpx}, \alpha_{pxy})$.

To verify this, we proceed as in the case $G = {\rm U}(1) \times pm$ in Sec.~\ref{subsec:pm}, and introduce two-component field operators $v_r$, with $r$ labeling the sites of the square lattice.  The generators act on the field operators by
\begin{eqnarray}
T_x v_r T^{-1}_x &=&  (\alpha_{txty})^{r_y/2} g_{tx} v_{r + \hat{x}}  \label{eqn:p4mm-txaction} \\
P_x v_r P^{-1}_x &=& g_{px} v^\dagger_{P_x r}  \label{eqn:p4mm-pxaction} \\
P_{xy} v_r P^{-1}_{xy} &=& g_{pxy} v^\dagger_{P_{xy} r }  \label{eqn:p4mm-pxyaction} \text{,}
\end{eqnarray}
where $\alpha_{txty} \in {\rm U}(1)$, $P_x r = (-x,y)$, $P_{xy} r = (y,x)$, and $g_{tx}, g_{px}, g_{pxy}$ are $2 \times 2$ unitary matrices.  The action of $T_y$ follows from Eq.~(\ref{eqn:p4mm-t-ty-relation2}) and is
\begin{equation}
T_y v_r T^{-1}_y = (\alpha_{txty})^{-r_x/2} g_{ty} v_{r + \hat{y} } \text{,} \label{eqn:p4mm-tyaction}
\end{equation}
where $g_{ty} = g_{pxy} g^*_{tx} g^*_{pxy}$.

The following families of projective representations form a generating set for $H^2_t(G', {\rm U}(1) )$:
\begin{enumerate}
\item $g_{tx} = g_{px} = g_{pxy} = 1$ gives a continuous family of representations with $[\omega]_{{\rm U}(1)} = (\alpha_{txty},1,1,1)$. 

\item $\alpha_{txty} = g_{pxy} = 1$, $g_{tx} = \sigma^z$, $g_{px} = i \sigma^y$ is a projective representation with $[\omega]_{{\rm U}(1)} = (1,-1,1,1)$.

\item $\alpha_{txty} = g_{pxy} = 1$, $g_{tx} = i \sigma^z$, $g_{px} = \sigma^x$ is a projective representation with $[\omega]_{{\rm U}(1)} = (1,1,-1,1)$.

\item $\alpha_{txty} = g_{tx} = g_{px} = 1$, $g_{pxy} = i \sigma^y$ is a projective representation with $[\omega]_{{\rm U}(1)} = (1,1,1,-1)$.
\end{enumerate}

Finally, the map $\rho_2$ is given by
\begin{eqnarray}
&& (\alpha_{txty}, \alpha_{px}, \alpha_{txpx}, \alpha_{pxy}) = \rho_2( [\omega_m]_{\zz} ) \nonumber \\
&=& (\sigma^m_{txty}, \sigma^m_{px}, \sigma^m_{txpx}, \sigma^m_{pxy} ) \text{.}
\end{eqnarray}
Therefore, the anomaly-negative fractionalization patterns are those with $\sigma^m_{px} = \sigma^m_{txpx} = \sigma^m_{pxy} = 1$.  The group ${\cal N}$ of anomaly-negative vison fractionalization classes is  ${\cal N} = (\zz)^3$.   The disjoint sets of SPT phases distinguished by the anomaly test are labeled by elements of ${\cal S} = H^2(G', \zz) / {\cal N} = (\zz)^3$.

We remark that in this case, all the anomalous fractionalization patterns we find can be understood in terms of the symmetry ${\rm U}(1) \times \zz^P$, by choosing different $\zz^P$ subgroups of $p4mm$.

\subsection{$G = {(\rm U}(1) \rtimes \zz^T) \times p4mm$}
\label{subsec:p4mm-z2T}

This is closely related to the case $G = {\rm U}(1) \times p4mm$, but now with time reversal symmetry added.  The $\zz^T$ time reversal forms a semidirect product with ${\rm U}(1)$.  The generators are as in Appendix~\ref{subsec:p4mm}, with the addition of the time reversal operation ${\cal T}$, and we have the relations
\begin{eqnarray}
T_x T_y T^{-1}_x T^{-1}_y &=& 1 \\
T_y P_x T^{-1}_y P_x &=& 1 \\
T_y &=& P_{xy} T_x P_{xy} \\
P_x^2 &=& 1 \\
T_x P_x T_x P_x &=& 1 \\
P_{xy}^2 &=& 1 \\
(P_x P_{xy})^4 &=& 1  \\
{\cal T}^2 &=& 1 \\
{\cal T} T_x &=& T_x {\cal T} \\
{\cal T} P_x &=& P_x {\cal T} \\
{\cal T} P_{xy} &=& P_{xy} {\cal T} \text{.}
\end{eqnarray}
The $m$ particle  symmetry fractionalization is specified by
\begin{eqnarray}
T_x^m T_y^m T^{m-1}_x T^{m-1}_y &=& \sigma^m_{txty} \\
T_y^m P_x^m T^{m-1}_y P_x^m &=& \sigma^m_{typx} \\
T_y^m &=& P_{xy}^m T_x^m P_{xy}^m  \\
(P_x^m)^2 &=& \sigma^m_{px} \\
T_x^m P_x^m T_x^m P_x^m &=& \sigma^m_{txpx} \\
(P_{xy}^m)^2 &=& \sigma^m_{pxy} \\
(P_x^m P_{xy}^m)^4 &=& \sigma^m_{pxpxy} \\
({\cal T}^m)^2 &=& \sigma^m_T \\
{\cal T}^m T_x^m &=&  \sigma^m_{Ttx} T_x^m {\cal T}^m \\
{\cal T}^m P_x^m &=&  \sigma^m_{Tpx} P_x^m {\cal T}^m \\
{\cal T}^m P_{xy}^m &=&  \sigma^m_{Tpxy} P_{xy}^m {\cal T}^m \text{,}
\end{eqnarray}
where the $\sigma^m$'s take values in $\zz$.
The $m$ particle fractionalization classes form the group $H^2(G', \zz) = (\zz)^{10}$; a generating set of projective representations verifying this result is exhibited in Appendix~A of Ref.~\onlinecite{essin13}.

To compute the $H^2_t(G', {\rm U}(1))$ cohomology, we begin by specifying the relations
\begin{eqnarray}
T_x^t T_y^t T^{t-1}_x T^{t-1}_y &=& \alpha_{txty} \\
T_y^t P_x^t T^{t-1}_y P_x^t &=& \alpha_{typx} \\
T_y^t &=& P_{xy}^t T_x^t P_{xy}^t \label{eqn:p4mmT-t-ty-relation}  \\
(P_x^t)^2 &=& \alpha_{px} \label{eqn:p4mmT-t-px-relation} \\
T_x^t P_x^t T_x^t P_x^t &=& \label{eqn:p4mmT-t-txpx-relation} \alpha_{txpx} \\
(P_{xy}^t)^2 &=& \alpha_{pxy} \label{eqn:p4mmT-t-pxy-relation} \\
(P_x^t P_{xy}^t)^4 &=& \alpha_{pxpxy} \\
({\cal T}^t)^2 &=& \alpha_T \\
{\cal T}^t T_x^t &=&  \alpha_{Ttx} T_x^t {\cal T}^t \\
{\cal T}^t P_x^t &=&  \alpha_{Tpx} P_x^t {\cal T}^t \\
{\cal T}^t P_{xy}^t &=&  \alpha_{Tpxy} P_{xy}^t {\cal T}^t \text{,}
\end{eqnarray}
where the $\alpha$'s take values in ${\rm U}(1)$.
Here, $t({\cal T}) = 1$, and $t$ is specified for the other generators in Appendix~\ref{subsec:p4mm}.

Proceeding first as in Appendix~\ref{subsec:p4mm}, we adjust the phase of $T^t_y$ to set $\alpha_{typx} \to 1$.  In order to leave Eq.~(\ref{eqn:p4mmT-t-ty-relation}) unchanged, we must also correspondingly adjust the phase of $T^t_x$.  Next, we adjust the phase of $P^t_x$ to set $\alpha_{pxpxy} \to 1$, which does not affect the other relations.  We also adjust the phase of ${\cal T}^t$ to set $\alpha_T \to 1$.  This modifies $\alpha_{Tpx} \to \alpha'_{Tpx} = \alpha^{-1}_T \alpha_{Tpx}$ and $\alpha_{Tpxy} \to \alpha'_{Tpxy} = \alpha^{-1}_T \alpha_{Tpxy}$. Conjugating Eq.~(\ref{eqn:p4mmT-t-px-relation}) by $P^t_x$, Eq.~(\ref{eqn:p4mmT-t-txpx-relation}) by $T^t_x P^t_x$, and Eq.~(\ref{eqn:p4mmT-t-pxy-relation}) by $P^t_{xy}$, we have $\alpha_{px}, \alpha_{txpx}, \alpha_{pxy} \in \zz$.  Finally, conjugating the last three relations by ${\cal T}^t$ gives $\alpha_{Ttx}, \alpha'_{Tpx}, \alpha'_{Tpxy} \in \zz$.

The relations thus take the form
\begin{eqnarray}
T_x^t T_y^t T^{t-1}_x T^{t-1}_y &=& \alpha_{txty} \label{eqn:p4mmT-t-txty-relation2} \\
T_y^t P_x^t T^{t-1}_y P_x^t &=& 1 \\
T_y^t &=& P_{xy}^t T_x^t P_{xy}^t \label{eqn:p4mmT-t-ty-relation2}  \\
(P_x^t)^2 &=& \alpha_{px} \in \zz \label{eqn:p4mmT-t-px-relation2} \\
T_x^t P_x^t T_x^t P_x^t &=& \alpha_{txpx} \in \zz \label{eqn:p4mmT-t-txpx-relation2}  \\
(P_{xy}^t)^2 &=& \alpha_{pxy} \in \zz \label{eqn:p4mmT-t-pxy-relation2} \\
(P_x^t P_{xy}^t)^4 &=& 1 \\
({\cal T}^t)^2 &=& 1 \\
{\cal T}^t T_x^t &=&  [ \alpha_{Ttx}  \in \zz ]T_x^t {\cal T}^t \\
{\cal T}^t P_x^t &=&  [ \alpha'_{Tpx} \in \zz] P_x^t {\cal T}^t \\
{\cal T}^t P_{xy}^t &=&  [ \alpha'_{Tpxy} \in \zz] P_{xy}^t {\cal T}^t \text{.}
\end{eqnarray}
This suggests that $H^2_t(G', {\rm U}(1) ) = {\rm U}(1) \times (\zz)^6$, with $[\omega]_{{\rm U}(1)} \in H^2_t(G', {\rm U}(1) )$ parametrized by $[\omega]_{{\rm U}(1)} = (\alpha_{txty}, \alpha_{px}, \alpha_{txpx}, \alpha_{pxy},\alpha_{Ttx},\alpha'_{Tpx},\alpha'_{Tpxy})$.

To verify this, we introduce field operators $v_r$ as in Appendix~\ref{subsec:p4mm}, for the case of ${\rm U}(1) \times p4mm$ symmetry.  The action of $T_x$, $P_x$, $P_{xy}$ and $T_y$ is given by Eqs.~(\ref{eqn:p4mm-txaction}-\ref{eqn:p4mm-tyaction}).  Time reversal acts by
\begin{equation}
{\cal T} v_r {\cal T}^{-1} = g_T v^\dagger_r \text{,}
\end{equation}
where $g_T$ is a $2 \times 2$ unitary matrix, satisfying $g_T^2 =1$ so that $\alpha_T = 1$.

The following families of projective representations form a generating set for $H^2_t(G', {\rm U}(1) )$:
\begin{enumerate}
\item $g_{tx} = g_{px} = g_{pxy} = g_T = 1$ gives a continuous family of representations with $[\omega]_{{\rm U}(1)} = (\alpha_{txty},1,1,1,1,1,1)$. 

\item $\alpha_{txty} = g_{pxy} = g_T= 1$, $g_{tx} = \sigma^z$, $g_{px} = i \sigma^y$ is a projective representation with $[\omega]_{{\rm U}(1)} = (1,-1,1,1,1,1,1)$.

\item $\alpha_{txty} = g_{pxy} = g_T = 1$, $g_{tx} = i \sigma^z$, $g_{px} = \sigma^x$ is a projective representation with $[\omega]_{{\rm U}(1)} = (1,1,-1,1,1,1,1)$.

\item $\alpha_{txty} = g_{tx} = g_{px} =  g_T = 1$, $g_{pxy} = i \sigma^y$ is a projective representation with $[\omega]_{{\rm U}(1)} = (1,1,1,-1,1,1,1)$.

\item $\alpha_{txty} = g_{px} = g_{pxy} = 1$, $g_{tx} = \sigma^z$, $g_T = \sigma^x$ is a projective representation with $[\omega]_{{\rm U}(1)} = (1,1,1,1,-1,1,1)$.

\item $\alpha_{txty} = g_{tx} = g_{pxy} = 1$, $g_{px} = \sigma^z$, $g_T = \sigma^x$ is a projective representation with $[\omega]_{{\rm U}(1)} = (1,1,1,1,1,-1,1)$.

\item $\alpha_{txty} = g_{tx} = g_{px} = 1$, $g_{pxy} = \sigma^z$, $g_T = \sigma^x$ is a projective representation with $[\omega]_{{\rm U}(1)} = (1,1,1,1,1,1,-1)$.
\end{enumerate}

Finally, the map $\rho_2$ is given by
\begin{eqnarray}
&& (\alpha_{txty}, \alpha_{px}, \alpha_{txpx}, \alpha_{pxy},\alpha_{Ttx},\alpha'_{Tpx},\alpha'_{Tpxy}) = \rho_2( \{ \sigma^m \text{'s} \} ) \nonumber \\
&=& (\sigma^m_{txty}, \sigma^m_{px}, \sigma^m_{txpx}, \sigma^m_{pxy},\sigma^m_{Ttx}, \sigma^m_T \sigma^m_{Tpx}, \sigma^m_T \sigma^m_{Tpxy} ) \text{.}
\end{eqnarray}
Therefore, the anomaly-negative fractionalization patterns are those with $\sigma^m_{px} = \sigma^m_{txpx} = \sigma^m_{pxy} = \sigma^m_{Ttx} = \sigma^m_T \sigma^m_{Tpx} = \sigma^m_T \sigma^m_{Tpxy} = 1$.  The group ${\cal N}$ of anomaly-negative vison fractionalization classes is ${\cal N} = (\zz)^4$.   The disjoint sets of SPT phases distinguished by the anomaly test are labeled by elements of ${\cal S} = H^2(G', \zz) / {\cal N} = (\zz)^6$.

\bibliography{btci}

\end{document}